\documentclass{article}
\usepackage{ifthen}
\usepackage{mdwlist}
\usepackage{amsmath,amssymb,amsfonts,amsthm}
\usepackage{bm}
\usepackage{hyperref}
\usepackage{subfigure}
\usepackage{caption}
\usepackage{graphicx}
\usepackage{xspace}
\usepackage{verbatim}
\usepackage{algorithm}
\usepackage{algpseudocode}
\usepackage[margin=1in]{geometry}
\usepackage{color}
\usepackage{thmtools}
\usepackage{thm-restate}
\usepackage{latexsym}
\usepackage{epsfig}
\usepackage{todonotes}

\def\eps{\varepsilon}
\renewcommand{\epsilon}{\ve}
\def\ve{\varepsilon}

\def\to{\rightarrow}

\newcommand{\dtv}{d_{\mathrm {TV}}}

\newtheorem{theorem}{Theorem}

\newtheorem{observation}{Observation}
\newtheorem{remark}{Remark}
\newtheorem{fact}{Fact}
\newtheorem{lemma}{Lemma}
\newtheorem{claim}{Claim}
\newtheorem{corollary}{Corollary}

\newtheorem{definition}{Definition}
\newtheorem{question}{Question}

\definecolor{Red}{rgb}{1,0,0}

\newcommand{\oldbound}[1]{{}}

\newcommand{\ab}{k}

\newcommand{\ns}{n}

\newcommand{\p}{p}
\newcommand{\q}{q}

\newcommand{\dst}{\alpha}
\newcommand{\fp}{\beta}
\newcommand{\Xon}{X_1^{\ns}}
\newcommand{\Yon}{Y_1^{\ns}}

\newcommand{\emp}{\hat{\p}_\ns}

\newcommand{\prfi}[1]{\varphi_{#1}}
\newcommand{\entp}[1]{H(#1)}
\newcommand{\entemp}{H(\emp)}
\newcommand{\dist}{\alpha}

\newcommand{\smb}{x}
\newcommand{\ps}{\p(\smb)}

\newcommand{\nsmb}{N_{\smb}}

\newcommand{\pemps}{\emp(\smb)}

\newcommand{\kldist}[2]{D(#1\lVert #2)}

\newcommand{\chisquarerestr}[3][]{{\operatorname{d}^{#1}_{\chi^2}\!\left({#2 \mid\mid #3}\right)}}
\newcommand{\chisquare}[2]{\chisquarerestr[]{#1}{#2}}

\newcommand{\Dk}{\Delta_{\ab}}
\newcommand{\Dgk}{\Delta_{\ge\frac1\ab}}
\newcommand{\ham}[2]{d_{ham}(#1,#2)}

\newcommand{\MM}{\texttt{MM}}

\newcommand{\plugin}{\texttt{plug-in}}

\newcommand{\polyn}{\texttt{poly}}

\usepackage[shortlabels]{enumitem}
\setitemize{noitemsep,topsep=0pt,parsep=0pt,partopsep=0pt}
\setenumerate{noitemsep,topsep=0pt,parsep=0pt,partopsep=0pt}

\makeatletter
\@ifundefined{theorem}{
  \theoremstyle{definition}
  \newtheorem{definition}{Definition}
  \theoremstyle{plain}

  \newtheorem{lemma}{Lemma}

  \theoremstyle{remark}

}{}
\makeatother

\newcommand{\ignore}[1]{}

\newcommand{\EE}{\mathbb{E}}

\newcommand{\RR}{\mathbb{R}}

\newcommand{\expectation}[1]{\EE\left[#1\right]}
\newcommand{\variance}[1]{Var\left(#1\right)}

\def \cX     {{\cal X}}

\newcommand{\Var}{{\rm Var}}

\newcommand{\absv}[1]{\left|#1\right|}

\def \Paren#1{{\left({#1}\right)}}

\newcommand{\ed}{\stackrel{\text{def}}{=}}

\newcommand{\probof}[1]{\Pr\Paren{#1}}

\def\ignore#1{}

\def \ve {{\lor}} 

\newcommand{\bi}{\begin{itemize}}
\newcommand{\ei}{\end{itemize}}

\def\orpro{\mathop{\mathchoice
   {\vee\kern-.49em\raise.7ex\hbox{$\cdot$}\kern.4em}
   {\vee\kern-.45em\raise.63ex\hbox{$\cdot$}\kern.2em}
   {\vee\kern-.4em\raise.3ex\hbox{$\cdot$}\kern.1em}
   {\vee\kern-.35em\raise2.2ex\hbox{$\cdot$}\kern.1em}}\limits}

\def\andpro{\mathop{\mathchoice
 {\wedge\kern-.46em\lower.69ex\hbox{$\cdot$}\kern.3em}
 {\wedge\kern-.46em\lower.58ex\hbox{$\cdot$}\kern.25em}
 {\wedge\kern-.38em\lower.5ex\hbox{$\cdot$}\kern.1em}
 {\wedge\kern-.3em\lower.5ex\hbox{$\cdot$}\kern.1em}}\limits}

\def\simge{\mathrel{%
   \rlap{\raise 0.511ex \hbox{$>$}}{\lower 0.511ex \hbox{$\sim$}}}}

\def\simle{\mathrel{
   \rlap{\raise 0.511ex \hbox{$<$}}{\lower 0.511ex \hbox{$\sim$}}}}

\title{INSPECTRE: Privately Estimating the Unseen}

\author {
Jayadev Acharya\thanks{Supported by NSF CCF-1657471 and a Cornell University startup grant.} \\
ECE, Cornell University\\
\tt{acharya@cornell.edu}
\and
Gautam Kamath\thanks{Supported by ONR N00014-12-1-0999, NSF CCF-1617730, CCF-1650733, and CCF-1741137. Work partially done while author was an intern at Microsoft Research, New England.} \\
EECS \& CSAIL, MIT\\
\tt{g@csail.mit.edu}
\and
Ziteng Sun\thanks{Supported by NSF CCF-1657471 and a Cornell University startup grant.} \\
ECE, Cornell University\\
\tt{zs335@cornell.edu}
\and
Huanyu Zhang\thanks{Supported by NSF CCF-1657471 and a Cornell University startup grant.} \\
ECE, Cornell University\\
\tt{hz388@cornell.edu}
}

\begin{document}
\maketitle
\begin{abstract}
We develop differentially private methods for estimating various distributional properties.
Given a sample from a discrete distribution $p$, some functional $f$, and accuracy and privacy parameters $\dist$ and $\eps$, the goal is to estimate $f(p)$ up to accuracy $\dist$, while maintaining $\eps$-differential privacy of the sample.

We prove almost-tight bounds on the sample size required for this problem for several functionals of interest, including support size, support coverage, and entropy.
We show that the cost of privacy is negligible in a variety of settings, both theoretically and experimentally.
Our methods are based on a sensitivity analysis of several state-of-the-art methods for estimating these properties with sublinear sample complexities.

\end{abstract}

\section{Introduction}

How can we infer a distribution given a sample from it?
If data is in abundance, the solution may be simple -- the empirical distribution will approximate the true distribution.
 However, challenges arise when data is scarce in comparison to the size of the domain, and especially when we wish to quantify ``rare events.''
This is frequently the case: for example, it has recently been observed that there are several very rare genetic mutations which occur in humans, and we wish to know how many  such mutations exist~\cite{KeinanC12,TennessenBOFKGMDLJKJLGRAANBSBBABSN12,NelsonWEKSVSTBFWAZLZZLLLWTHNWACZWCNM12}.
Many of these mutations have only been seen once, and we can infer that there are many which have not been seen at all.
Over the last decade, a large body of work has focused on developing theoretically sound and effective tools for such settings~\cite{OrlitskySW16} and references therein, including the problem of estimating the frequency distribution of rare genetic variations~\cite{ZouVVKCSLSDM16}.

However, in many settings where one wishes to perform statistical inference, data may contain sensitive information about individuals.
For example, in medical studies, where the data may contain individuals' health records and whether they carry some disease which bears a social stigma.
Alternatively, one can consider a map application which suggests routes based on aggregate positions of individuals, which contains delicate information including users' residence data.
In these settings, it is critical that our methods protect sensitive information contained in the dataset.  
This does not preclude our overall goals of statistical analysis, as we are trying to infer properties of the population $p$, and not the samples which are drawn from said population.

That said, without careful experimental design, published statistical findings may be prone to leaking sensitive information about the sample.
As a notable example, it was recently shown that one can determine the identity of some individuals who participated in genome-wide association studies~\cite{HomerSRDTMPSNC08}.
This realization has motivated a surge of interest in developing data sharing techniques with an explicit focus on maintaining privacy of the data~\cite{JohnsonS13,UhlerSF13,YuFSU14,SimmonsSB16}. 

Privacy-preserving computation has enjoyed significant study in a number of fields, including statistics and almost every branch of computer science, including cryptography, machine learning, algorithms, and database theory -- see, e.g.,~\cite{Dalenius77,AdamW89,AgrawalA01,DinurN03,Dwork08,DworkR14} and references therein.
Perhaps the most celebrated notion of privacy, proposed by theoretical computer scientists, is \emph{differential privacy}~\cite{DworkMNS06}.
Informally, an algorithm is differentially private if its outputs on neighboring datasets (differing in a single element) are statistically close (for a more precise definition, see Section~\ref{sec:prelim}).
Differential privacy has become the standard for theoretically-sound data privacy, leading to its adoption by several large technology companies, including Google and Apple~\cite{ErlingssonPK14, AppleDP17}.

Our focus in this paper is to develop tools for privately performing several distribution property estimation tasks.
In particular, we study the tradeoff between statistical accuracy, privacy, and error rate in the sample size. 
Our model is that we are given sample access to some unknown discrete distribution $p$, over a domain of size $\ab$, which is possibly unknown in some tasks.
We wish to estimate the following properties:
\begin{itemize}
\item {\bf Support Coverage}: If we take $m$ samples from the distribution, what is the expected number of unique elements we expect to see?
\item {\bf Support Size}: How many elements of the support have non-zero probability?
\item {\bf Entropy}: What is the Shannon entropy of the distribution?
\end{itemize}
For more formal statements of these problems, see Section~\ref{sec:probs}.
We require that our output is $\dist$-accurate, satisfies $(\eps, 0)$-differential privacy, and is correct with probability $1 - \fp$.
The goal is to give an algorithm with minimal sample complexity $n$, while simultaneously being computationally efficient. 

{\bf Theoretical Results.} Our main results show that privacy can be achieved for all these problems at a very low cost.
For example, if one wishes to privately estimate entropy, this incurs an additional additive cost in the sample complexity which is very close to linear in $1/\dist\eps$.
We draw attention to two features of this bound.
First, this is independent of $\ab$.
All the problems we consider have complexity $\Theta(\ab/\log \ab)$, so in the primary regime of study where $\ab \gg 1/\dist\eps$, this small additive cost is dwarfed by the inherent sample complexity of the non-private problem.
Second, the bound is almost linear in $1/\dist\eps$.
We note that performing even the most basic statistical task privately, estimating the bias of a coin, incurs this linear dependence.
Surprisingly, we show that much more sophisticated inference tasks can be privatized at almost no cost.
In particular, these properties imply that the additive cost of privacy is $o(1)$ in the most studied regime where the support size is large.
In general, this is not true -- for many other problems, including distribution estimation and hypothesis testing, the additional cost of privacy depends significantly on the support size or dimension~\cite{DiakonikolasHS15,CaiDK17,AcharyaSZ17,AliakbarpourDR17}.
We also provide lower bounds, showing that our upper bounds are almost tight.
A more formal statement of our results appears in Section~\ref{sec:results}.

{\bf Experimental Results.} We demonstrate the efficacy of our method with experimental evaluations.
As a baseline, we compare with the non-private algorithms of~\cite{OrlitskySW16} and~\cite{WuY18}.
Overall, we find that our algorithms' performance is nearly identical, showing that, in many cases, privacy comes (essentially) for free. 
We begin with an evaluation on synthetic data.
Then, inspired by~\cite{ValiantV13,OrlitskySW16}, we analyze text corpus consisting of words from Hamlet, in order to estimate the number of unique words which occur.
Finally, we investigate name frequencies in the US census data.
This setting has been previously considered by~\cite{OrlitskySW16}, but we emphasize that this is an application where private statistical analysis is critical.
This is proven by efforts of the US Census Bureau to incorporate differential privacy into the 2020 US census~\cite{DajaniLSKRMGDGKKLSSVA17}.

{\bf Techniques.} Our approach works by choosing statistics for these tasks which possess bounded sensitivity, which is well-known to imply privacy under the Laplace or Gaussian mechanism.
We note that bounded sensitivity of statistics is not always something that can be taken for granted.
Indeed, for many fundamental tasks, optimal algorithms for the non-private setting may be highly sensitive, thus necessitating crucial modifications to obtain differential privacy~\cite{AcharyaDK15, CaiDK17}.
Thus, careful choice and design of statistics must be a priority when performing inference with privacy considerations.

To this end, we leverage recent results of~\cite{AcharyaDOS17}, which studies estimators for non-private versions of the problems we consider.
The main technical work in their paper exploits bounded sensitivity to show sharp cutoff-style concentration bounds for certain estimators, which operate using the principle of best-polynomial approximation.
They use these results to show that a single algorithm, the Profile Maximum Likelihood (PML), can estimate all these properties simultaneously.
On the other hand, we consider the sensitivity of these estimators for purposes of privacy -- the same property is utilized by both works for very different purposes, a connection which may be of independent interest.

We note that bounded sensitivity of a statistic may be exploited for purposes other than privacy.
For instance, by McDiarmid's inequality, any such statistic also enjoys very sharp concentration of measure, implying that one can boost the success probability of the test at an additive cost which is logarithmic in the inverse of the failure probability.
One may naturally conjecture that, if a statistical task is based on a primitive which concentrates in this sense, then it may also be privatized at a low cost.
However, this is not true -- estimating a discrete distribution in $\ell_1$ distance is such a task, but the cost of privatization depends significantly on the support size~\cite{DiakonikolasHS15}.

One can observe that, algorithmically, our method is quite simple: compute the non-private statistic, and add a relatively small amount of Laplace noise.
The non-private statistics have recently been demonstrated to be practical~\cite{OrlitskySW16, WuY18}, and the additional cost of the Laplace mechanism is minimal.
This is in contrast to several differentially private algorithms which invoke significant overhead in the quest for privacy.
Our algorithms attain almost-optimal rates (which are optimal up to constant factors for most parameter regimes of interest), while simultaneously operating effectively in practice, as demonstrated in our experimental results.

{\bf Related Work.} Over the last decade, there have been a flurry of works on the problems we study in this paper by the computer science and information theory communities, including Shannon and R\'enyi entropy estimation~\cite{Paninski03, ValiantV17b, JiaoVHW17, AcharyaOST17, ObremskiS17, WuY18}, support coverage and support size estimation~\cite{OrlitskySW16, WuY18}. 
A recent paper studies the general problem of estimating functionals of discrete distribution from samples in terms of the smoothness of the functional~\cite{FukuchiS17}.
These have culminated in a nearly-complete understanding of the sample complexity of these properties, with optimal sample complexities (up to constant factors) for most parameter regimes.


Recently, there has been significant interest in performing statistical tasks under differential privacy constraints.
Perhaps most relevant to this work are~\cite{CaiDK17, AcharyaSZ17, AliakbarpourDR17}, which study the sample complexity of differentialy privately performing classical distribution testing problems, including identity and closeness testing.
Other works investigating private hypothesis testing include~\cite{WangLK15,GaboardiLRV16,KiferR17,KakizakiSF17,Rogers17,GaboardiR17}, which focus less on characterizing the finite-sample guarantees of such tests, and more on understanding their asymptotic properties and applications to computing p-values.
There has also been study on private distribution learning~\cite{DiakonikolasHS15,DuchiJW17,KarwaV18}, in which we wish to estimate parameters of the distribution, rather than just a particular property of interest.
A number of other problems have been studied with privacy requirements, including clustering~\cite{WangWS15,BalcanDLMZ17}, principal component analysis~\cite{ChaudhuriSS13,KapralovT13,HardtP14b}, ordinary least squares~\cite{Sheffet17}, and much more.


\section{Preliminaries}

We will start with some definitions.

Let $\Delta\ed \{(p(1),\ldots, p(\ab)):p(i)\ge0, \sum_{i=1}^\ab p(i) = 1, 1\le \ab\le\infty\}$ be the set of discrete distributions over a countable support. Let $\Dk$ be the set of distributions in $\Delta$ with at most $\ab$ non-zero probability values. A \emph{property} $f(p)$ is a mapping from $\Delta\to\RR$. We now describe the classical distribution property estimation problem, and then state the problem under differential privacy. 

\paragraph{Property Estimation Problem.} Given $\dist,\beta$, $f$, and independent samples $\Xon$ from an unknown distribution $\p$, design an estimator $\hat f:\Xon\to\RR$ such that with probability at least $1-\beta$, $\absv{\hat{f}(\Xon)-f(\p)}<\dist$. The \emph{sample complexity} of $\hat f$, is
$
C_{\hat{f}}(f, \dist, \beta) \ed \min\{n: \probof{\absv{\hat{f}(\Xon)-f(\p)}>\dist}<\beta\}
$ is the smallest number of samples to estimate $f$ to accuracy $\dist$, and error $\beta$.  We study the problem for $\beta = 1/3$, and by the median trick, we can boost the error probability to $\beta$ with an additional multiplicative $\log (1/\beta)$ more samples: $C_{\hat f}(f,\dist) \ed C_{\hat{f}}(f, \dist, 1/3)$.
The sample complexity of estimating a property $f(\p)$ is the minimum sample complexity over all estimators: $C(f,\dist) = \min_{\hat f} C_{\hat{f}}(f, \dist)$.

An estimator $\hat{f}$ is $\eps$-differentially private (DP)~\cite{DworkMNS06} if for any $\Xon$ and $\Yon$, with $\ham{\Xon}{\Yon}\le 1$, $\frac{\probof{f(\Xon)\in S}}{\probof{f(\Yon)\in S}}\le e^\eps$, for all $i$, and measurable $S$. 

\paragraph{Private Property Estimation.} Given $\dist, \eps, \beta$, $f$, and independent samples $\Xon$ from an unknown distribution $\p$, design an $\eps$-differentially private estimator $\hat f:\Xon\to\RR$ such that with probability at least $1-\beta$, $\absv{\hat{f}(\Xon)-f(\p)}<\dist$. Similar to the non-private setting, the \emph{sample complexity} of $\eps$-differentially private estimation problem is $C(f,\dist, \eps) = \min_{\hat f: \hat f \text{is $\eps$-DP}}C_{\hat{f}}(f, \dist, 1/3)$, the smallest number of samples $\ns$ for which there exists such a $\pm\dist$ estimator with error probability at most 1/3. 

In their original paper~\cite{DworkMNS06} provides a scheme for differential privacy, known as the Laplace mechanism.
This method adds Laplace noise to a non-private scheme in order to make it private. 
We first define the sensitivity of an estimator, and then state their result in our setting. 
\begin{definition}
	\label{def:sensitivity}
	The \emph{sensitivity} of an estimator $\hat{f}:[\ab]^{\ns}\to\RR$ is 
	$\Delta_{n,\hat{f}}\ed \max_{\ham{\Xon}{\Yon}\le1} \absv{\hat{f}(\Xon)-\hat{f}(\Yon)}.$ Let $D_{\hat f}(\dist,\eps) = \min\{{\ns}: \Delta_{n,\hat{f}}\le \dist\eps\}$.\end{definition}

\label{sec:prelim}

\begin{lemma}
\label{lem:main-sensitivity}
\[
C(f,\dist, \eps)  = O\Paren{\min_{\hat f} \left\{C_{\hat f}(f,\dist/2)+ D_{\hat f}\left(\frac{\dist}{4},\eps\right)\right\}}.
\]
\end{lemma}
\begin{proof}
~\cite{DworkMNS06} showed that for a function with sensitivity $\Delta_{n,\hat{f}}$, adding Laplace noise $X \sim Lap(\Delta_{n,\hat{f}}/\eps)$ makes the output $\eps$-differentially private.
By the definition of $D_{\hat f}(\frac{\dist}{4},\eps)$, the Laplace noise we add has parameter at most $\frac{\dist}{4}$. Recall that the probability density function of $Lap(b)$ is $\frac1{2b} e^{-\frac{|x|}{b}}$, hence we have $\probof{|X| > \alpha/2} < \frac{1}{e^2}$. By the union bound, we get an additive error less than $\alpha = \frac{\alpha}{2} + \frac{\alpha}{2}$ with probability at most $1/3 + \frac{1}{e^2} < 0.5$. Hence, with the median trick, we can boost the error probability to $1/3$, at the cost of a constant factor in the number of samples.
\end{proof}

To prove sample complexity lower bounds for differentially private estimators, we observe that the estimator can be used to test between two distributions with distinct property values, hence is a harder problem. For lower bounds on differentially private testing, \cite{AcharyaSZ17} gives the following argument based on coupling:

\begin{lemma}
	\label{lem:coupling}
	Suppose there is a coupling between distributions $\p$ and $\q$ over $\cX^\ns$, such that $\expectation{\ham{\Xon}{\Yon}} \le D$. Then, any $\eps$-differentially private algorithm that distinguishes between $\p$ and $\q$ with error probability at most $1/3$ must satisfy $D = \Omega\Paren{\frac1{\eps}}$.
\end{lemma}

\subsection{Problems of Interest}
\label{sec:probs}

\paragraph{Support Size.} The support size of a distribution $\p$ is $S(\p) =\absv{\{x:\p(x)>0\}}$, the number of symbols with non-zero probability values. However, notice that estimating $S(\p)$ from samples can be hard due to the presence of symbols with negligible, yet non-zero probabilities. To circumvent this issue,~\cite{RaskhodnikovaRSS09} proposed to study the problem when the smallest probability is bounded. Let $\Dgk
  \ed\left\{\p\in\Delta:
  \p(x)\in\{0\}\cup\left[1/\ab,1\right]\right\}$
be the set of all distributions where all non-zero probabilities have value at least $1/\ab$. For $\p\in\Dgk$, our goal is to estimate $S(\p)$ up to $\pm\dist\ab$ with the least number of samples from $\p$.

\paragraph{Support Coverage.} For a distribution $\p$, and an integer $m$, let $S_m(\p) = \sum_{\smb} (1 - (1-\ps)^m)$, be the expected number of symbols that appear when we obtain $m$ independent samples from the distribution $\p$. The objective is to find the least number of samples $\ns$ in order to estimate $S_m\Paren{\p}$ to an additive $\pm \dst m$.

Support coverage arises in many ecological and biological studies~\cite{ColwellCGLMCL12} to quantify the number of \emph{new} elements (gene mutations, species, words, etc) that can be expected to be seen in the future.  
Good and Toulmin~\cite{GoodT56} proposed an estimator that  for any constant $\dst$, requires $m/2$ samples to estimate  $S_m(p)$.

\paragraph{Entropy.} 
The Shannon entropy of a distribution $\p$ is $H(p) =\sum_x p(x)\log\frac1{p(x)}$,
$H(p)$ is a central object in information theory~\cite{CoverT06}, and also arises in many fields such as machine learning~\cite{Nowozin12}, neuroscience~\cite{BerryWM97, NemenmanBRS04}, and others. Estimating $H(p)$ is hard with any finite number of samples due to the possibility of infinite support. To circumvent this,
    a natural approach is to consider distributions in $\Dk$.
 The goal is to estimate the entropy of a distribution in $\Dk$ to an additive $\pm \dist$, where $\Dk$ is all discrete distributions over at most $k$ symbols.

\section{Statement of Results}
\label{sec:results}
Our theoretical results for estimating support coverage, support size, and entropy are given below.
Algorithms for these problems and proofs of these statements are provided in Section~\ref{sec:theory}.
Our experimental results are described and discussed in Section~\ref{sec:exp}.

\begin{restatable}{theorem}{supportcoverage}
\label{thm:supportcoverage}
For any $\eps = \Omega(1/m)$, the sample complexity of support coverage is 
$$
C(S_m, \dist, \eps)= O\Paren{\frac{m\log (1/\dist)}{\log m}+ \frac{m\log (1/\dist)}{\log (\eps m)}}.
$$
Furthermore, 
$$
C(S_m, \dist, \eps)= \Omega\Paren{\frac{m\log (1/\dist)}{\log m}+ \frac{1}{\dist \eps}}.
$$
\end{restatable}

\begin{restatable}{theorem}{ssize}
\label{thm:ssize}
For any $\eps = \Omega(1/\ab)$, the sample complexity of support size estimation is 
$$
C(S, \dist, \eps)= O\Paren{\frac{\ab\log^2 (1/\dist)}{\log \ab}+ \frac{\ab\log^2 (1/\dist)}{\log (\eps \ab)}}.
$$
Furthermore,
$$
C(S, \dist, \eps)= \Omega\Paren{\frac{\ab\log^2 (1/\dist)}{\log \ab}+ \frac{1}{\dist \eps}}.
$$
\end{restatable}

\begin{restatable}{theorem}{entropy}
\label{thm:entropy}
Let $\lambda>0$ be \emph{any} small fixed constant. 
For instance, $\lambda$ can be chosen to be any constant between $0.01$ and $1$.
We have the following upper bounds on the sample complexity of entropy estimation:
$$
C(H, \dist, \eps)= O\Paren{\frac{\ab}{\dist}+\frac{\log^2(\min\{\ab,\ns\})}{\dist^2}+\frac{1}{\dist\eps}\log\Paren{\frac1{\dist\eps}}}
$$
and
$$
C(H, \dist, \eps)= O\Paren{\frac{\ab}{\lambda^2\dist\log\ab}+\frac{\log^2(\min\{\ab,\ns\})}{\dist^2} + \Paren{\frac1{\dist\eps}}^{1+\lambda}}.
$$
Furthermore,
$$
C(H, \dist, \eps) =\Omega\Paren{\frac{\ab}{\dist\log\ab}+\frac{\log^2(\min\{\ab,\ns\})}{\dist^2}+\frac{\log\ab}{\dist\eps}}.
$$
\end{restatable}

\noindent We provide some discussion of our results. 
At a high level, we wish to emphasize the following two points:
\begin{enumerate}
\item Our upper bounds show that the cost of privacy in these settings is often negligible compared to the sample complexity of the non-private statistical task, especially when we are dealing with distributions over a large support.
Furthermore, our upper bounds are almost tight in all parameters.
\item The algorithmic complexity introduced by the requirement of privacy is minimal, consisting only of a single step which noises the output of an estimator.
In other words, our methods are realizable in practice, and we demonstrate the effectiveness on several synthetic and real-data examples.
\end{enumerate}

First, we examine our results on support size and support coverage estimation.
We note that we focus on the regime where $\eps$ is not exceptionally small, as the privacy requirement becomes somewhat unusual.
For instance, non-privately, if we have $m$ samples for the problem of support coverage, then the empirical plug-in estimator is the best we can do.
However, if $\eps = O(1/m)$, then group privacy~\cite{DworkR14} implies that the algorithm's output distribution on \emph{any} dataset of $m$ samples must be very similar -- however, these samples may have an \emph{arbitrary} value of support coverage $\in [m]$, which precludes hopes for a highly accurate estimator.
To avoid degeneracies of this nature, we restrict our attention to $\eps = \Omega(1/m)$.
In this regime, if $\eps = \Omega(m^{\gamma}/m)$ for any constant $\gamma > 0$, then up to constant factors, our upper bound is within a constant factor of the optimal sample complexity without privacy constratints.
In other words, for most meaningful values of $\eps$, privacy comes for free. 

Next, we turn our attention to entropy estimation.
We note that the second upper bound in Theorem~\ref{thm:entropy} has a parameter $\lambda$ that indicates a tradeoff between the sample complexity incurred in the first and third term.
This parameter determines the degree of a polynomial to be used for entropy estimation.
As the degree becomes smaller (corresponding to a large $\lambda$), accuracy of the polynomial estimator decreases, however, at the same time, low-degree polynomials have a small sensitivity, allowing us to privatize the outcome. 

In terms of our theoretical results, one can think of $\lambda = 0.01$.
With this parameter setting, it can be observed that our upper bounds are almost tight.
For example, one can see that the upper and lower bounds match to either logarithmic factors (when looking at the first upper bound), or a very small polynomial factor in $1/\dist\eps$ (when looking at the second upper bound).
For our experimental results, we experimentally determined an effective value for the parameter $\lambda$ on a single synthetic instance.
We then show that this choice of parameter generalizes, giving highly-accurate private estimation in other instances, on both synthetic on real-world data.

\section{Algorithms and Analysis}
\label{sec:theory}
In this section, we prove our results for support coverage in Section~\ref{sec:coverage}, support size in Section~\ref{sec:ssize}, and entropy in Section~\ref{sec:entropy}. 
In each section, we first describe and analyze our algorithms for the relevant problem.
We then go on to describe and analyze a lower bound construction, showing that our upper bounds are almost tight.

All our algorithms fall into the following simple framework:
\begin{enumerate}
\item Compute a non-private estimate of the property;
\item Privatize this estimate by adding Laplace noise, where the parameter is determined through analysis of the estimator and potentially computation of the estimator's sensitivity.
\end{enumerate}

\subsection{Support Coverage Estimation}
\label{sec:coverage}
In this section, we prove Theorem~\ref{thm:supportcoverage}, about support coverage estimation: 
\supportcoverage*
Our upper bound is described and analyzed in Section~\ref{sec:coverage-ub}, while our lower bound appears in Section~\ref{sec:coverage-lb}.
\subsubsection{Upper Bound for Support Coverage Estimation}
\label{sec:coverage-ub}

Let $\prfi{i}$ be the number of symbols that appear $i$ times in $\Xon$. We will use the following non-private support coverage estimator from~\cite{OrlitskySW16}:
\[
\hat{S}_m(\Xon) = \sum_{i=1}^{\ns}\prfi{i}\Paren{1+(-t)^i\cdot \probof{Z\ge i}},
\] 
where $Z$ is a Poisson random variable with mean $r$ (which is a parameter to be instantiated later), and $t=(m-n)/n$. 

Our private estimator of support coverage is derived by adding Laplace noise to this non-private estimator with the appropriate noise parameter, and thus the performance of our private estimator, is analyzed by bounding the sensitivity and the bias of this non-private estimator according to Lemma~\ref{lem:main-sensitivity}.

The sensitivity and bias of this estimator is bounded in the following lemmas.  
\begin{lemma} \label{lem:sensitivity-coverage}
Suppose $m>2\ns$, then the maximum coefficient of $\prfi{i}$ in $\hat{S}_m(p)$ is at most $1+e^{r(t-1)}$. 
\end{lemma}

\begin{proof}
By the definition of $Z$, we know $\probof{Z\ge i} = \sum_{k = i}^{\infty} e^{-r}\frac{r^k}{k!}$, hence we have:
\begin{align}
|1+(-t)^i\cdot \probof{Z\ge i}| &\le  1+ t^i \sum_{k = i}^{\infty} e^{-r}\frac{r^k}{k!}   \nonumber \\
						&\le  1 + e^{-r} \sum_{k = i}^{\infty}\frac{(rt)^k}{k!}  \nonumber \\
					   	&\le 1 + e^{-r} \sum_{k = 0}^{\infty}\frac{(rt)^k}{k!}  \nonumber \\
						&= 1 + e^{r(t-1)} \nonumber
\end{align}
\end{proof}

\noindent The bias of the estimator is bounded in Lemma 4 of~\cite{AcharyaDOS17}:
\begin{lemma}\label{lem:bias-coverage}
Suppose $m>2\ns$, then
\[
\absv{\expectation{\hat{S}_m(\Xon)} - S_m(\p)} \le 2+2e^{r(t-1)}+\min(m, S(p))\cdot e^{-r}. 
\]
\end{lemma}

Using these results, letting $r = \log (1/\alpha)$,~\cite{OrlitskySW16} showed that there is a constant $C$, such that with $n = C\frac{m}{\log m}\log(1/\dist)$ samples, with probability at least 0.9, 
\[
\absv{\frac{\hat{S}_m(\Xon)}m- \frac{S_m(\p)}m} \le \dist.
\]



Our upper bound in Theorem~\ref{thm:supportcoverage} is derived by the following analysis of the sensitivity of $\frac{\hat{S}_m(\Xon)}m$.
	
If we change one sample in $\Xon$, at most two of the $\prfi{j}$'s change. Hence by Lemma~\ref{lem:sensitivity-coverage}, the sensitivity of the estimator satisfies
\begin{align}
	\Delta\Paren{\frac{\hat{S}_m(\Xon)}m} \le & \frac2m\cdot\Paren{1+e^{r(t-1)}}.\label{eqn:bound-sen-cov}
\end{align}

By Lemma~\ref{lem:main-sensitivity}, there is a private algorithm for support coverage estimation as long as 
\[
\Delta\Paren{\frac{\hat{S}_m(\Xon)}m} \le \dist\eps,
\]
which by~\eqref{eqn:bound-sen-cov} holds if
\begin{align}
2(1 + \exp(r(t-1))) \le \dist\eps m.
\end{align}
Let $r = \log (3/\alpha)$, note that $t-1 = \frac m{\ns} -2$. Suppose $\dist\eps m>2$, then, the condition above reduces to 
\[
\log \Paren{\frac{3}{\dist}} \cdot \Paren{\frac mn-2} \le \log \Paren{\frac12\dist\eps m -1}.
\]

This is equivalent to
\begin{align}
 n & \ge \frac{m \log (3/\dist)}{\log (\frac12\dist \eps m - 1) + 2 \log (3/\dist)}  \nonumber \\
    &= \frac{m \log (3/\dist)}{\log ( \frac32 \eps m - 3 / \dist) +  \log  (3/\dist)}  \nonumber
\end{align}

Suppose $\dist\eps m >2$, then the condition above reduces to the requirement that
\[
n = O\Paren{\frac{m\log (1/\dist)}{\log (\eps m)}}.
\]

\subsubsection{Lower Bound for Support Coverage Estimation}
\label{sec:coverage-lb}

We now prove the lower bound described in Theorem~\ref{thm:supportcoverage}.
Note that the first term in the lower bound is the sample complexity of non-private support coverage estimation, shown in~\cite{OrlitskySW16}.
Therefore, we turn our attention to prove the latter term in the sample complexity.

%
%

Consider the following two distributions. $u_1$ is uniform over $[m(1+\dist)]$. $u_2$ is distributed over $m+1$ elements $[m] \cup \{\triangle\}$ where $u_2[i] = \frac{1}{m(1+\dist)} \forall i \in [m]$ and $u_2[\triangle] = \frac{\dist}{1+\dist}$. Moreover, $\triangle \notin [m(1+\dist)]$. Then,

\[
S_m(u_1) = m(1+\dist)\cdot \Paren {1-  \Paren{1-\frac1{m(1+\dist)}}^m},
\]

and 
\begin{align}
&S_m(u_2) = \nonumber \\
&~~m\cdot \Paren {1-  \Paren{1-\frac1{m(1+\dist)}}^m} + \Paren {1-  \Paren{1-\frac\dist{1+\dist}}^m} \nonumber 
\end{align}

hence,
\begin{align}
&S_m(u_2) - S_m(u_1) \nonumber \\
& = m\dist \cdot \Paren {1-  \Paren{1-\frac1{m(1+\dist)}}^m} -  \Paren {1-  \Paren{1-\frac\dist{1+\dist}}^m} \nonumber \\
& = \Omega(\dist m) \nonumber 
\end{align}

Hence we know there support coverage differs by $\Omega(\dist m)$. Moreover, their total variation distance is $\frac{\dist}{1+\dist}$.
The following lemma is folklore, based on the coupling interpretation of total variation distance, and the fact that total variation distance is subadditive for product measures.

\begin{lemma} \label{lem:min_coupling}
For any two distributions $\p$, and $\q$, there is a coupling between $\ns$ iid samples from the two distributions with an expected Hamming distance of $\dtv(\p,\q)\cdot\ns$.  
\end{lemma}


Using Lemma~\ref{lem:min_coupling} and $\dtv(u_1,u_2) = \frac{\dist}{1+\dist}$, we have
\begin{lemma} \label{lem:coupling_coverage}
Suppose $u_1$ and $u_2$ are as defined before, there is a coupling between $u_1^\ns$ and $u_2^\ns$ with expected Hamming distance equal to $\frac{\dist}{1+\dist} \ns$.
\end{lemma}

Moreover, given $\ns$ samples, we must be able to privately distinguish between $u_1$ and $u_2$ given an $\dist$ accurate estimator of support coverage with privacy considerations. Thus, according to Lemma~\ref{lem:coupling} and~\ref{lem:coupling_coverage}, we have:
\[
\frac\dist{1+\dist} \ns\ge \frac1{\eps}\Rightarrow n = \Omega\Paren{\frac1{\eps\dist}}.
\]

%
%
%

\subsection{Support Size Estimation}
In this section, we prove our main theorem about support size estimation, Theorem~\ref{thm:ssize}:

\ssize*

\noindent Our upper bound is described and analyzed in Section~\ref{sec:ssize-ub}, while our lower bound appears in Section~\ref{sec:ssize-lb}.
\label{sec:ssize}
\subsubsection{Upper Bound for Support Size Estimation}
\label{sec:ssize-ub}



In \cite{OrlitskySW16}, it is shown that the support coverage estimator can be used to obtain optimal results for estimating the support size of a distribution.
In this fashion, taking $m=\ab\log(3/\dist)$, we we may use an estimate of the support coverage $S_m(\p)$ as an estimator of $S(\p)$.
In particular, their result is based on the following observation. 

\begin{lemma}
\label{lem:sssc}
Suppose $m\ge \ab\log (3/\dist)$, then for any $\p\in\Dgk$, 
\[
\absv{S_m(\p) - S(\p)}\le \frac{\dist\ab}3.
\]
\end{lemma}

\begin{proof}
From the definition of $S_m(\p)$, we have $S_m(\p)\le S(\p)$. For the other side, 
\begin{align}
S(p) -S_m(p) &= \sum_{\smb} \Paren{1-\p(x)}^m \le \sum_{\smb} e^{-m\p(x)}  \nonumber \\
		     &\le \ab\cdot e^{-\log (3/\dist)}  ~~~~~ = \frac{\ab\dist}{3}. \qedhere
\end{align}
\end{proof}

Therefore, estimating $S_m(\p)$ for $m = k\log (3/\dist)$, up to $\pm \dist\ab/3$, also estimates $S(p)$ up to $\pm \dist k$. Therefore, the goal is to estimate the smallest value of $\ns$ to solve the support coverage problem. 

Suppose $r = \log (3/\dist)$, and $m = \ab\log (3/\dist) = \ab \cdot r$ in the support coverage problem. Then, we have 
\begin{align}
t= \frac mn-1 = \frac{\ab\log (3/\dist)}{n}-1.\label{eq:exp-t}
\end{align}
Then, by Lemma~\ref{lem:bias-coverage} in the previous section, we have 
\begin{align}
&~~~\absv{\expectation{\hat{S}_m(\Xon)} - S(\p)} \nonumber\\ 
&\le \absv{\expectation{\hat{S}_m(\Xon)} - S_m(\p)}  + \absv{S_m(\p) - S(\p)}\nonumber\\ 
&\le 2+2e^{r(t-1)}+\min\{m, \ab\}\cdot e^{-r} + \frac{k\dist}{3}\nonumber\\
&\le 2+2e^{r(t-1)}+\ab\cdot e^{-\log (3/\dist)} + \frac{k\dist}{3}\nonumber\\
&\le  2 + 2e^{r(t-1)}+ 2\frac{k\dist}{3}.\nonumber
\end{align}

We will find conditions on $\ns$ such that the middle term above is at most $\ab\dist$. Toward this end, note that $2e^{r(t-1)} \le \dist\ab$ holds if and only if ${r(t-1)} \le \log\Paren{\frac{\dist\ab}{2}}$. Plugging in~\eqref{eq:exp-t}, this holds when
\begin{align}
\log (3/\dist)\cdot \Paren{\frac{\ab\log (3/\dist)}{n}-2} \le \log\Paren{\frac{\dist\ab}{2}},\nonumber
\end{align}
which is equivalent to 

\begin{align}
\ns\ge \frac{\ab\log^2(3/\dist)}{\log{\frac{\dist\ab}{2}}+2\log \frac 3\dist } = O\Paren{ \frac{\ab\log^2(1/\dist)}{\log{\ab }} }\nonumber
\end{align}
where we have assumed without loss of generality that $\dist>\frac1\ab$. 

The computations for sensitivity are very similar. From Lemma 1~\ref{lem:main-sensitivity}, we need to find the value of $\ns$ such that 
\[
2+2e^{r(t-1)}\le \dist\eps\ab, 
\]
where we assume that $\ns \le \frac12\ab\log (3/\dist)$, else we just add noise to the true number of observed distinct elements
By computations similar to the expectation case, this reduces to 
\begin{align}
\ns \ge \frac{\ab\log^2(3/\dist)}{\log{\frac{\dist\eps\ab}{2}}+\log \frac 3\dist}.\nonumber
\end{align}
Therefore, this gives us a sample complexity of
\begin{align}
\ns  = O\Paren{\frac{\ab\log^2(1/\dist)}{\log\Paren{\eps\ab}}}
\end{align}
for the sensitivity result to hold. 

We note that the bound above blows up when $\eps\le \frac1\ab$. However, we note that our lower bound implies that we need at least $\Omega(1/\eps) =\Omega(k)$ samples in this case, which is not in the sub-linear regime that we are interested in. We therefore consider only the regime where the privacy parameter $\eps$ is at least $1/k$. 

\subsubsection{Lower Bound for Support Size Estimation}
\label{sec:ssize-lb}
In this section, we prove a lower bound for support size estimation, as described in Theorem~\ref{thm:ssize}.
The techniques are similar to those for support coverage in Section~\ref{sec:coverage-lb}.
The first term of the complexity is the lower bound for non-private setting.
This follows by combining the lower bound of~\cite{OrlitskySW16} for support coverage, with the equivalence between estimation of support size and coverage as implied by Lemma~\ref{lem:sssc}.
We focus on the second term in the sequel.

Consider the following two distributions: $u_1$ is a uniform distribution over $[k]$ and $u_2$ is a uniform distribution over $[(1-\dist)k]$. Then the support size of these two distribution differs by $\dist k$, and $\dtv(u_1,u_2) = \dist$. 

Hence by Lemma~\ref{lem:min_coupling}, we know the following:

\begin{lemma} \label{lem:coupling_ssize}
	Suppose $u_1 \sim U[k]$ and $u_2 \sim U[(1-\dist)k]$, there is a coupling between $u_1^\ns$ and $u_2^\ns$ with expected Hamming distance equal to $ \dist \ns $.
\end{lemma}

Moreover, given $\ns$ samples, we must be able to privately distinguish between $u_1$ and $u_2$ given an $\dist$ accurate estimator of entropy with privacy considerations. 
Thus, according to Lemma~\ref{lem:coupling}
and Lemma~\ref{lem:coupling_ssize}, we have:
\[
\dist \ns\ge \frac1{\eps}\Rightarrow n = \Omega\Paren{\frac1{\eps\dist}}.
\]

\subsection{Entropy Estimation}
\label{sec:entropy}
In this section, we prove our main theorem about entropy estimation, Theorem~\ref{thm:entropy}:

\entropy*

\noindent We describe and analyze two upper bounds.
The first is based on the empirical entropy estimator, and is described and analyzed in Section~\ref{sec:entropy-ub-empirical}.
The second is based on the method of best-polynomial approximation, and appears in Section~\ref{sec:entropy-ub-poly}.
Finally, our lower bound is in Section~\ref{sec:entropy-lb}.
\subsubsection{Upper Bound for Entropy Estimation: The Empirical Estimator}
\label{sec:entropy-ub-empirical}
Our first private entropy estimator is derived by adding Laplace noise into the empirical estimator. The parameter of the Laplace distribution is $\frac{ \Delta(\entemp)}{\eps}$, where $\Delta(\entemp)$ denotes the sensitivity of the empirical estimator. By analyzing its sensitivity and bias, we prove an upper bound on the sample complexity for private entropy estimation and get the first upper bound in Theorem~\ref{thm:entropy}.

Let $\emp$ be the empirical distribution, and let $\entemp$ be the entropy of the empirical distribution.
The theorem is based on the following three facts:
\begin{align}
&\Delta(\entemp) = O\Paren{\frac{\log\ns}{\ns}}.\label{eqn:sensitivity-emp}\\
&\absv{\entp{\p}-\expectation{\entemp}} =O\Paren{\frac{\ab}{\ns}},\label{eq:bias-emp}\\
&\variance{\entemp} =O\Paren{\frac{\log^2(\min\{\ab,\ns\})}{\ns}},\label{eq:variance-emp}
\end{align}
With these three facts in hand, the sample complexity of the empirical estimator can be bounded as follows.
 By Lemma~\ref{lem:main-sensitivity}, 
we need $\Delta(\entemp) \le \dist \eps$, which gives $\ns=O\Paren{\frac1{\dist \eps} \log(\frac1{\dist \eps})}$. We also need $\absv{\entp{\p}-\expectation{\entemp}} = O\Paren{\dist}$ and $\variance{\entemp} = O\Paren{\dist^2}$, which gives $\ns = O\Paren{\frac{\ab}{\ns}+   \frac{ \log^2(\min\{\ab,\ns\})}{\dist^2}}$.

\paragraph{Proof of~\eqref{eqn:sensitivity-emp}.}
The largest change in any $\nsmb$ when we change one symbol is one. Moreover, at most two $\nsmb$ change. Therefore, 
\begin{align}
\Delta(\entemp) &\leq  2\cdot \max_{j=1\ldots \ns-1} \absv{ \frac{j+1}{n}\log\frac{n}{j+1}-\frac{j}{n}\log\frac{n}{j}}\nonumber\\
&=  2\cdot \max_{j=1\ldots \ns-1}\absv{\frac{j}{\ns}\log \frac{j}{j+1}+\frac1{\ns}\log \frac{\ns}{j+1}}\label{eqn:sum-one}\\
&\leq  2 \cdot \max_{j=1\ldots \ns-1}\max\left\{\absv{\frac{j}{\ns}\log \frac{j}{j+1}}, \absv{ \frac1{\ns}\log \frac{\ns}{j+1}}\right\}\label{eqn:sum-two}\\
&\leq  2\cdot \max\left\{\frac 1\ns, \frac{\log \ns}{\ns}\right\},\nonumber\\
& = 2\cdot \frac{\log \ns}{\ns}.
\end{align}

\paragraph{Proof of~\eqref{eq:bias-emp}.} By the concavity of entropy function, we know that 
\[
\expectation{\entp{\emp}} \le \entp{\p}.
\]
Therefore, 
\begin{align}
&\expectation{\absv{\entp{\p}-\entp{\emp}}} =  \entp{\p}-\expectation{{\entp{\emp}}}\nonumber\\
&=  \expectation{\sum_{\smb}\Paren{\pemps\log \pemps-\ps\log\ps}} \nonumber\\
&=  \expectation{\sum_{\smb}\pemps\log \frac{\pemps}{\ps}}+\expectation{\sum_{\smb}(\pemps-\ps)\log\ps}\nonumber\\
&=  \expectation{\kldist{\emp}{\p}}\label{eqn:kl-exp}\\
&\leq \expectation{\chisquare{\emp}{\p}}\label{eqn:kl-chisq}\\
&=  \expectation{\sum_{\smb}\frac{(\pemps-\ps)^2}{\ps}}\nonumber\\
&\le  {\sum_{\smb}\frac{(\ps/\ns)}{\ps}}\label{eqn:var-bin}\\
&= \frac\ab\ns.
\end{align}

\paragraph{Proof of~\eqref{eq:variance-emp}.}The variance bound of $\frac{\log^2 k}{n}$ is given precisely in Lemma 15 of~\cite{JiaoVHW17}. 
To obtain the other half of the bound of, we apply the bounded differences inequality in the form stated in Corollary 3.2 of~\cite{BoucheronLM13}.
\begin{lemma}
	\label{lem:bdd-diff-var}
	Let $f:\Omega^{\ns}\to\RR$ be a function. Suppose further that 
	\[
	\max_{z_1,\ldots, z_n, z_{i}^{'}} \absv{f(z_1,\ldots,z_n) - f(z_1,\ldots, z_{i-1}, z_i^{'},\ldots,z_n)}\le c_i.
	\]
	Then for independent variables $Z_1,\ldots,Z_{\ns}$, 
	\[
	\Var\Paren{f(Z_1,\ldots,Z_{\ns})} \le \frac14\sum_{i=1}^{\ns} c_i^2.
	\]
\end{lemma}

Therefore, using Lemma~\ref{lem:bdd-diff-var} and Equation~\eqref{eqn:sensitivity-emp}
\[
\variance{\entemp}\le \ns\cdot\Paren{\frac{4\log^2\ns}{\ns^2}} = \frac{4\log^2\ns}{\ns}.
\]

%

\subsubsection{Upper Bound for Entropy Estimation: Best-Polynomial Approximation}
\label{sec:entropy-ub-poly}
We prove an upper bound on the sample complexity for private entropy estimation if one  adds Laplace noise into best-polynomial estimator.
This will give us the second upper bound in Theorem~\ref{thm:entropy}.

In the non-private setting the optimal sample complexity of estimating $H(\p)$ over $\Delta_k$ is given by Theorem 1 of~\cite{WuY16}
\[
\Theta\Paren{\frac{\ab}{\dist\log\ab}+\frac{\log^2(\min\{\ab,\ns\})}{\dist^2}}.
\]
However, this estimator can have a large sensitivity.~\cite{AcharyaDOS17} designed an estimator that has the same sample complexity but a smaller sensitivity. We restate Lemma 6 of~\cite{AcharyaDOS17} here:
\begin{lemma}
Let $\lambda>0$ be a fixed small constant, which may be taken to be any value between $0.01$ and $1$. Then there is an entropy estimator with sample complexity 
\[
\Theta\Paren{\frac1{\lambda^2}\cdot\frac{\ab}{\dist\log\ab}+\frac{\log^2(\min\{\ab,\ns\})}{\dist^2}},
\]
and has sensitivity $\ns^\lambda/\ns$. 
\end{lemma}
We can now invoke Lemma~\ref{lem:main-sensitivity}
on the estimator in this lemma to obtain the upper bound on private entropy estimation. 

%

\subsubsection{Lower Bound for Entropy Estimation}
\label{sec:entropy-lb}
We now prove the lower bound for entropy estimation. 
Note that any lower bound on privately testing two distributions $\p$, and $\q$ such that $H(\p)-H(\q)= \Theta(\alpha)$ is a lower bound on estimating entropy. 

We analyze the following construction for Proposition 2 of~\cite{WuY16}. The two distributions $\p$, and $\q$ over $[k]$ are defined as:
\begin{align}
p(1) = \frac23,& p(i) = \frac{1-p(1)}{k-1}, \text{for $i=2,\ldots,k$},\\
q(1) = \frac{2-\eta}3,& q(i) = \frac{1-q(1)}{k-1}, \text{for $i=2,\ldots,k$}.
\end{align}
Then, by the grouping property of entropy, 
\[
H(\p) = h(2/3) + \frac13\cdot \log (k-1), \text{ and } H(\q) = h((2-\eta)/3) + \frac{1+\eta}3\cdot \log (k-1),
\]
which gives
\[
H(\p)-H(\q) = \Omega(\eta\log k). 
\]
For $\eta = \dist/\log\ab$, the entropy difference becomes $\Theta(\dist)$. 

The total variation distance between $\p$ and $\q$ is $\eta/3$. By Lemma 5 in the paper, 
there is a coupling over $\Xon$, and $\Yon$ generated from $\p$ and $\q$ with expected Hamming distance at most $\dtv(\p,\q)\cdot \ns$. 
This along with Lemma 2 in the paper 
gives a lower bound of $\Omega\Paren{\log k/\dist\eps}$ on the sample complexity.  

\section{Experiments}
\label{sec:exp}
We evaluated our methods for entropy estimation and support coverage on both synthetic and real data.
Overall, we found that privacy is quite cheap: private estimators achieve accuracy which is comparable or near-indistinguishable to non-private estimators in many settings.
Our results on entropy estimation and support coverage appear in Sections~\ref{sec:exp-entropy} and~\ref{sec:exp-coverage}, respectively.
Code of our implementation is available at \url{https://github.com/HuanyuZhang/INSPECTRE}.

\subsection{Entropy}
\label{sec:exp-entropy}
We compare the performance of our entropy estimator with a number of alternatives, both private and non-private.
Non-private algorithms considered include the plug-in estimator (\plugin),  the Miller-Madow Estimator (\MM)~\cite{Miller55}, the sample optimal polynomial approximation estimator (\polyn) of~\cite{WuY16}.
We analyze the privatized versions of plug-in, and \polyn\ in Sections~\ref{sec:entropy-ub-empirical} and~\ref{sec:entropy-ub-poly}, respectively.
The implementation of the latter is based on code from the authors of~\cite{WuY16}\footnote{See \url{https://github.com/Albuso0/entropy} for their code for entropy estimation.}.
We compare performance on different distributions including uniform, a distribution with two steps, Zipf(1/2), a distribution with Dirichlet-1 prior, and a distribution with Dirichlet-$1/2$ prior, and over varying support sizes.

While \plugin, and \MM\ are parameter free,  \polyn\ (and its private counterpart) have to choose the degree $L$ of the polynomial  to use, which manifests in the parameter $\lambda$ in the statement of Theorem~\ref{thm:entropy}.~\cite{WuY16} suggest the value of $L=1.6\log k$ in their experiments. However, since we add further noise, we choose a single $L$ as follows: (i) Run  privatized \polyn\ for different $L$ values and distributions for $k=2000$, $\eps=1$, (b) Choose the value of $L$ that performs well across different distributions (See Figure~\ref{fig:Lcomparison}). We choose $L=1.2\cdot \log k$ from this, and  use it for all other experiments. 
To evaluate the sensitivity of \polyn, we computed the estimator's value at all possible input values, computed the sensitivity, (namely, $\Delta = \max_{\ham{\Xon}{\Yon}\le 1}|\polyn(\Xon)-\polyn(\Yon)|$), and added noise distributed as $\text{Lap}\left(0,\frac{\Delta}{\eps}\right)$.

%
%
\begin{figure*}[h]
\centering
\captionsetup[subfigure]{labelformat=empty}
\subfigure{
\includegraphics[width=0.18\textwidth]{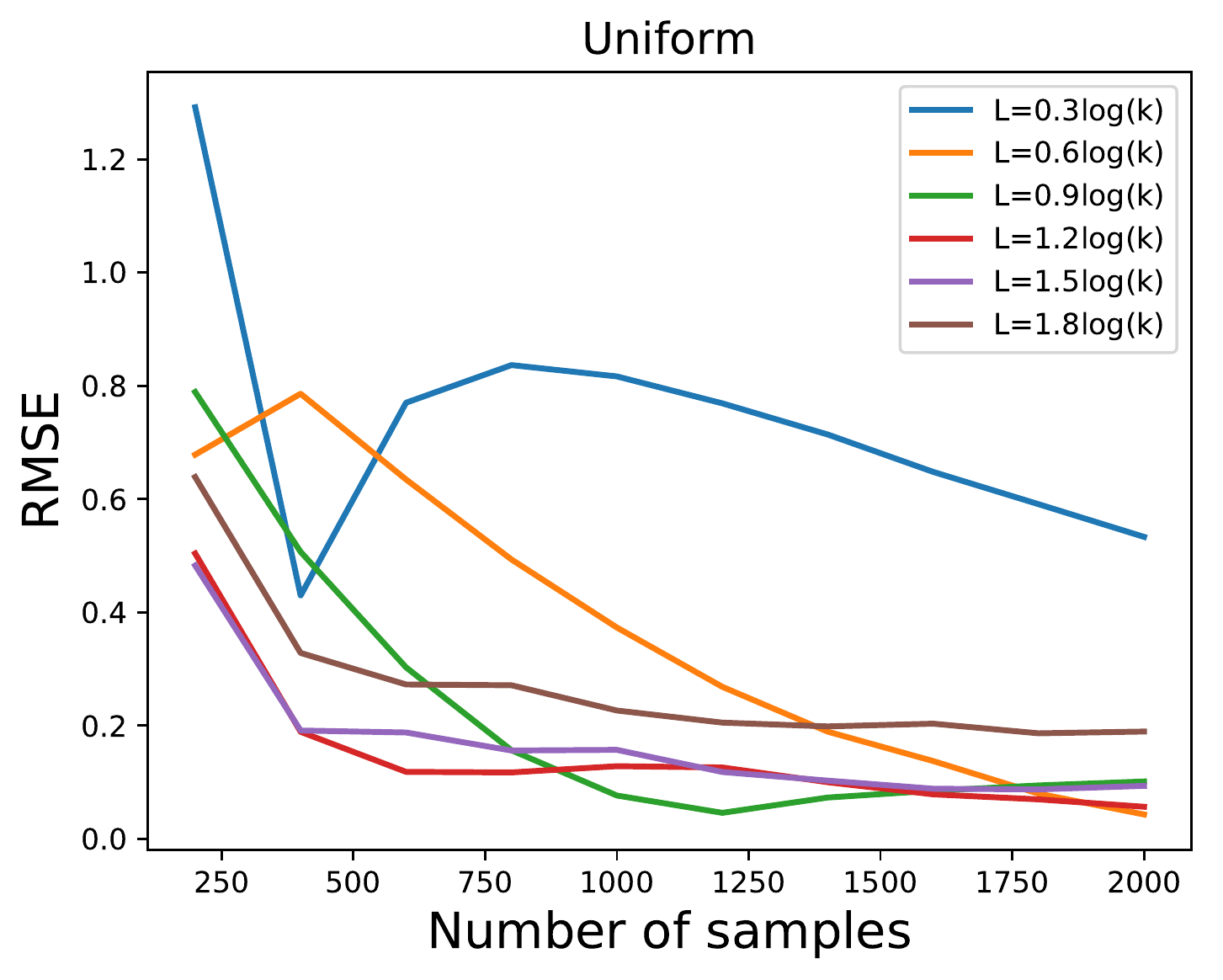}
}
\subfigure{
\includegraphics[width=0.18\textwidth]{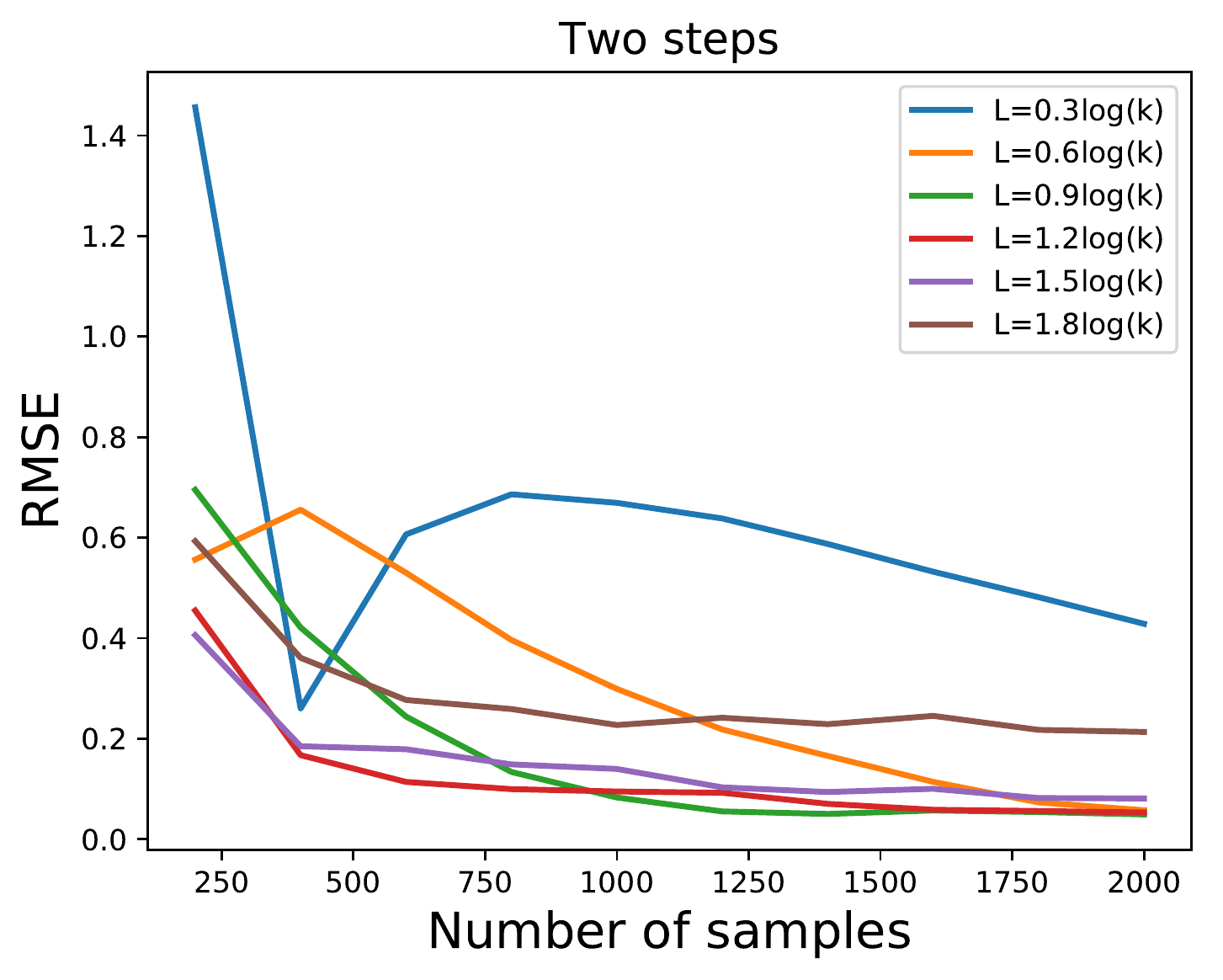}
}
\subfigure{
\includegraphics[width=0.18\textwidth]{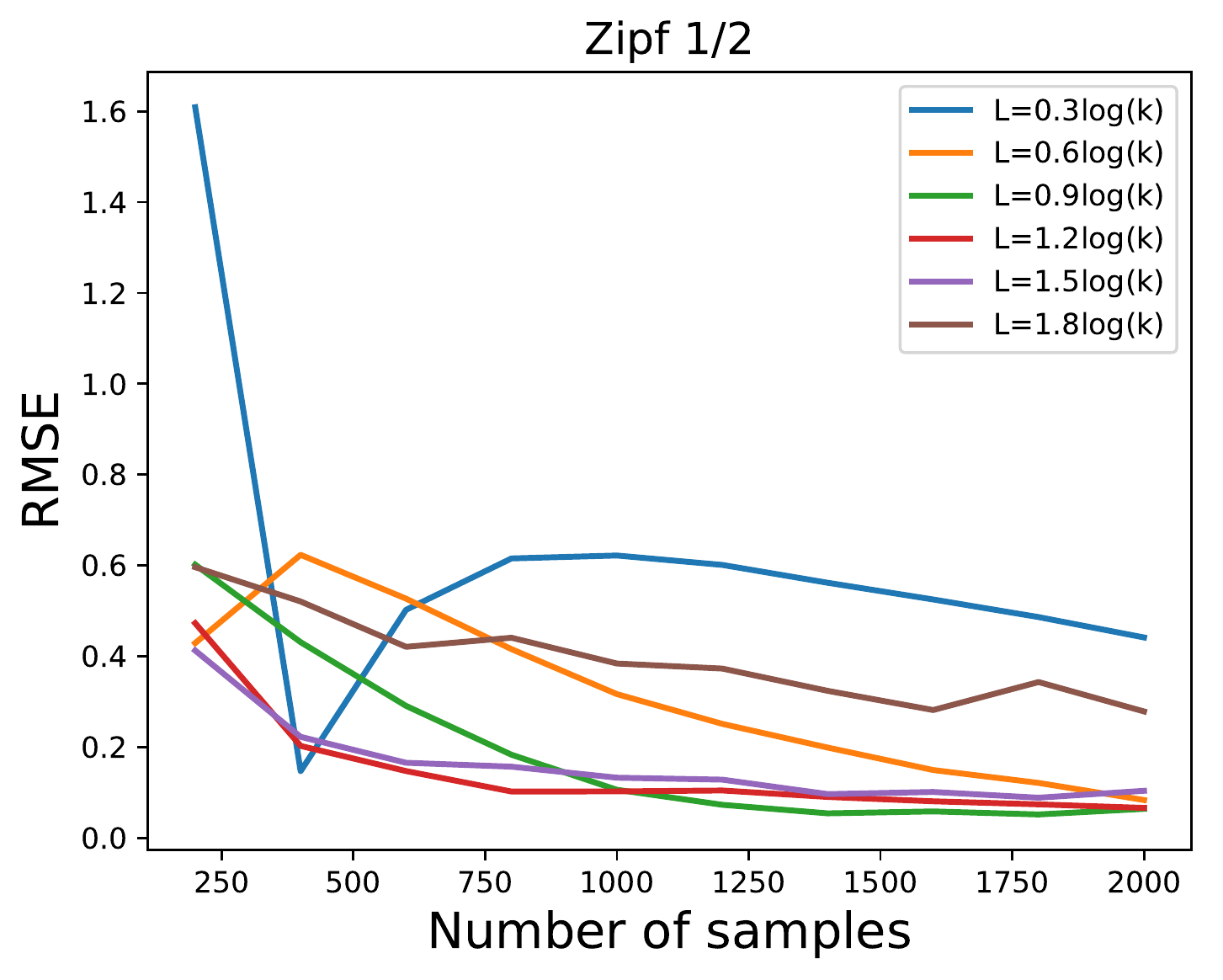}
}
\subfigure{
\includegraphics[width=0.18\textwidth]{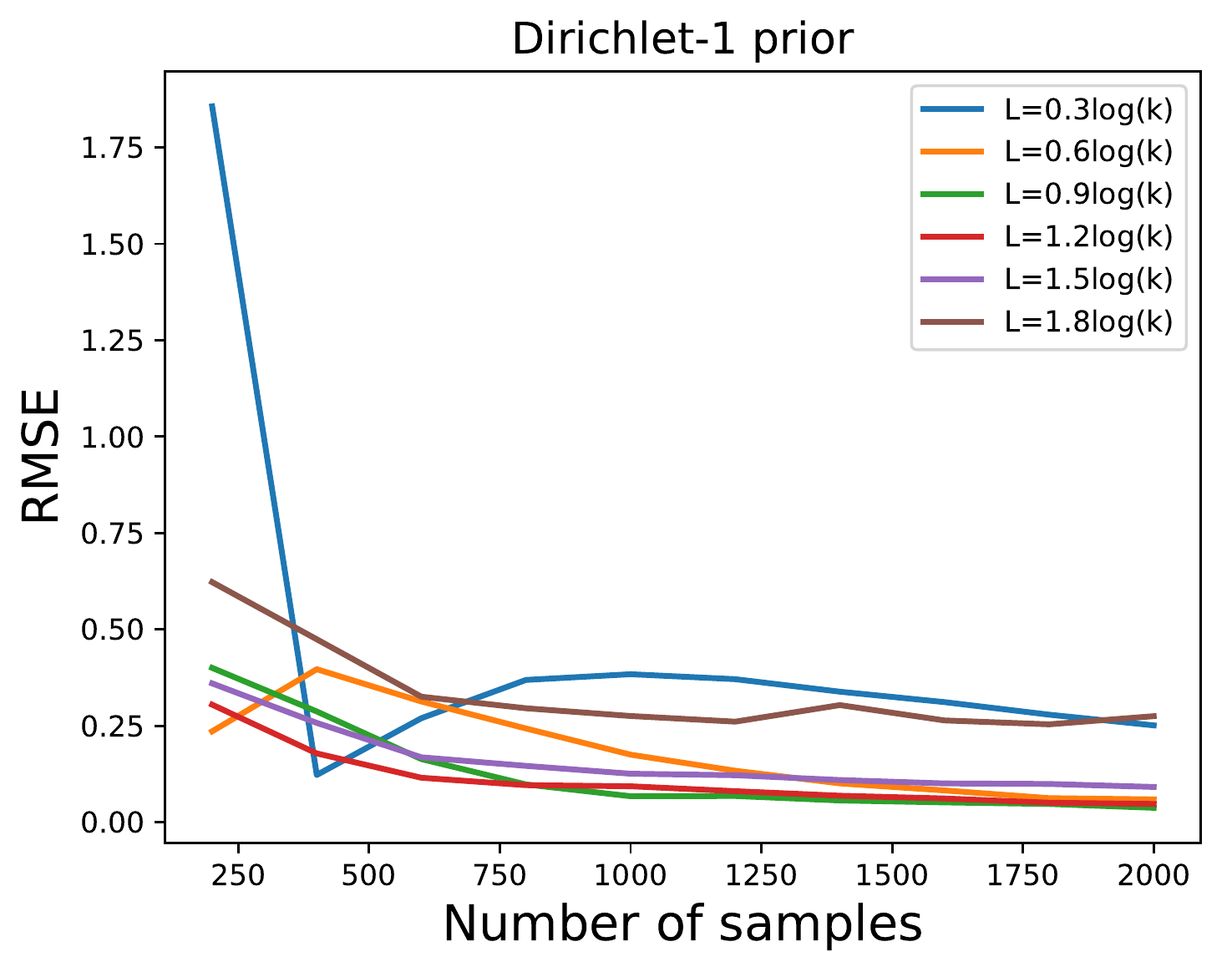}
}
\subfigure{
\includegraphics[width=0.18\textwidth]{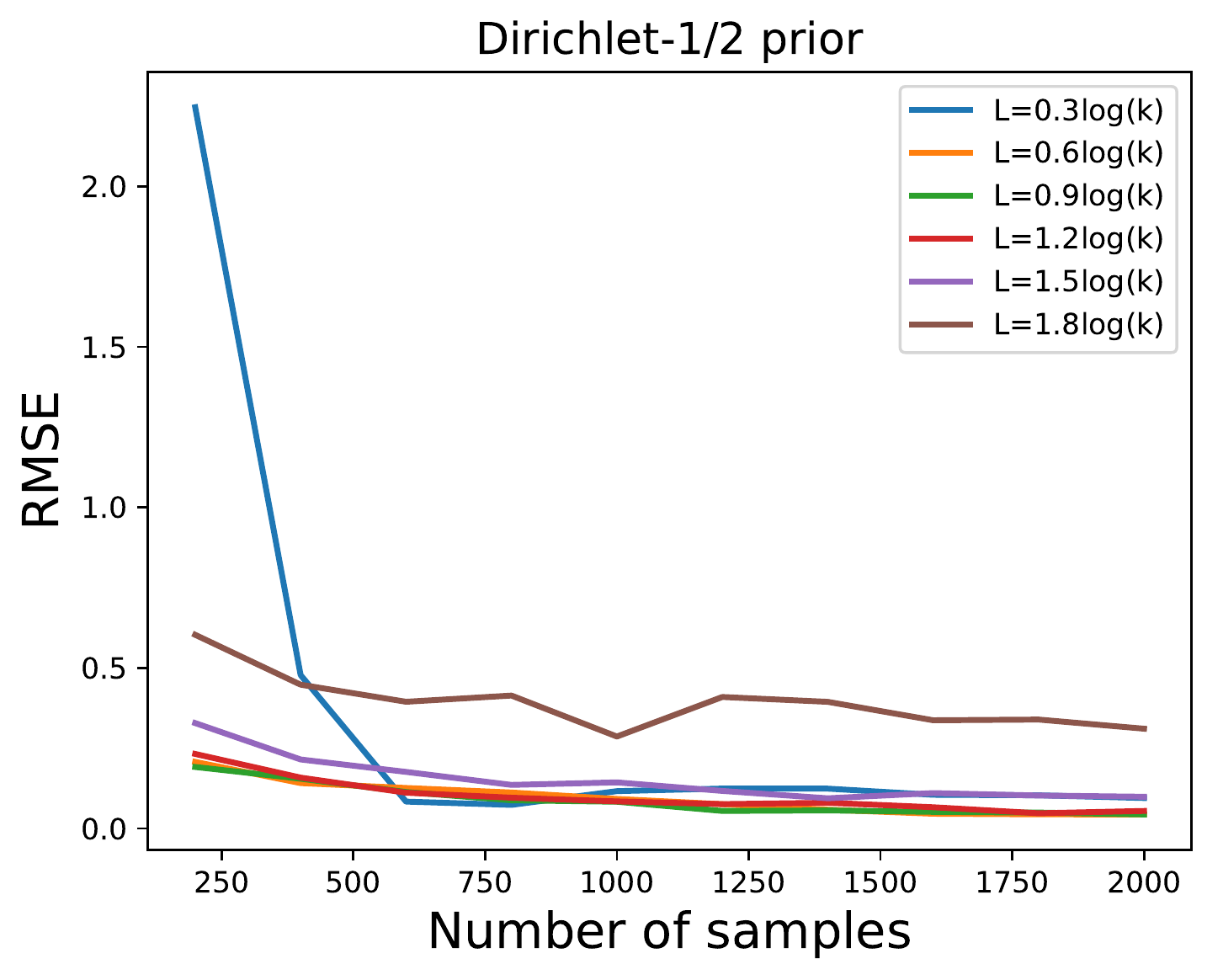}
}
\caption{RMSE comparison for private Polynomial Approximation Estimator with various values for degree $L$, $k=2000$, $\eps=1$.} 
\label{fig:Lcomparison} 

\end{figure*}


The RMSE of various estimators for $k=1000$, and $\eps=1$ for various distributions  are illustrated in Figure~\ref{fig:entropy_k1000}.
The RMSE is averaged over 100 iterations in the plots.

%
%
\begin{figure*}[h]
\centering
\captionsetup[subfigure]{labelformat=empty}
\subfigure{
\includegraphics[width=0.18\textwidth]{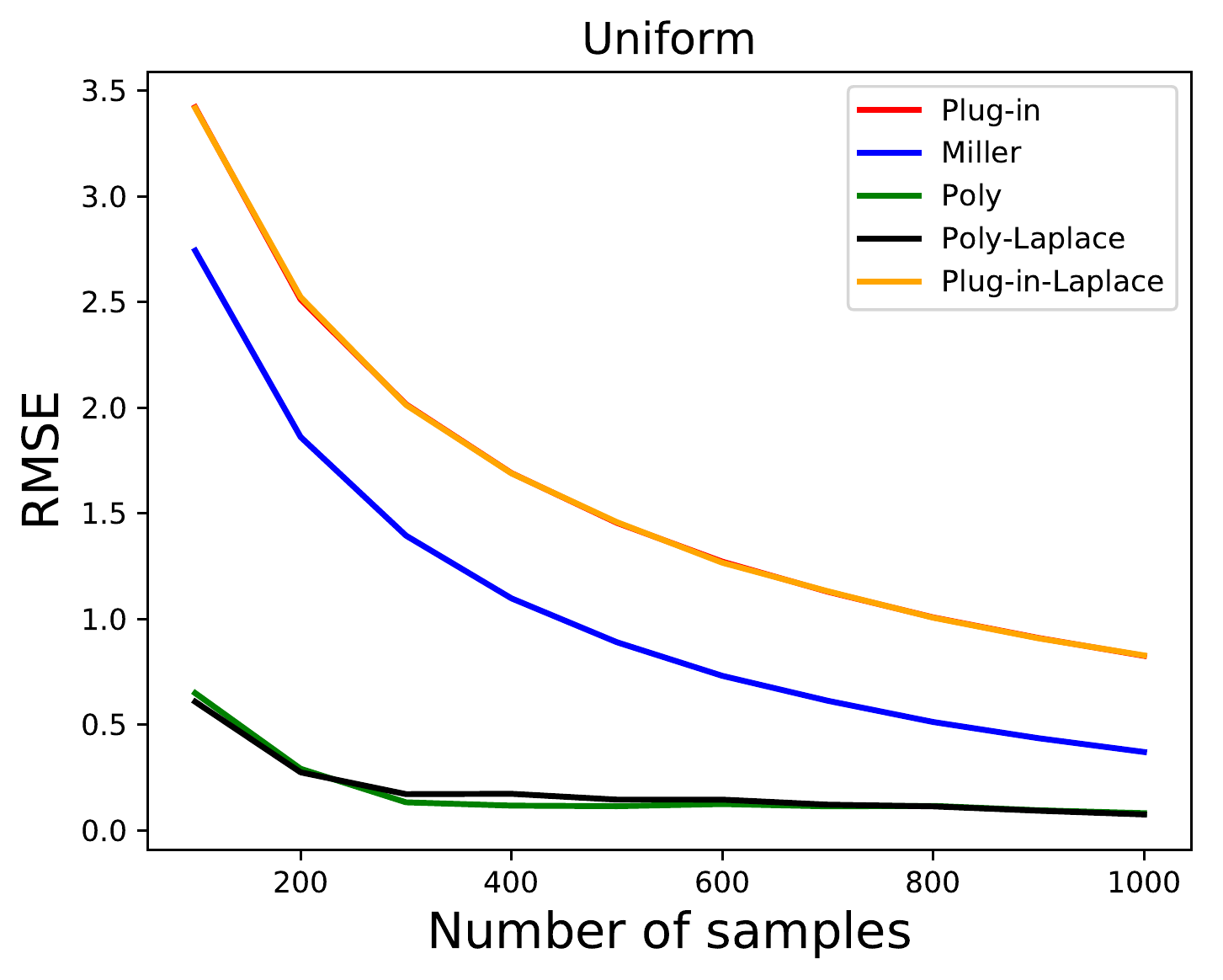}
}
\subfigure{
\includegraphics[width=0.18\textwidth]{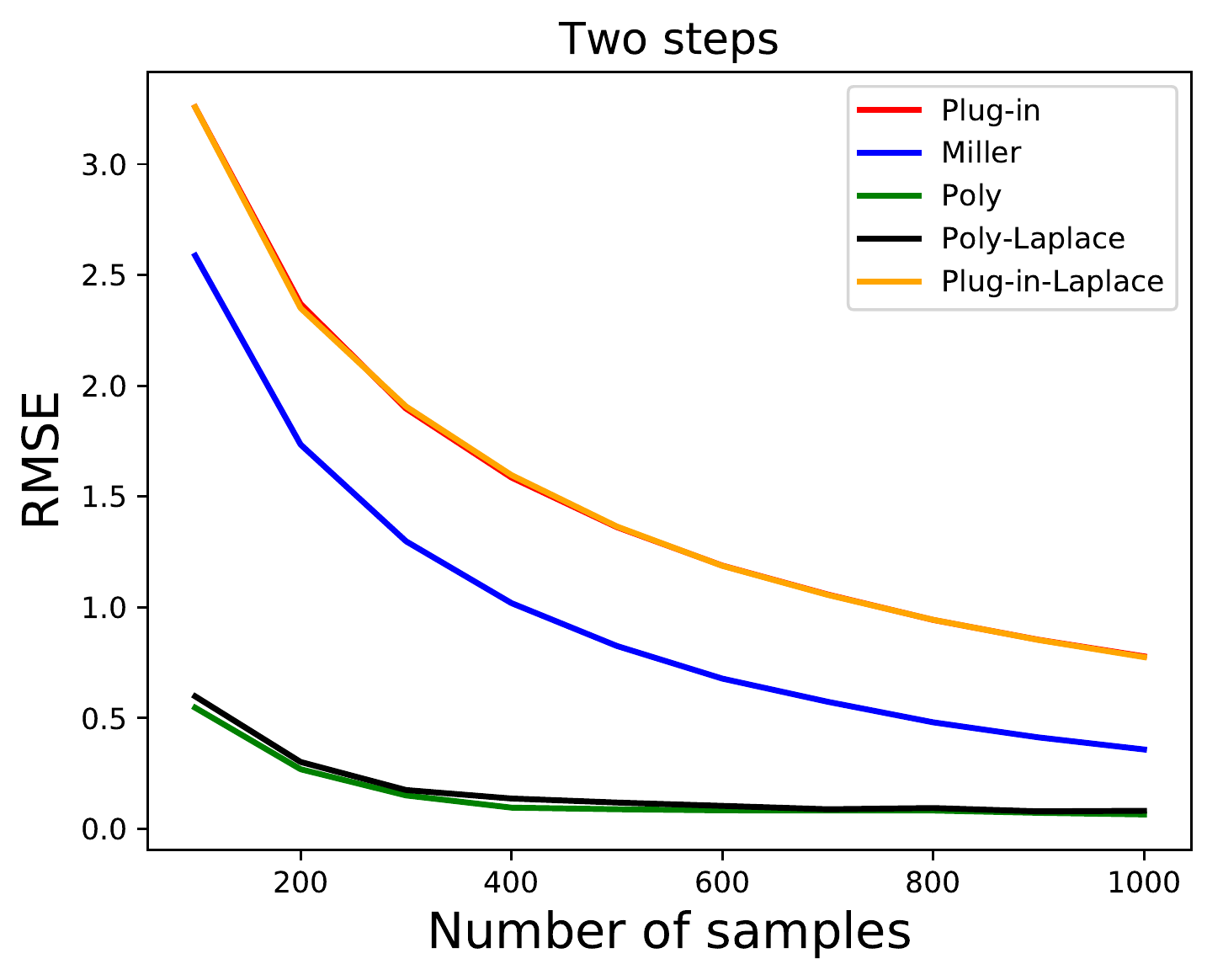}
}
\subfigure{
\includegraphics[width=0.18\textwidth]{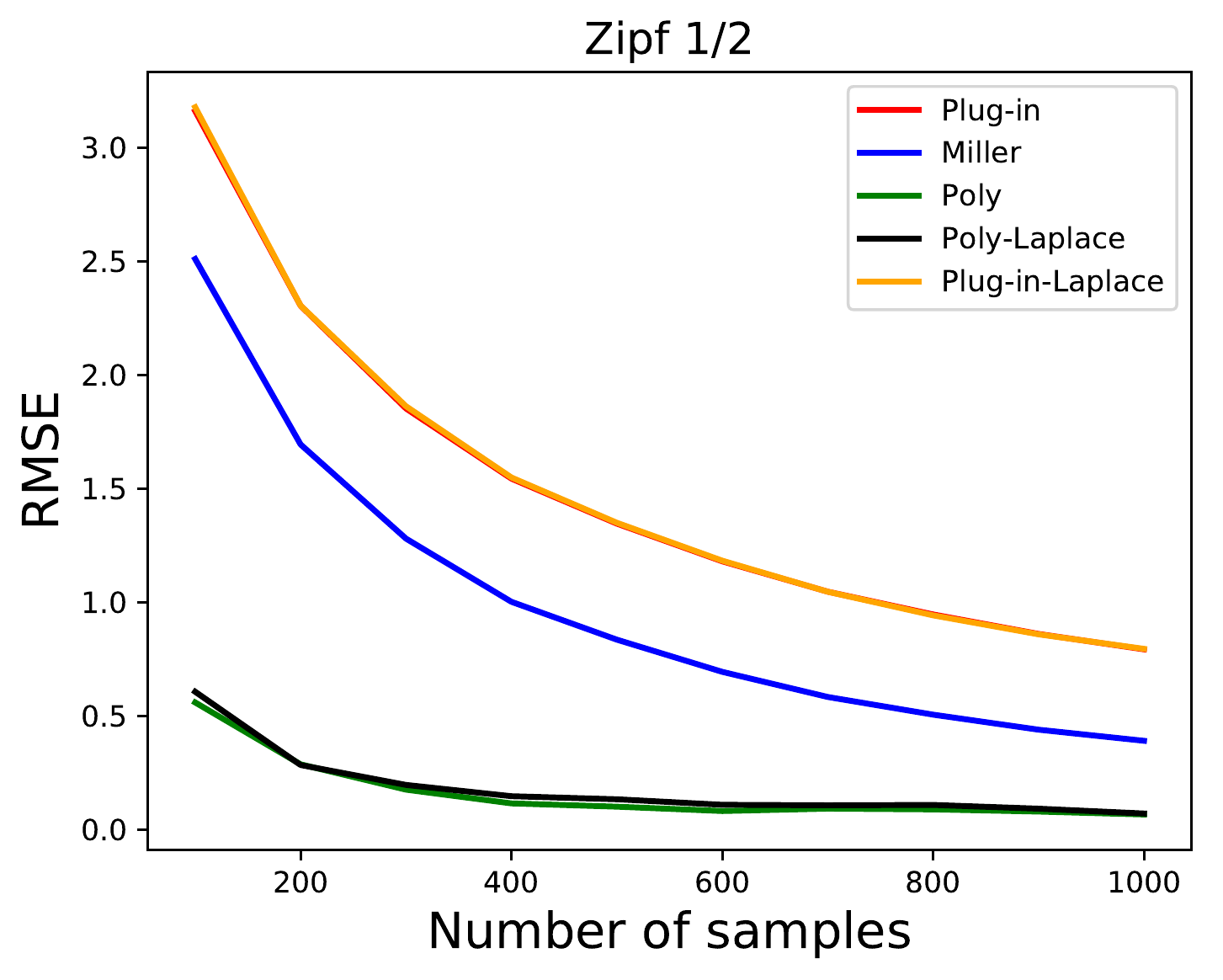}
}
\subfigure{
\includegraphics[width=0.18\textwidth]{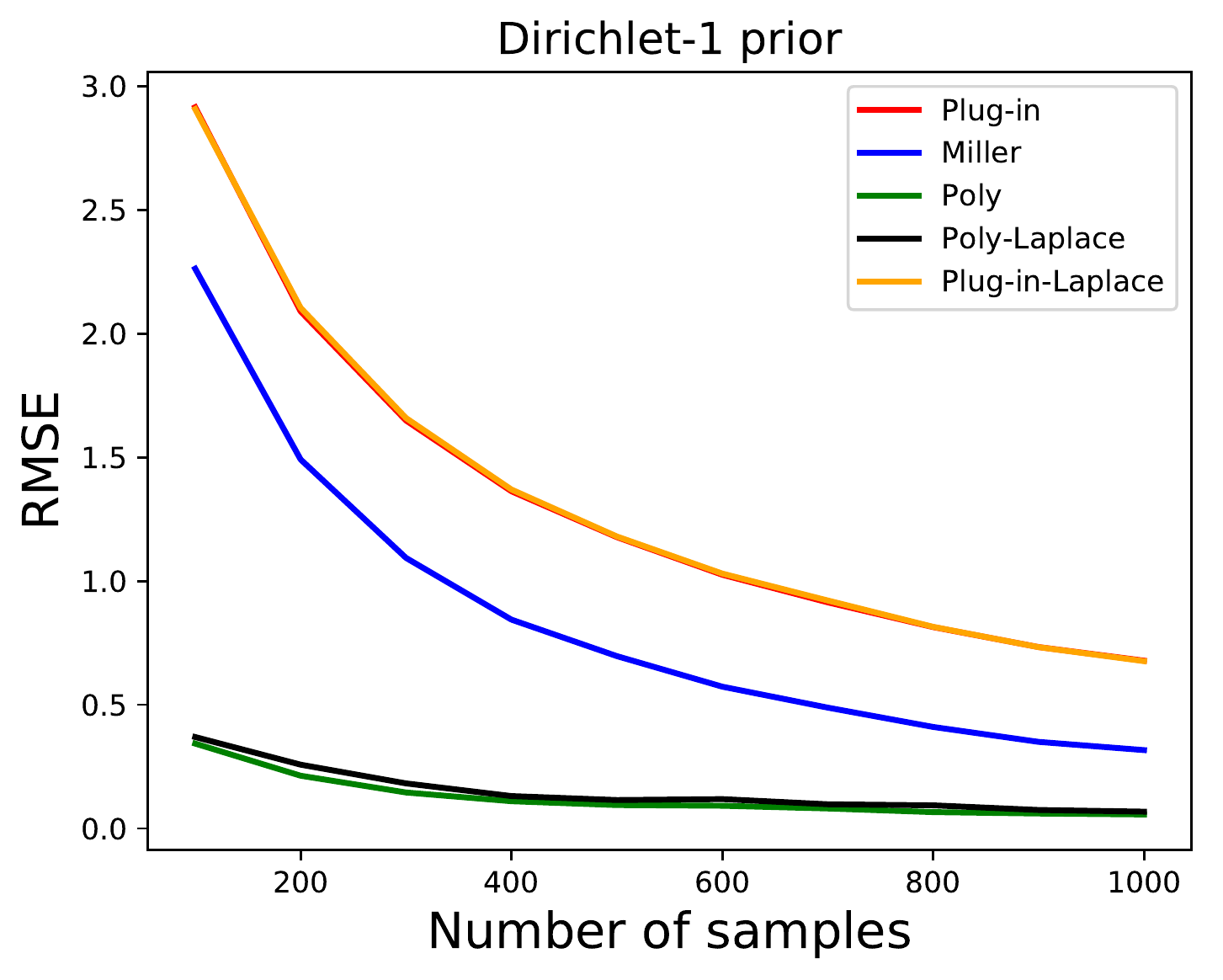}
}
\subfigure{
\includegraphics[width=0.18\textwidth]{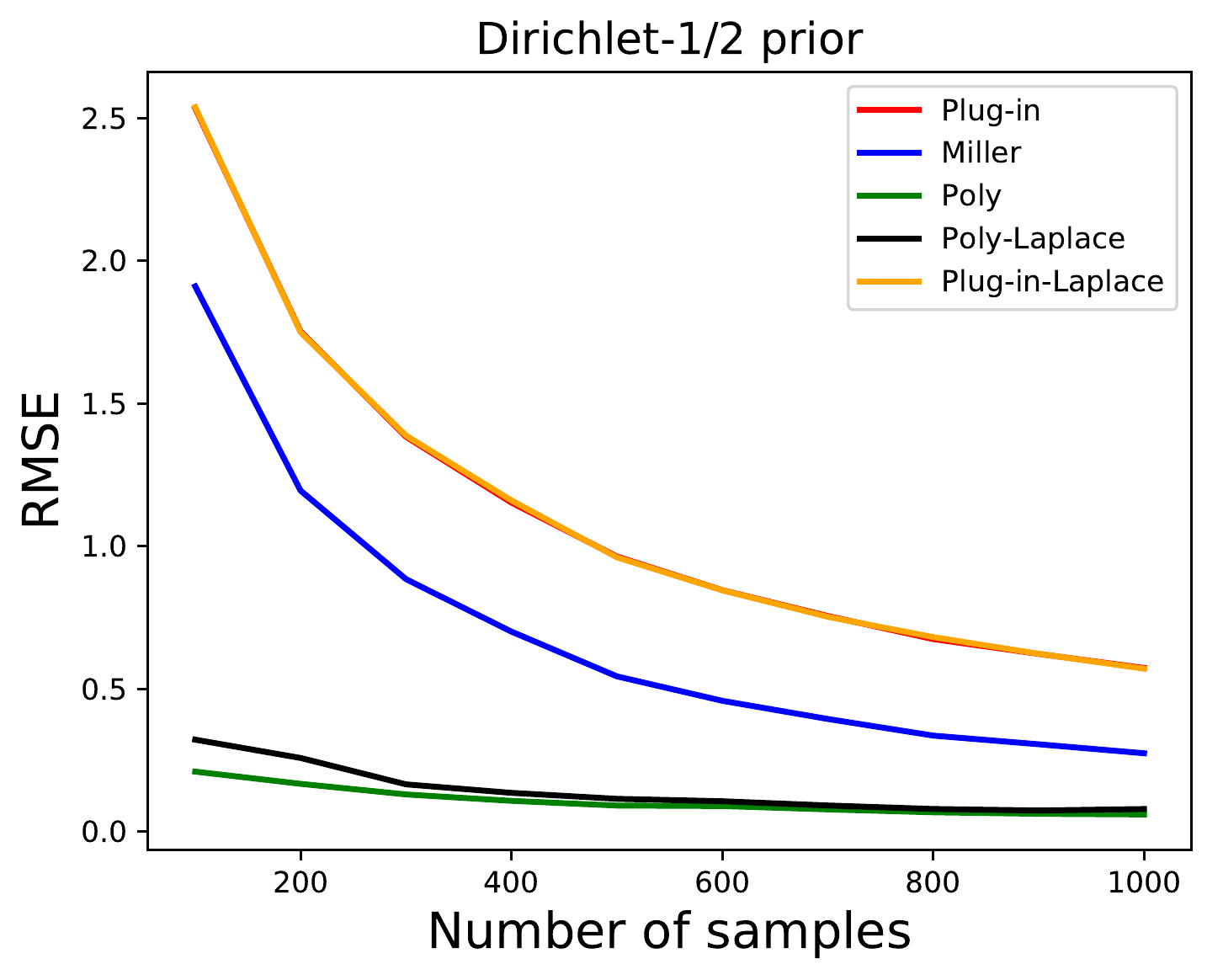}
}
\caption{Comparison of various estimators for the entropy, $k=1000$, $\eps =1$.} 
\label{fig:entropy_k1000} 

\end{figure*}

We observe that the performance of our private-\polyn\ is near-indistinguishable from the non-private \polyn, particularly as the number of samples increases.
It also performs significantly better than all other alternatives, including the non-private Miller-Madow and the plug-in estimator.
The cost of privacy is minimal for several other settings of $k$ and $\varepsilon$, for which results appear in Section~\ref{sec:supp-experiments}.

\subsection{Support Coverage}
\label{sec:exp-coverage}

We investigate the cost of privacy for the problem of support coverage.
We provide a comparison between the Smoothed Good-Toulmin  estimator (SGT) of~\cite{OrlitskySW16} and our algorithm, which is a privatized version of their statistic (see Section~\ref{sec:coverage-ub}).
Our implementation is based on code provided by the authors of~\cite{OrlitskySW16}.
As shown in our theoretical results, the sensitivity of SGT  is at most $2 (1+e^r(t-1))$, necessitating the addition of Laplace noise with parameter  $2 (1+e^{r(t-1)})/\eps$.
Note that while the theory suggests we select the parameter $r = \log(1/\dist)$, $\dist$ is unknown.
We instead set $r = \frac{1}{2t} \log_{e} \frac{n(t+1)^2}{t-1}$, as previously done in~\cite{OrlitskySW16}.

\subsubsection{Evaluation on Synthetic Data}

In our synthetic experiments, we consider different distributions over different support sizes $\ab$. We generate $\ns = \ab/2$ samples, and then estimate the support coverage at $m=\ns \cdot t$. For large $t$, estimation is harder. 
Some results of our evaluation on synthetic are displayed in Figure~\ref{fig:synthetic_k20000}.
We compare the performance of SGT, and privatized versions of SGT with parameters $\eps = 1, 2,$ and $10$.
For this instance, we fixed the domain size $k = 20000$.
We ran the methods described above with $n = k/2$ samples, and estimated the support coverage at $m = nt$, for $t$ ranging from $1$ to $10$.
The performance of the estimators is measured in terms of RMSE over 1000 iterations.
\begin{figure*}[h]
\centering

\subfigure{
\includegraphics[width=0.18\textwidth]{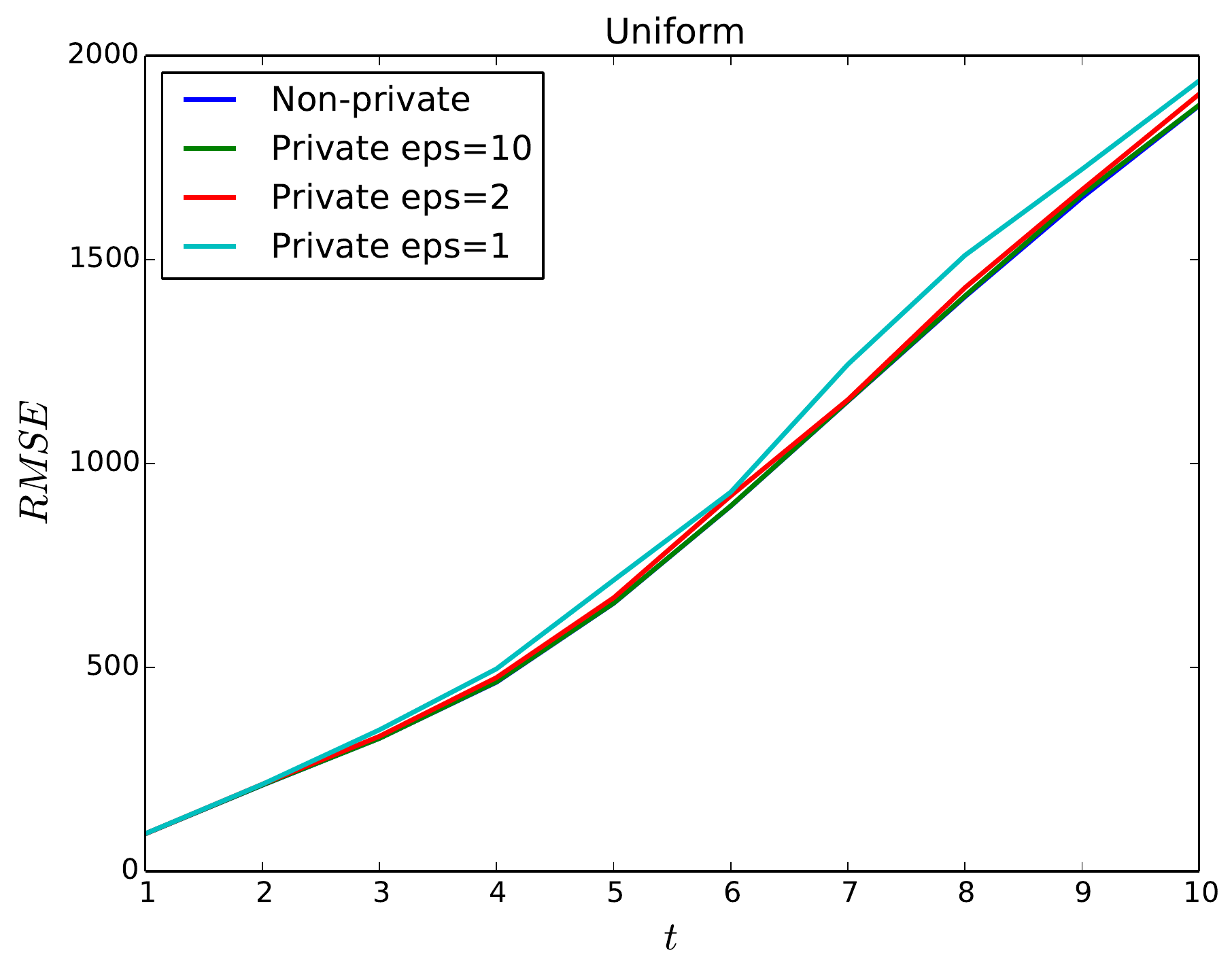}
}
\subfigure{
\includegraphics[width=0.18\textwidth]{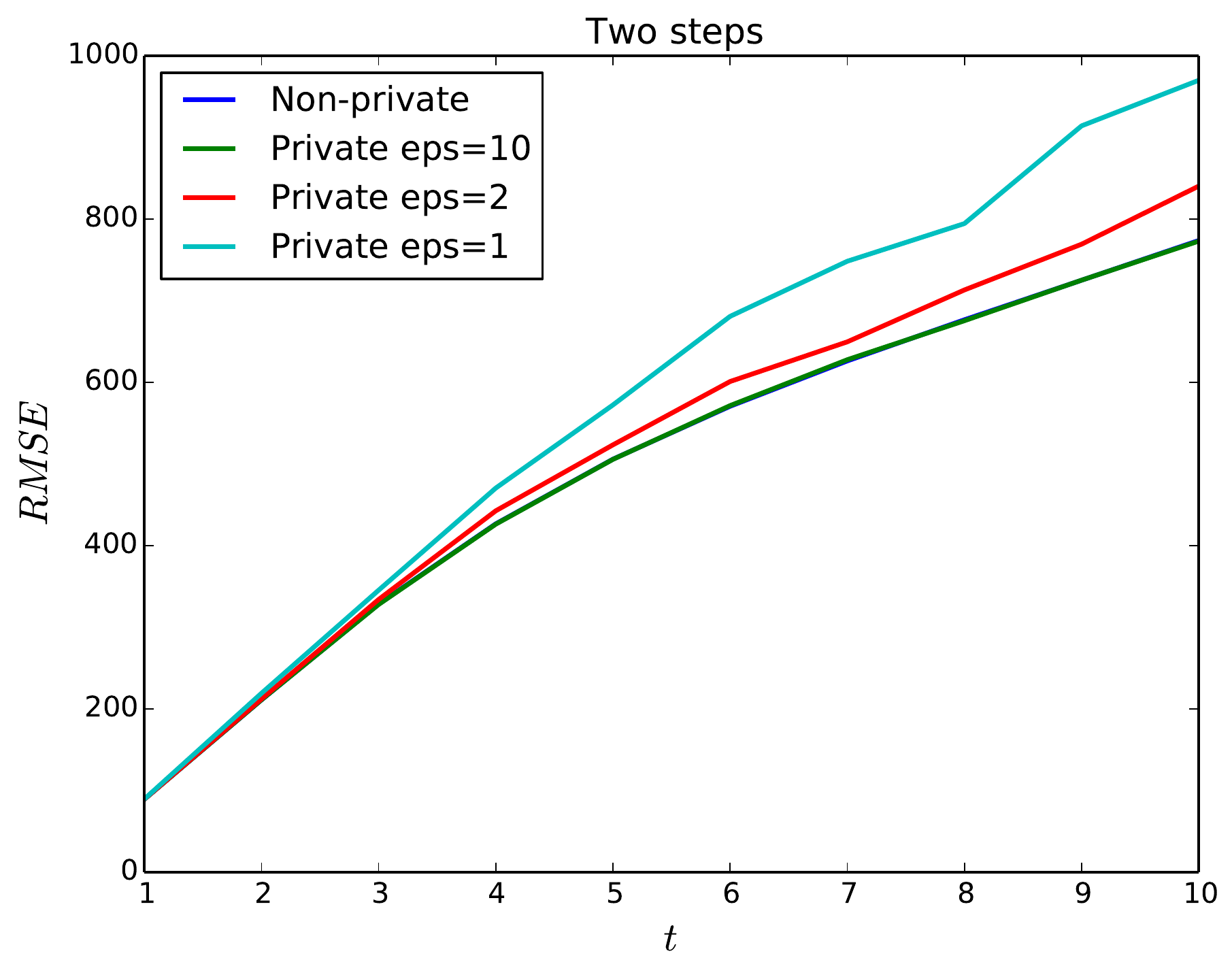}
}
\subfigure{
\includegraphics[width=0.18\textwidth]{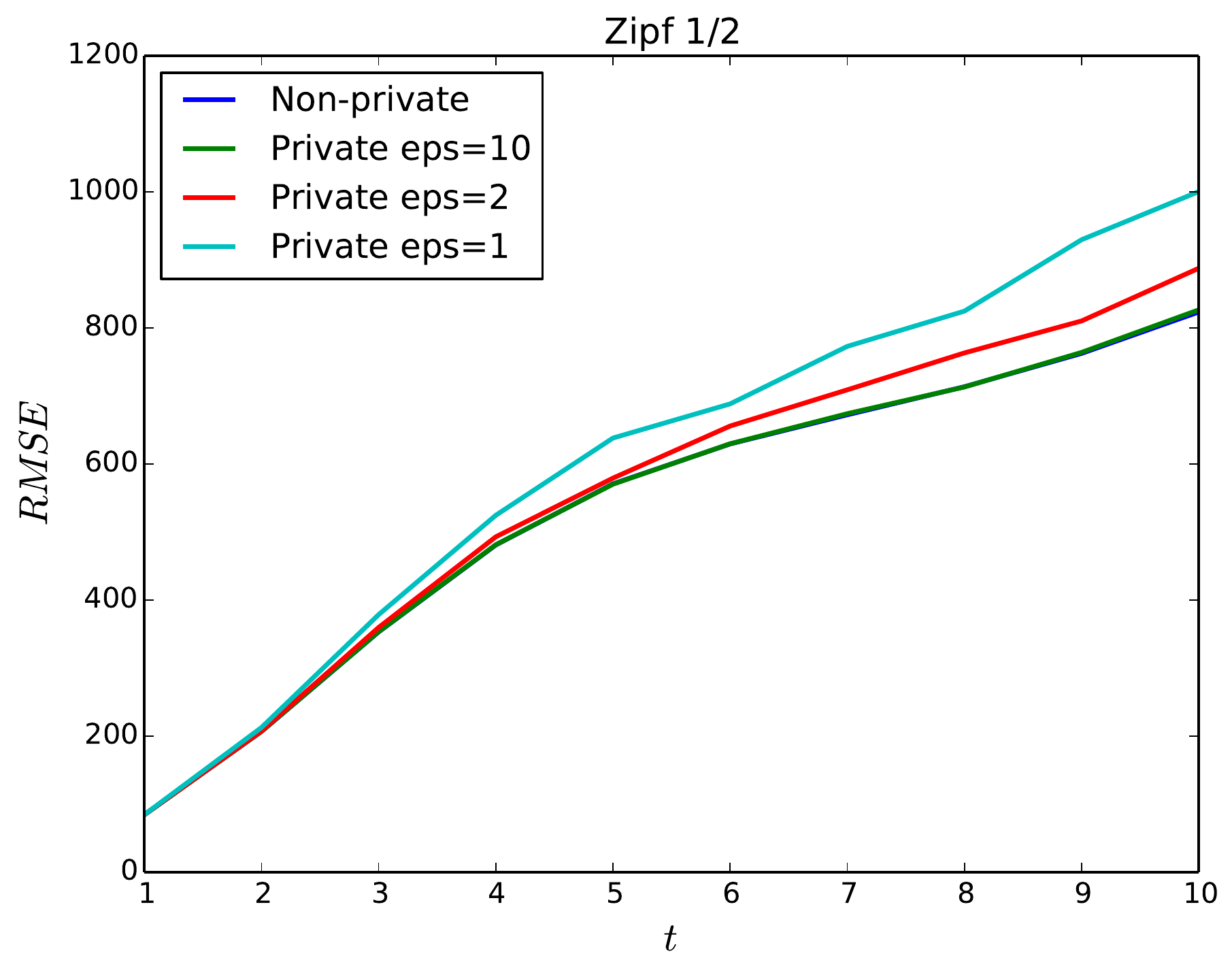}
}
\subfigure{
\includegraphics[width=0.18\textwidth]{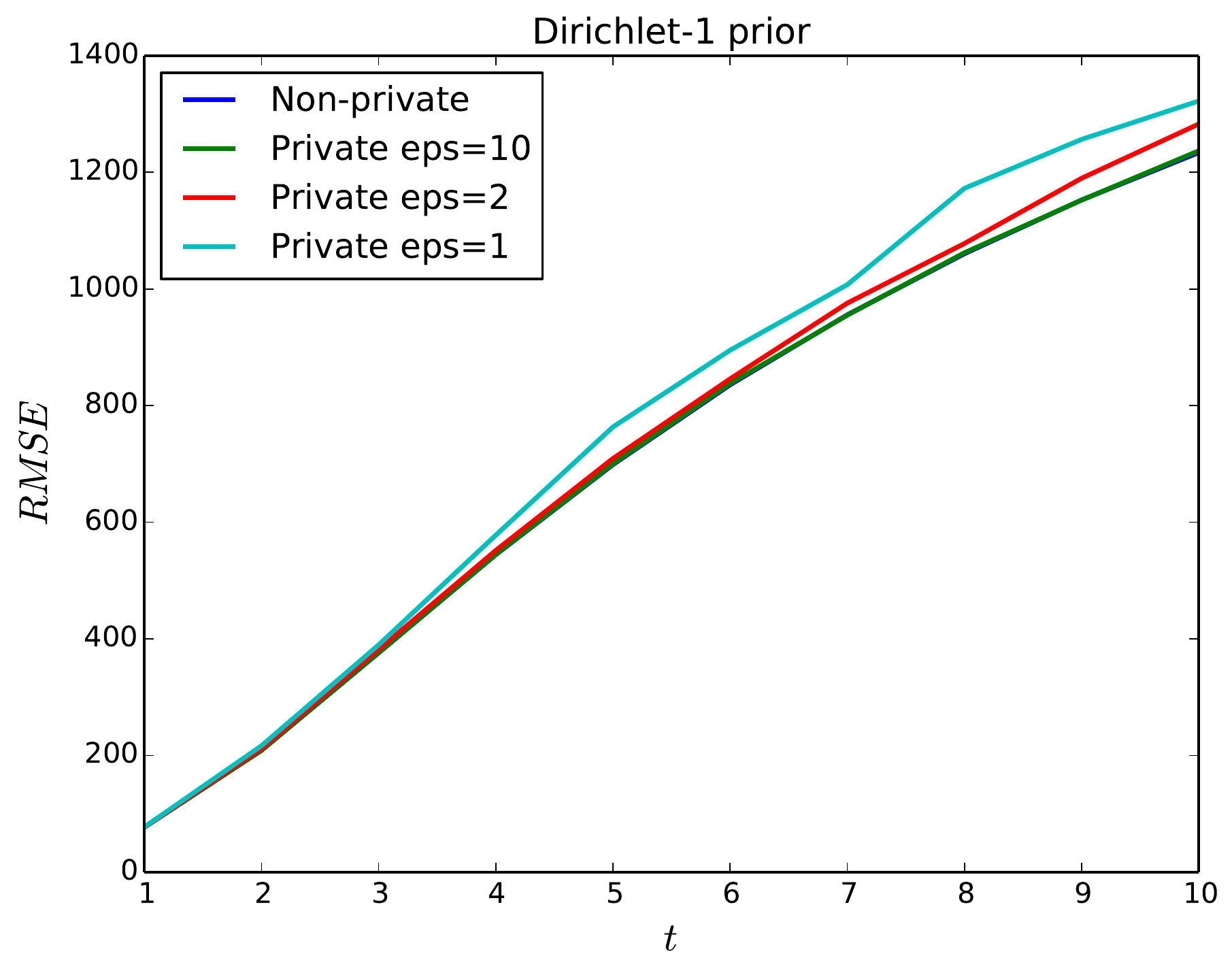}
}
\subfigure{
\includegraphics[width=0.18\textwidth]{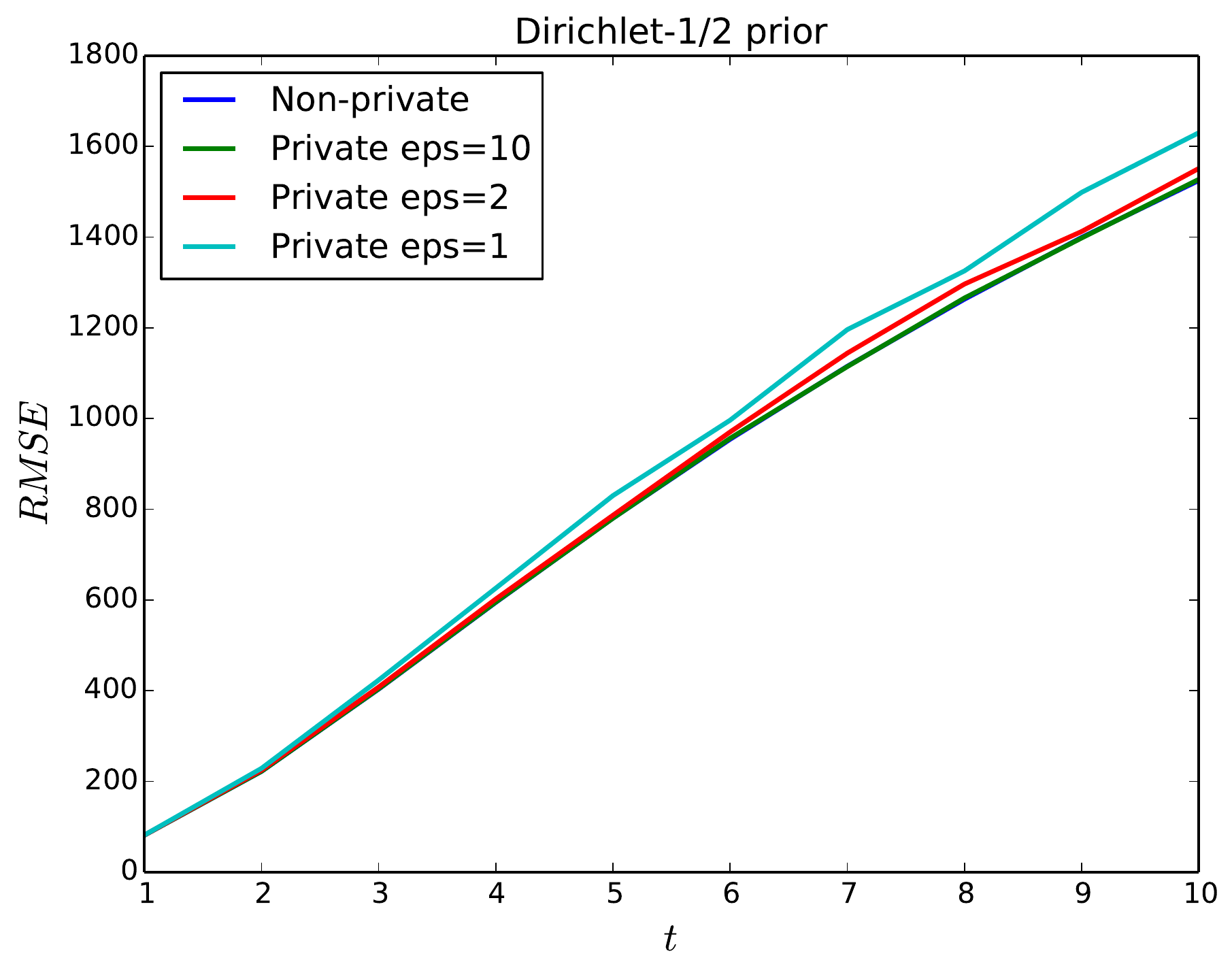}
}

\caption{Comparison between the private estimator with the non-private SGT when $k=20000$} 
\label{fig:synthetic_k20000} 

\end{figure*}

We observe that, in this setting, the cost of privacy is relatively small for reasonable values of $\eps$.
This is as predicted by our theoretical results, where unless $\eps$ is extremely small (less than $1/k$) the non-private sample complexity dominates the privacy requirement.
However, we found that for smaller support sizes (as shown in Section~\ref{sec:supp-exp-coverage}), the cost of privacy can be  significant.
We provide an intuitive explanation for why no private estimator can perform well on such instances.
To minimize the number of parameters, we instead argue about the related problem of support-size estimation.
Suppose we are trying to distinguish between distributions which are uniform over supports of size $100$ and $200$.
We note that, if we draw $n = 50$ samples, the ``profile'' of the samples (i.e., the histogram of the histogram) will be very similar for the two distributions.
In particular, if one modifies only a few samples (say, five or six), one could convert one profile into the other.
In other words, these two profiles are almost-neighboring datasets, but simultaneously correspond to very different support sizes.
This pits the two goals of privacy and accuracy at odds with each other, thus resulting in a degradation in accuracy.

\subsubsection{Evaluation on Census Data and Hamlet}
We conclude with experiments for support coverage on two real-world datasets, 
the 2000 US Census data and the text of Shakespeare's play Hamlet, inspired by investigations in~\cite{OrlitskySW16} and~\cite{ValiantV17b}.
Our investigation on US Census data is also inspired by the fact that this is a setting where privacy is of practical importance, evidenced by the proposed adoption of differential privacy in the 2020 US Census~\cite{DajaniLSKRMGDGKKLSSVA17}.

The Census dataset contains a list of last names that appear at least 100 times.
Since the dataset is so oversampled, even a small fraction of the data is likely to contain almost all the names.
As such, we make the task non-trivial by subsampling $m_{total} = 86080$ individuals from the data, obtaining $20412$ distinct last names.
We then sample $n$ of the $m_{total}$ individuals without replacement and attempt to estimate the total number of last names.
Figure~\ref{fig:name} displays the RMSE over 100 iterations of this process. We observe that even an exceptionally stringent privacy budget of $\eps = 0.5$, the performance is almost indistinguishable from the non-private SGT estimator.


\begin{figure}[h]
\centering 
\includegraphics [scale=0.3]{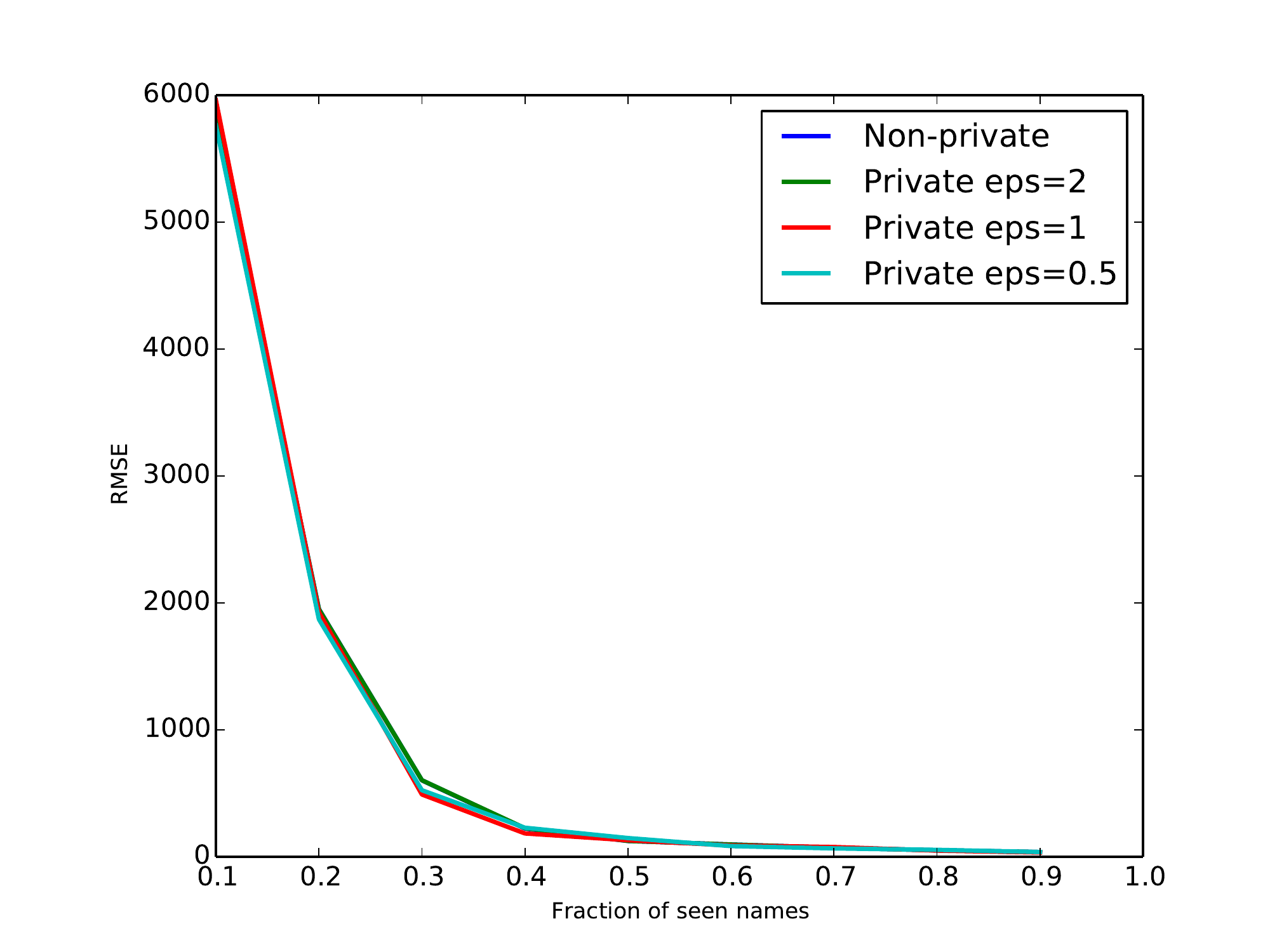}
\caption{Comparison between our private estimator with the SGT on Census Data} 
\label{fig:name} 
\end{figure}

The Hamlet dataset has $m_{total} = 31,999$ words, of which 4804 are distinct.
Since the distribution is not as oversampled as the Census data, we do not need to subsample the data.
Besides this difference, the experimental setup is identical to that of the Census dataset.
Once again, as we can see in Figure~\ref{fig:hamlet}, we get near-indistinguishable performance between the non-private and private estimators, even for very small values of $\eps$.
Our experimental results demonstrate that privacy is realizable in practice, with particularly accurate performance on real-world datasets. 

\begin{figure}[h]
\centering 
\includegraphics [scale=0.3]{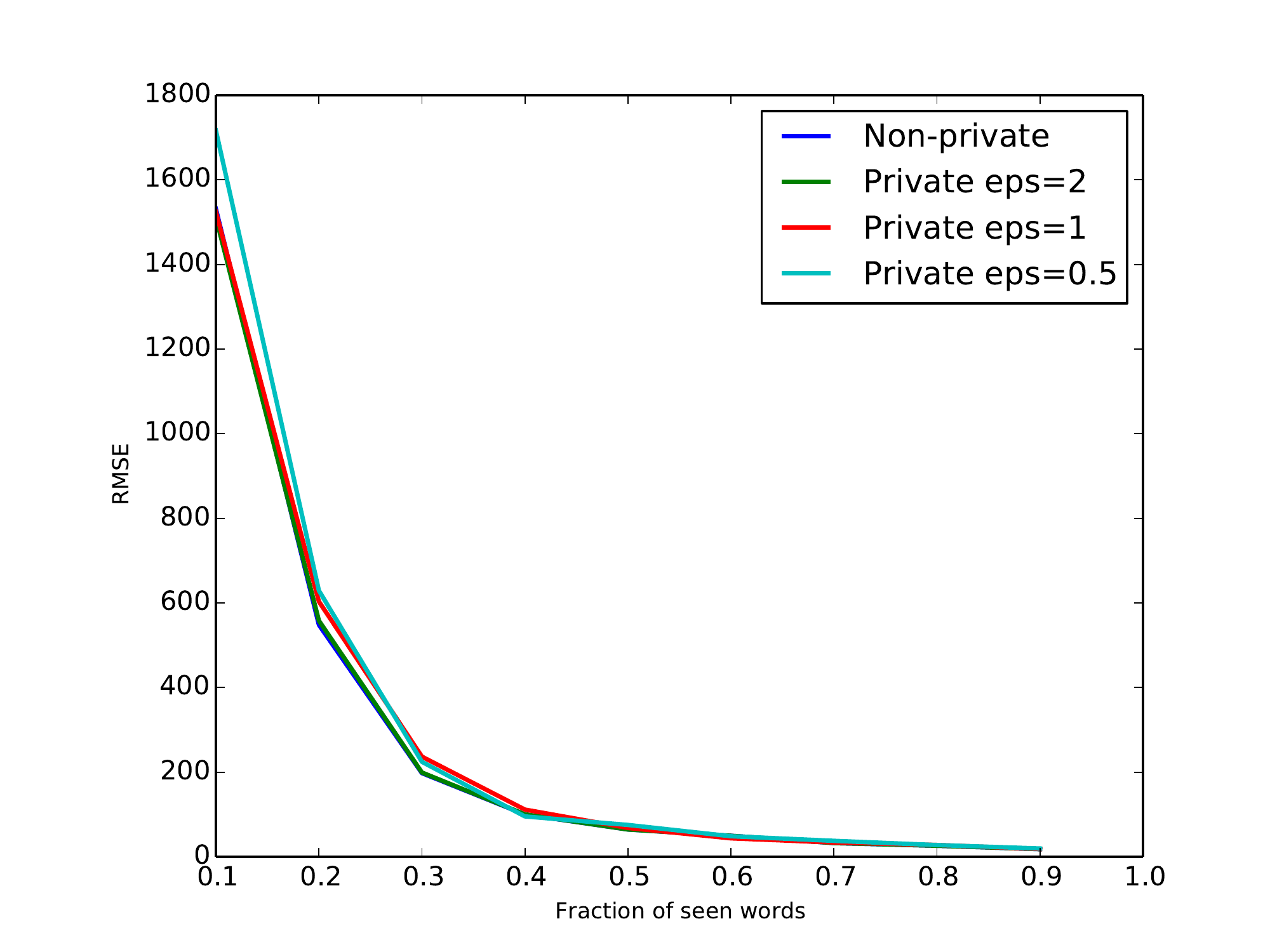}
\caption{Comparison between our private estimator with the SGT on Hamlet} 
\label{fig:hamlet} 
\end{figure}

\bibliographystyle{alpha}
\bibliography{biblio}
\appendix
\section{Additional Experimental Results}
\label{sec:supp-experiments}
This section contains additional plots of our synthetic experimental results.
Section~\ref{sec:supp-exp-entropy} contains experiments on entropy estimation, while Section~\ref{sec:supp-exp-coverage} contains experiments on estimation of support coverage.
\subsection{Entropy Estimation}
\label{sec:supp-exp-entropy}
We present four more plots of our synthetic experimental results for entropy estimation.
Figures~\ref{fig:entropy-k100-eps1} and~\ref{fig:entropy-k100-eps2} are on a smaller support of $k =100$, with $\eps = 1$ and $2$, respectively.
Figures~\ref{fig:entropy-k1000-eps05} and~\ref{fig:entropy-k1000-eps2} are on a support of $k=1000$, with $\eps = 0.5$ and $2$.
\begin{figure*}
\centering
\subfigure[]{
\begin{minipage}[b]{0.3\textwidth}
\includegraphics[width=1\textwidth]{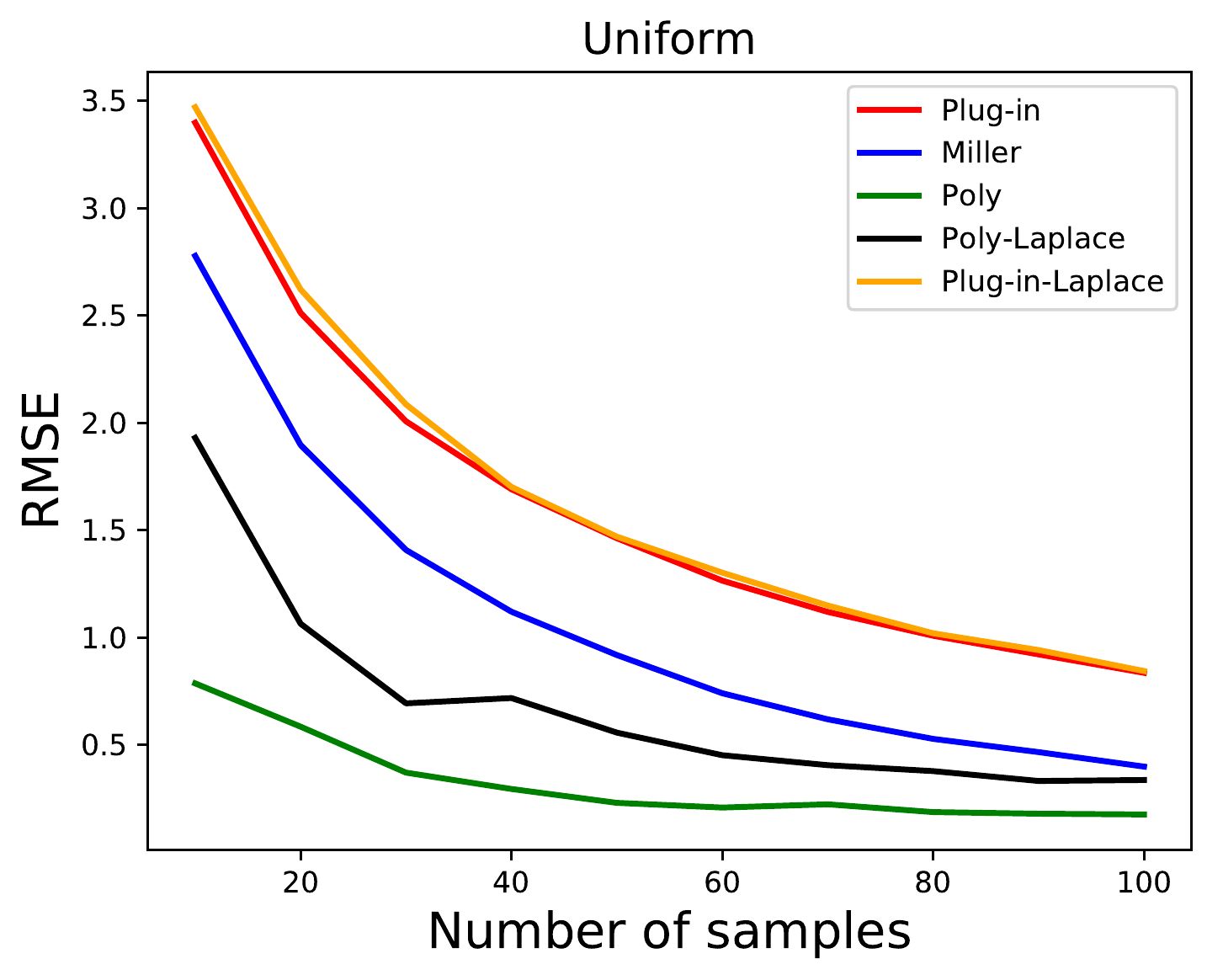}
\includegraphics[width=1\textwidth]{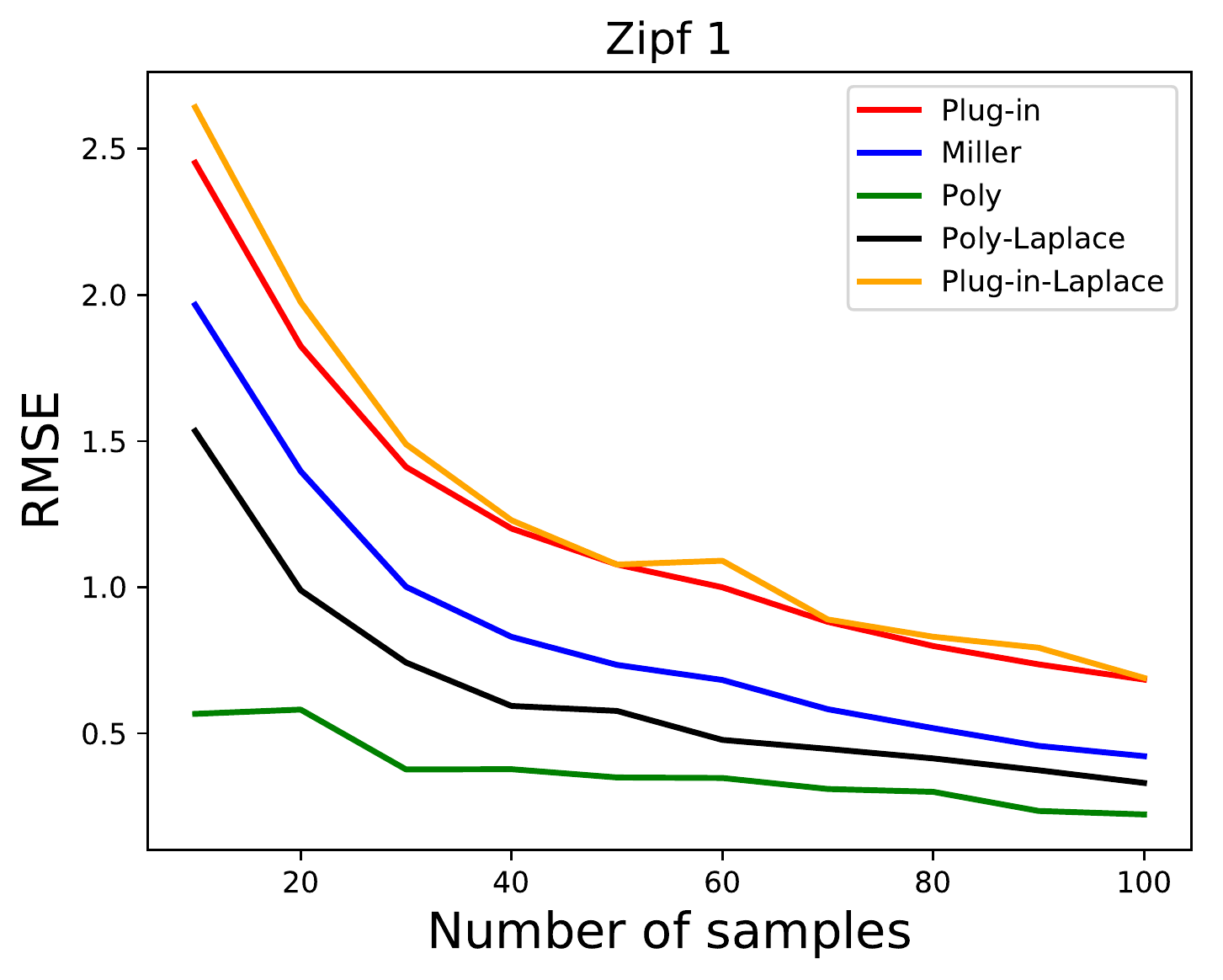}
\end{minipage}
}
\subfigure[]{
\begin{minipage}[b]{0.3\textwidth}
\includegraphics[width=1\textwidth]{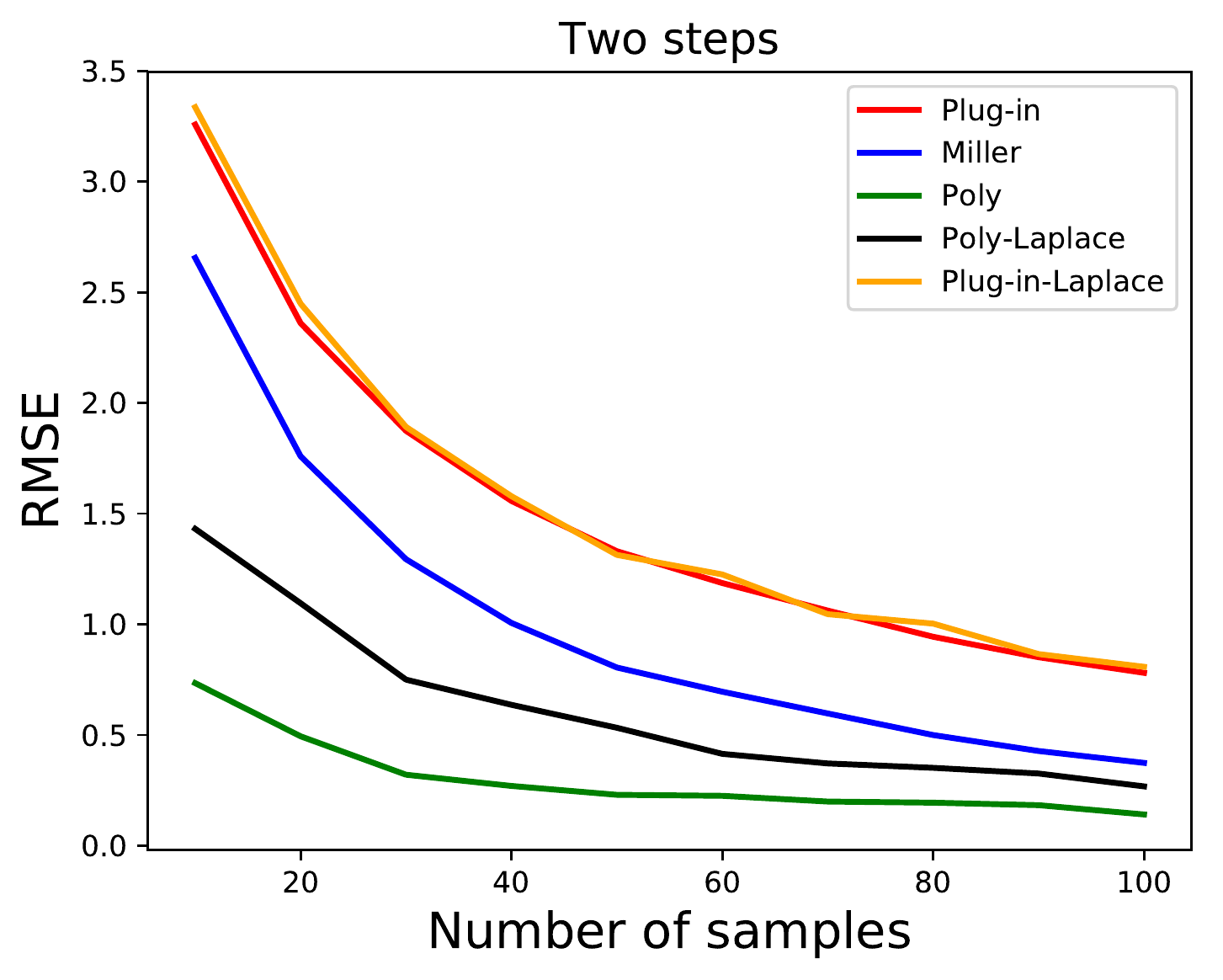}
\includegraphics[width=1\textwidth]{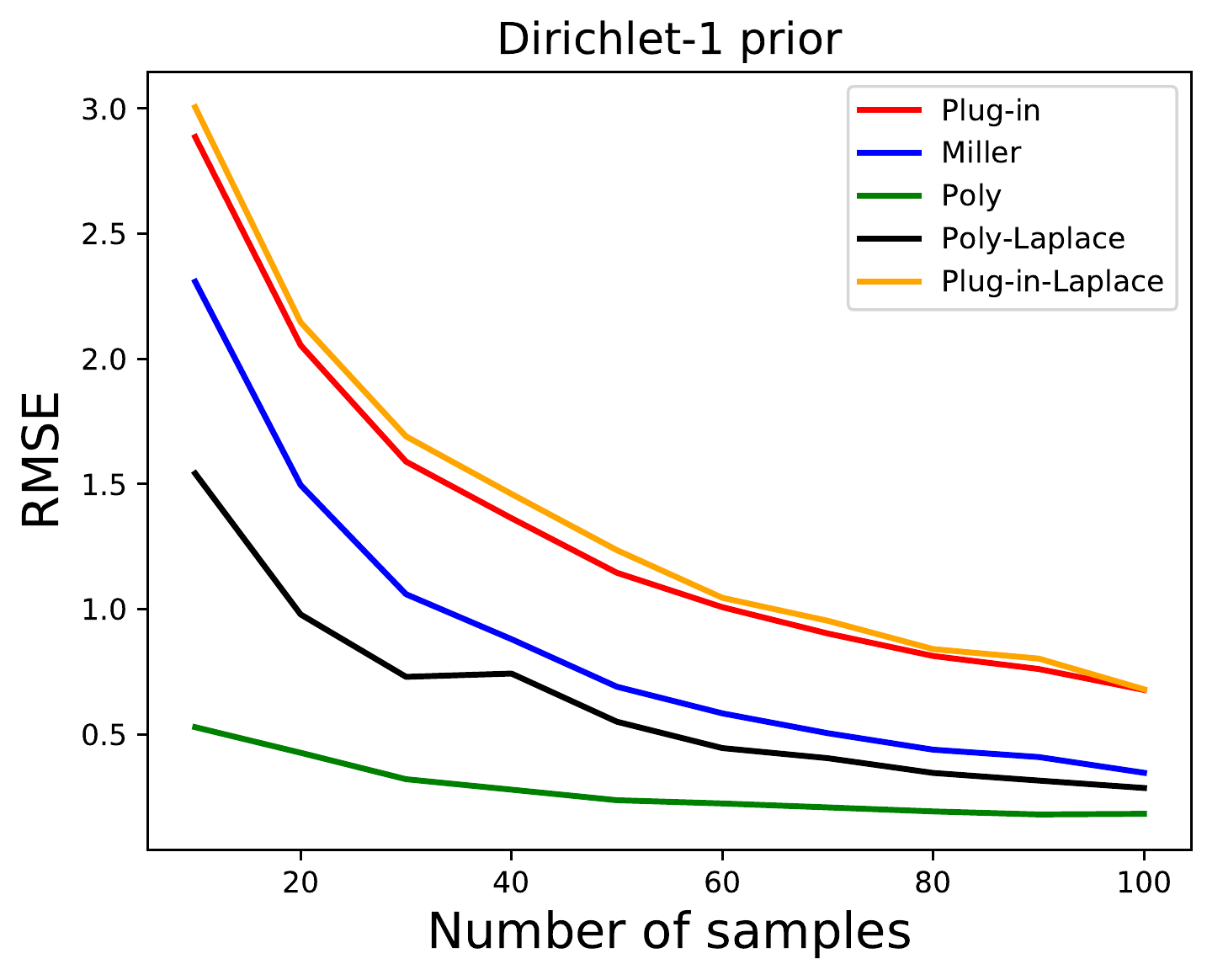}
\end{minipage}
}
\subfigure[]{
\begin{minipage}[b]{0.3\textwidth}
\includegraphics[width=1\textwidth]{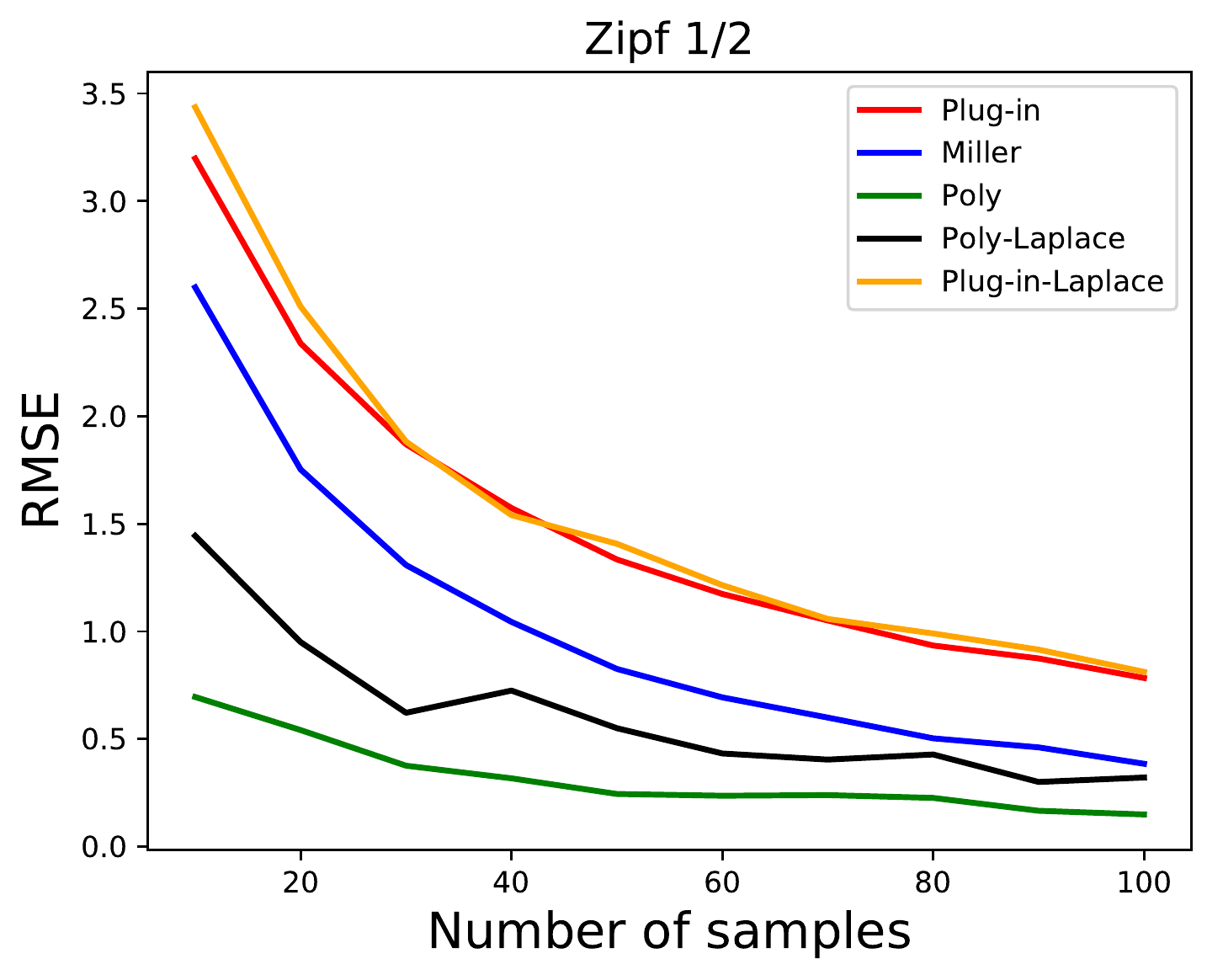}
\includegraphics[width=1\textwidth]{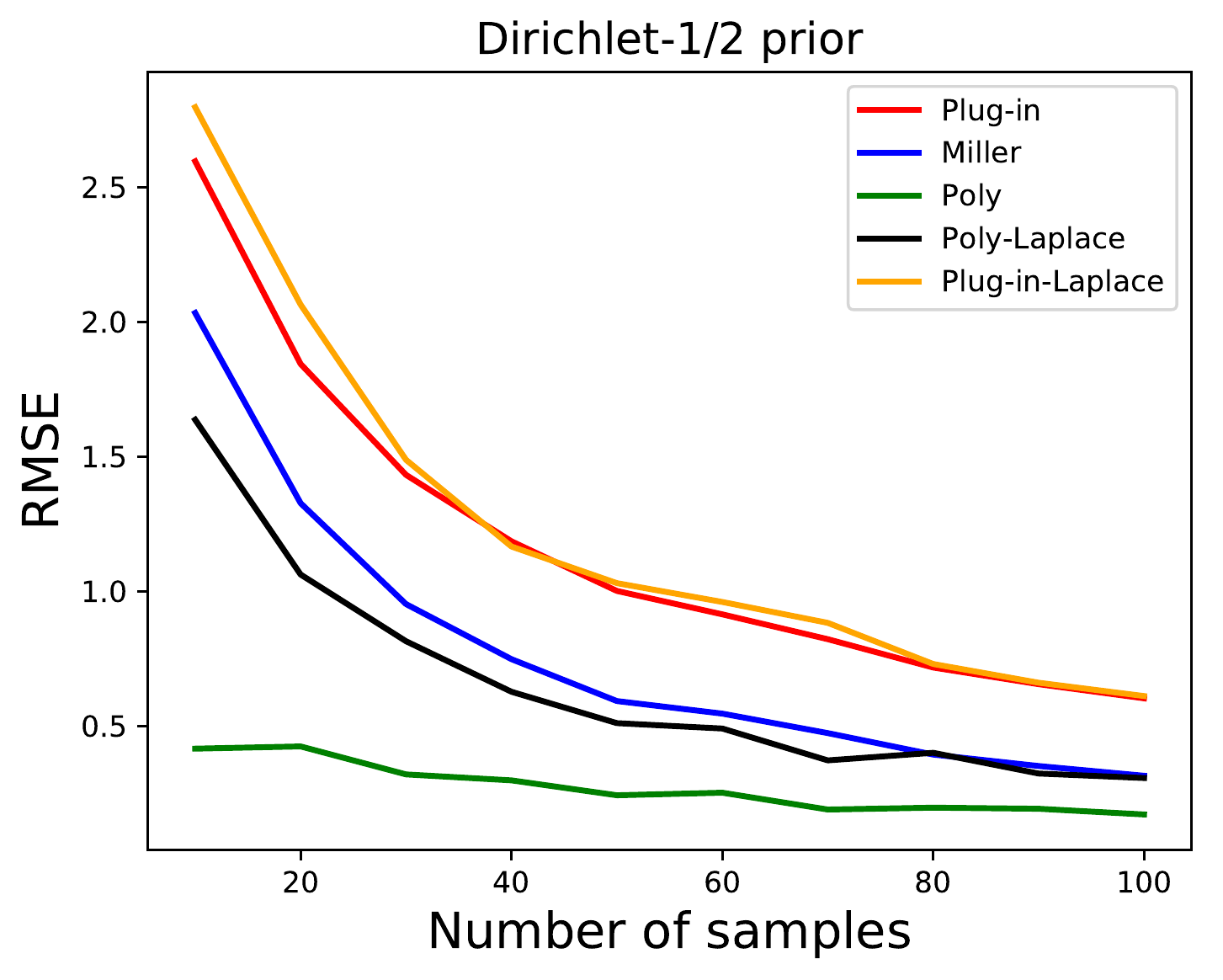}
\end{minipage}
}
\caption{Comparison of various estimators for the entropy, $k=100$, $\eps =1$.} 
\label{fig:entropy-k100-eps1}
\end{figure*}
\begin{figure*}
\centering
\subfigure[]{
\begin{minipage}[b]{0.3\textwidth}
\includegraphics[width=1\textwidth]{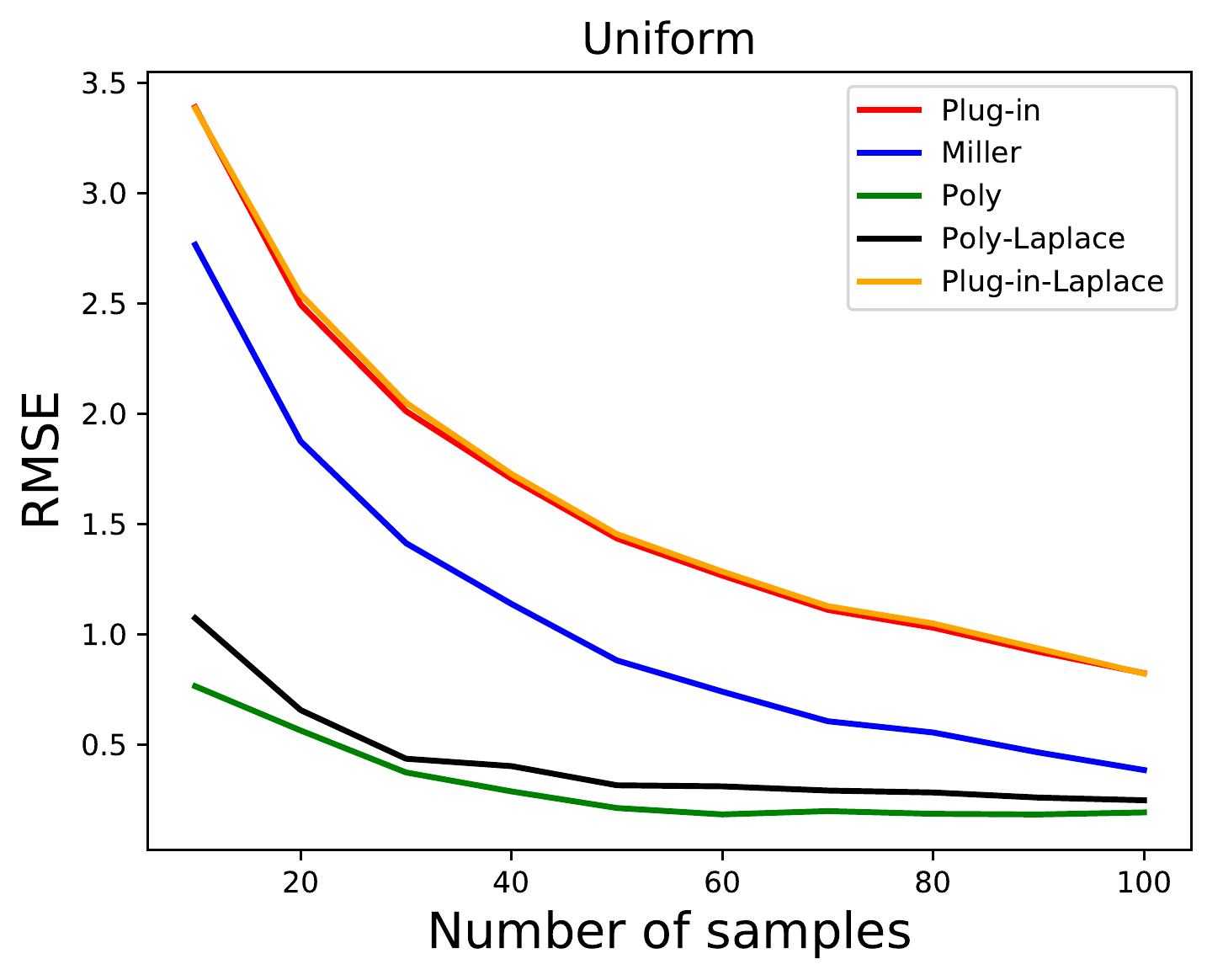}
\includegraphics[width=1\textwidth]{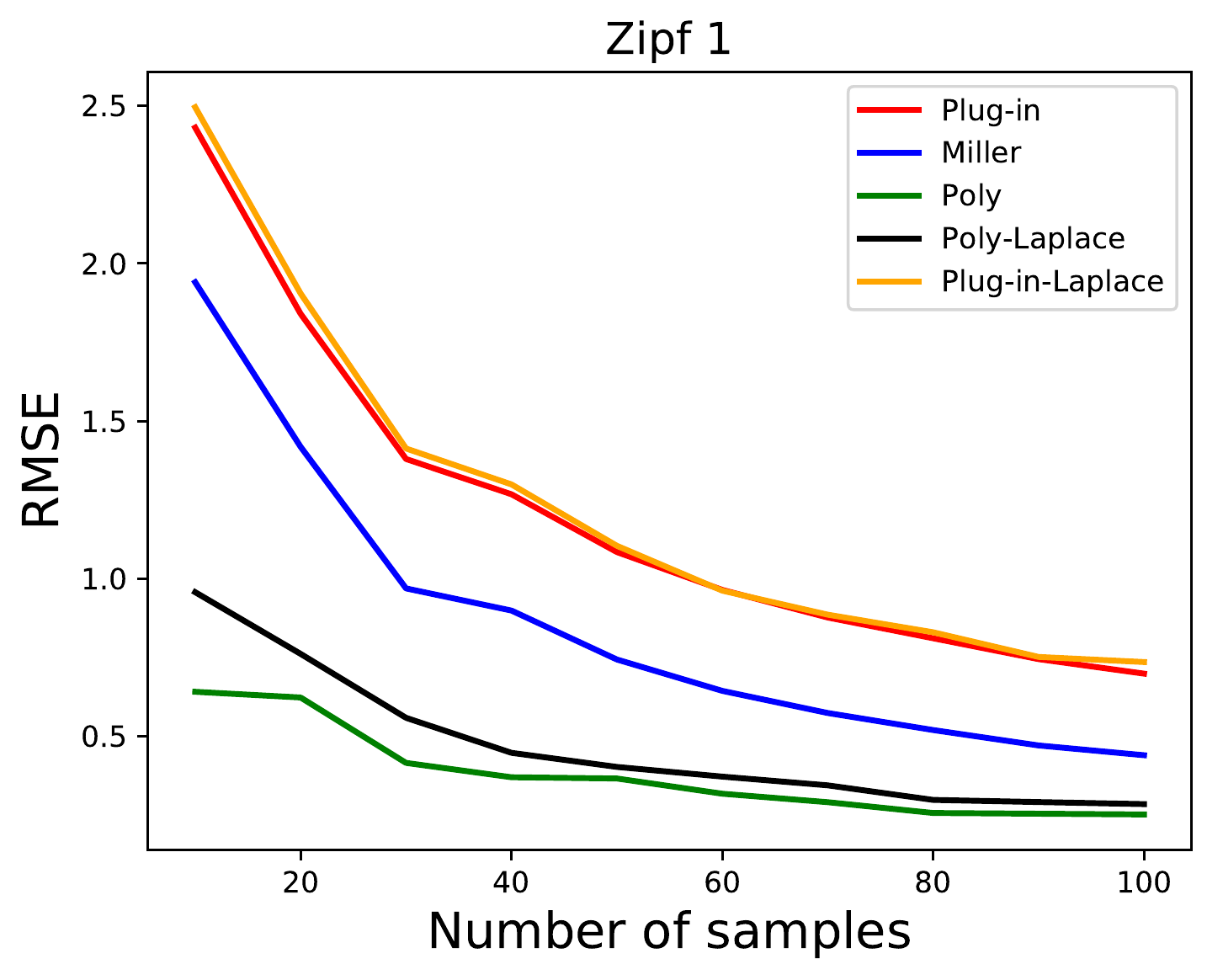}
\end{minipage}
}
\subfigure[]{
\begin{minipage}[b]{0.3\textwidth}
\includegraphics[width=1\textwidth]{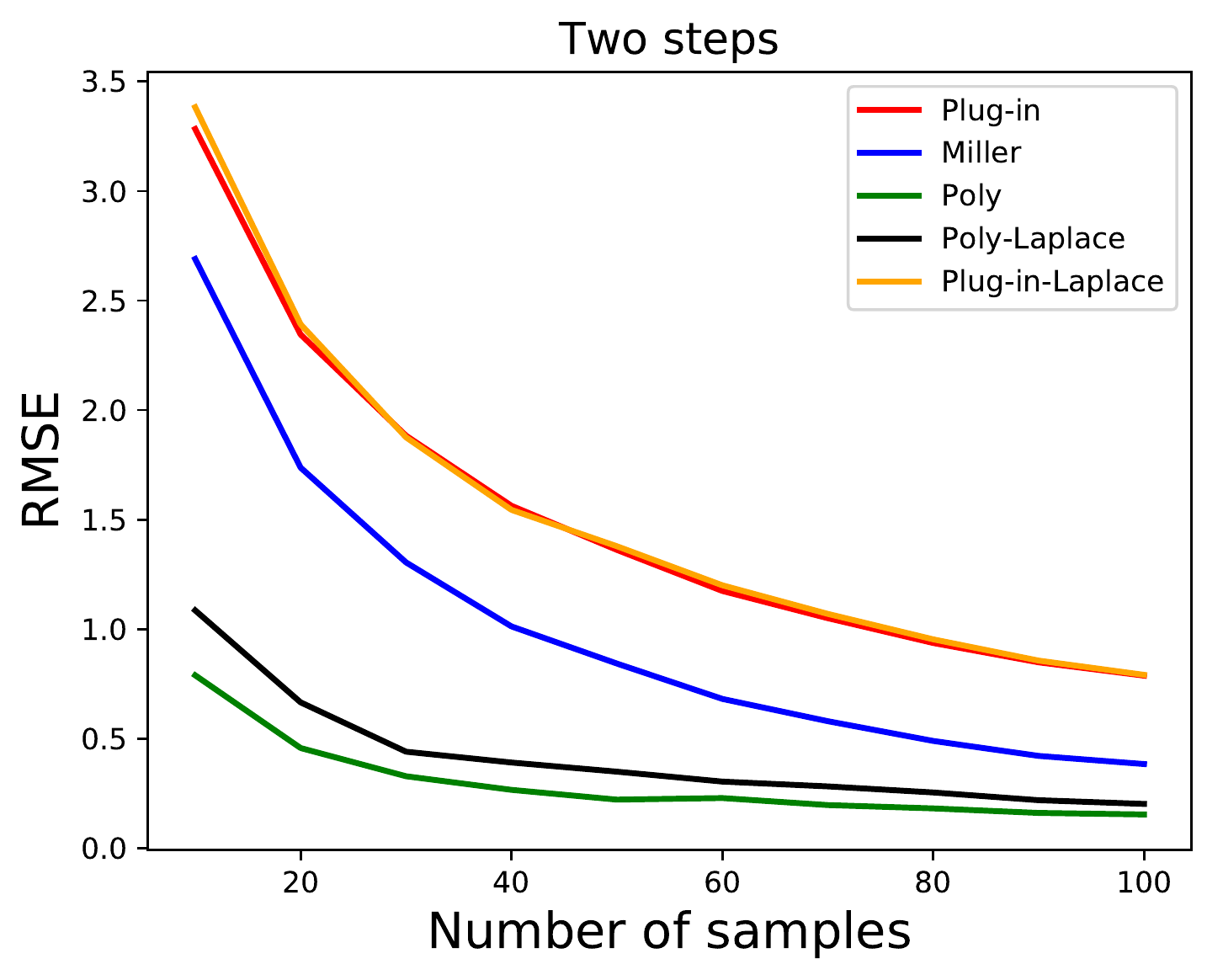}
\includegraphics[width=1\textwidth]{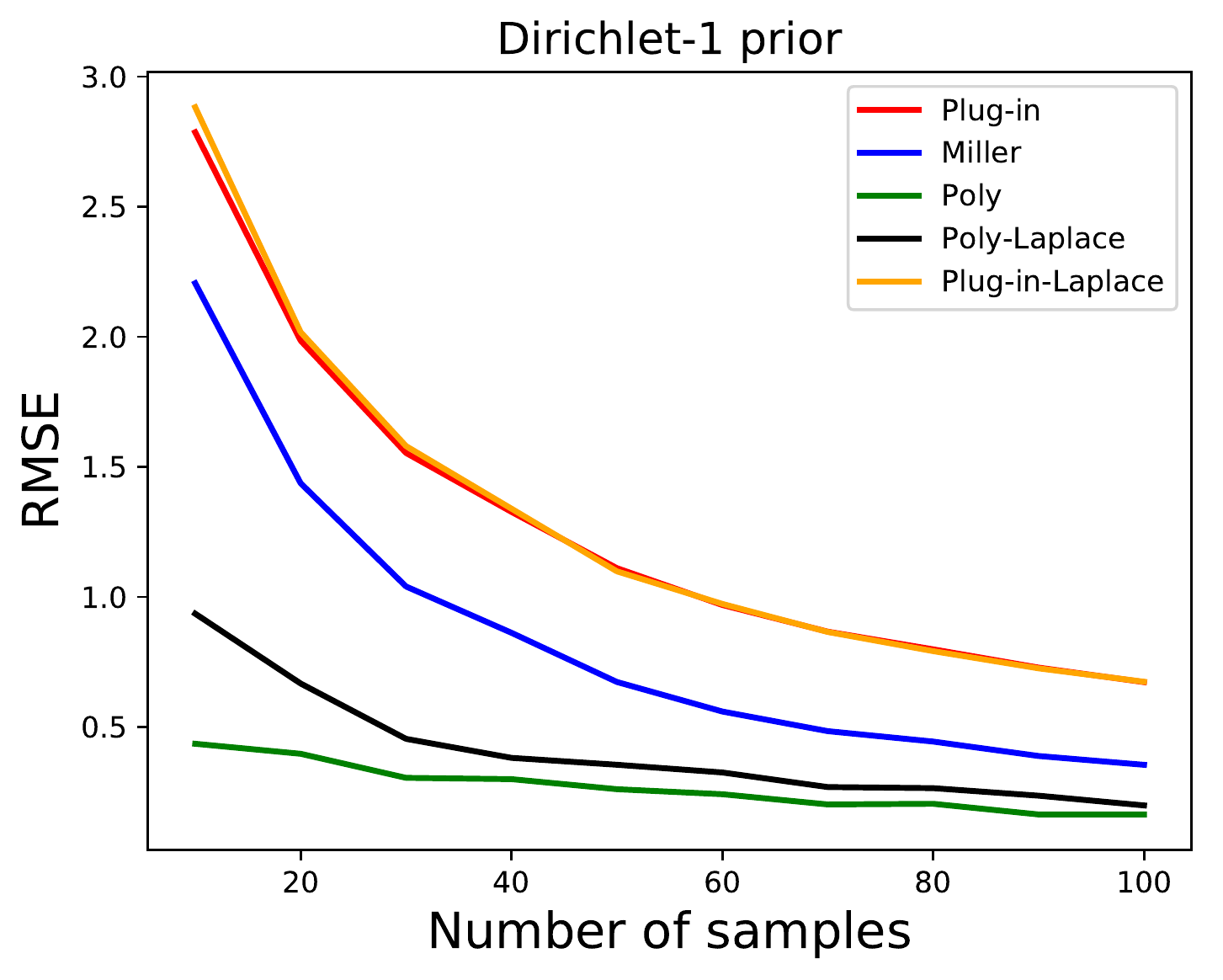}
\end{minipage}
}
\subfigure[]{
\begin{minipage}[b]{0.3\textwidth}
\includegraphics[width=1\textwidth]{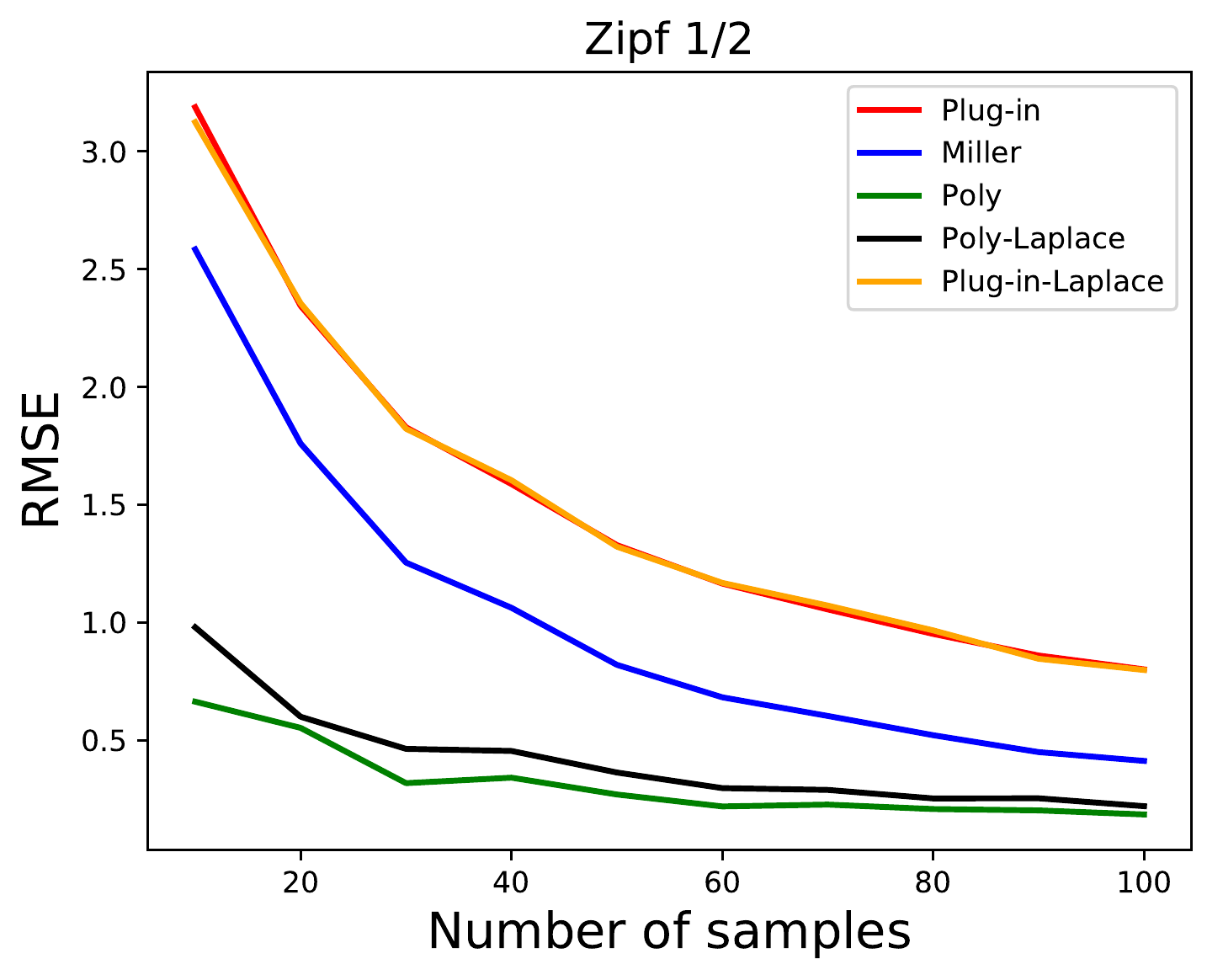}
\includegraphics[width=1\textwidth]{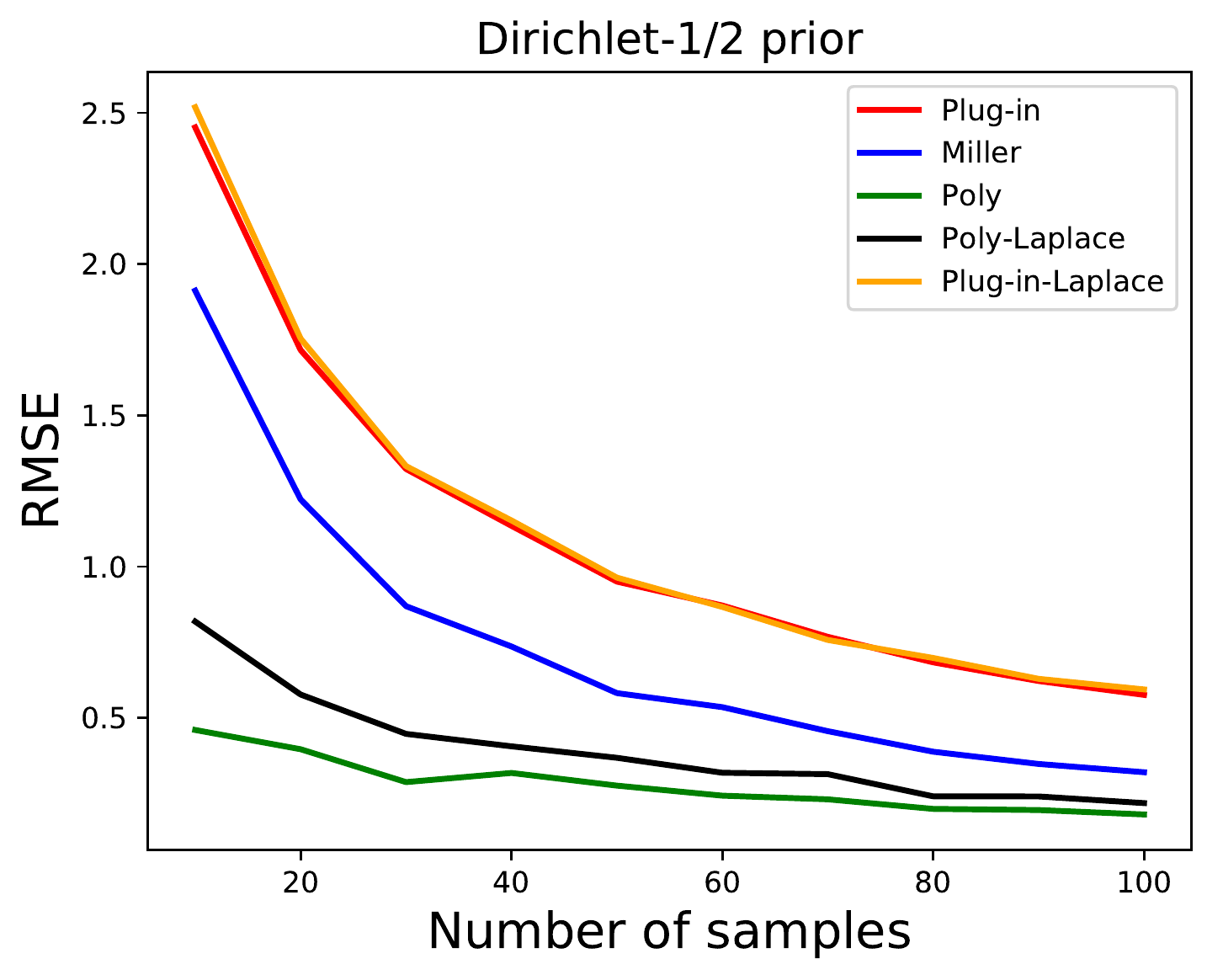}
\end{minipage}
}
\caption{Comparison of various estimators for the entropy, $k=100$, $\eps =2$.} 
\label{fig:entropy-k100-eps2}
\end{figure*}
\begin{figure*}
\centering
\subfigure[]{
\begin{minipage}[b]{0.3\textwidth}
\includegraphics[width=1\textwidth]{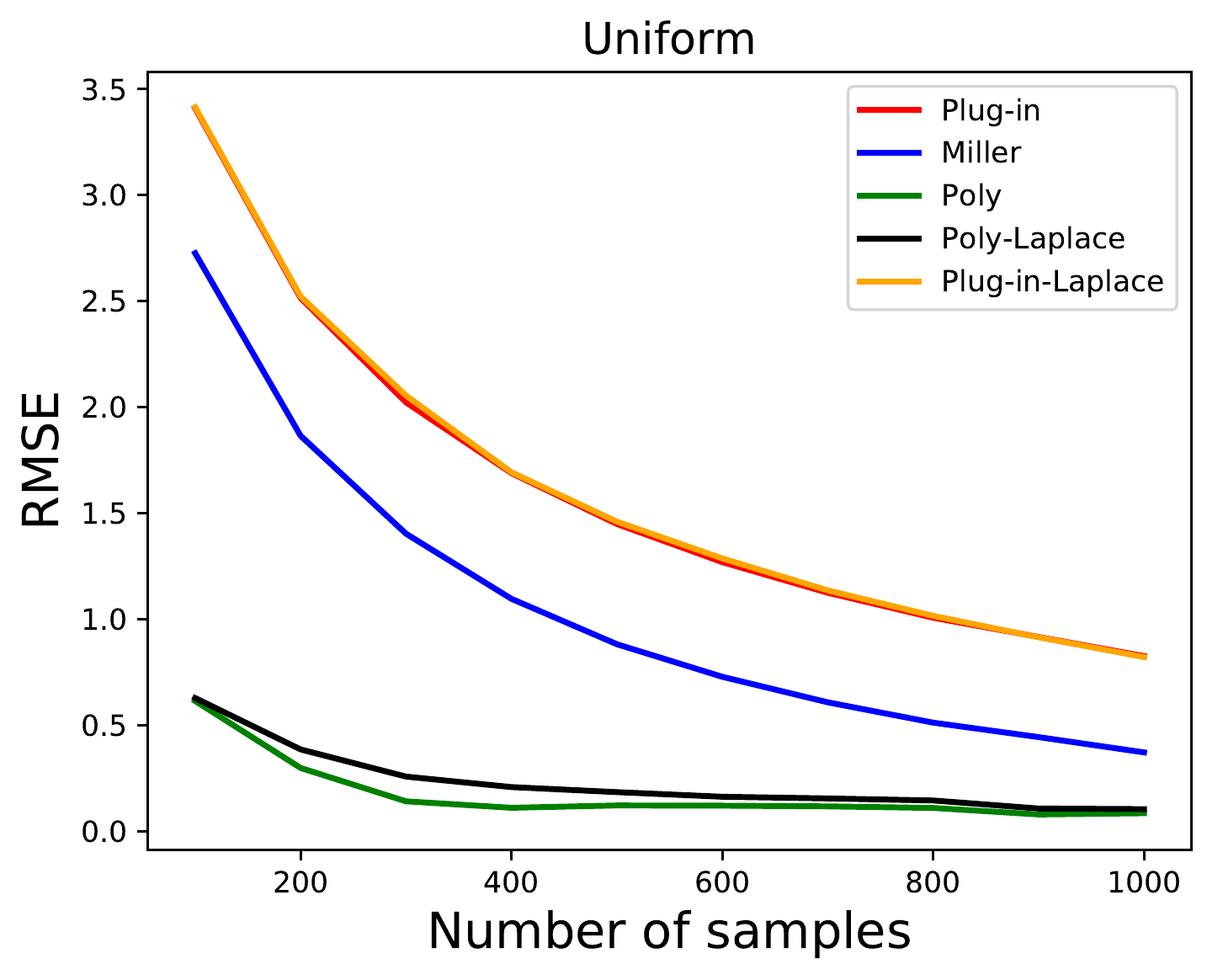}
\includegraphics[width=1\textwidth]{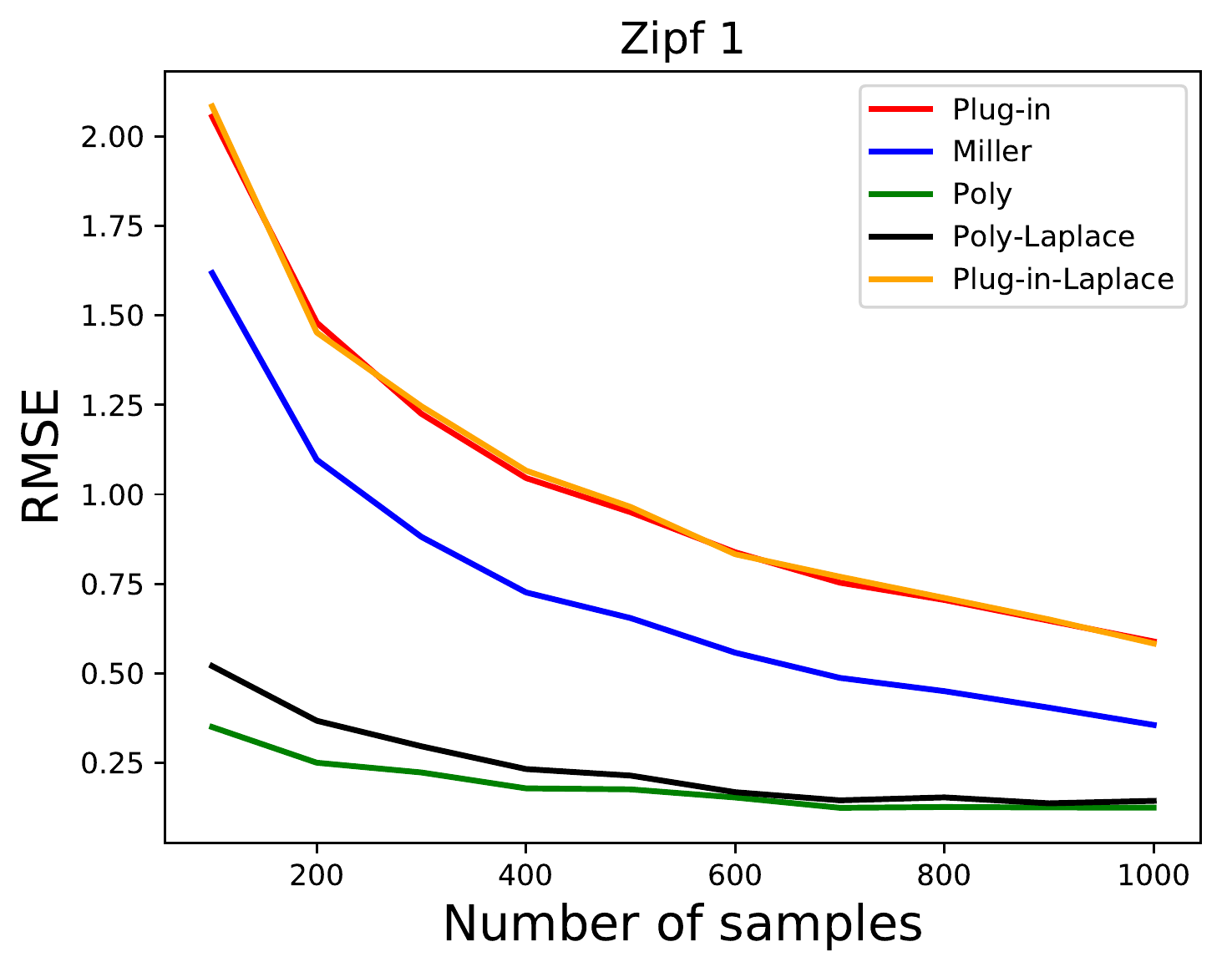}
\end{minipage}
}
\subfigure[]{
\begin{minipage}[b]{0.3\textwidth}
\includegraphics[width=1\textwidth]{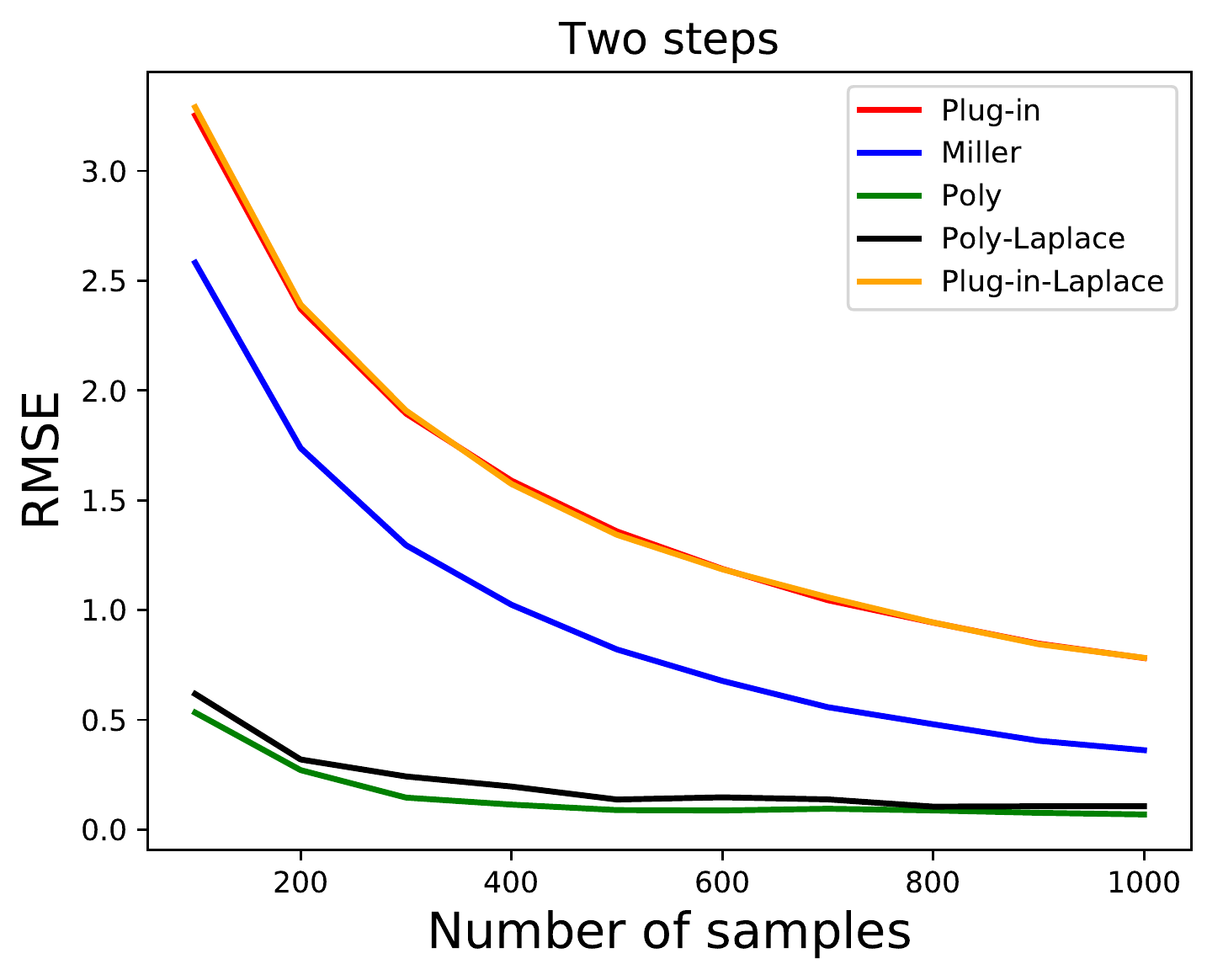}
\includegraphics[width=1\textwidth]{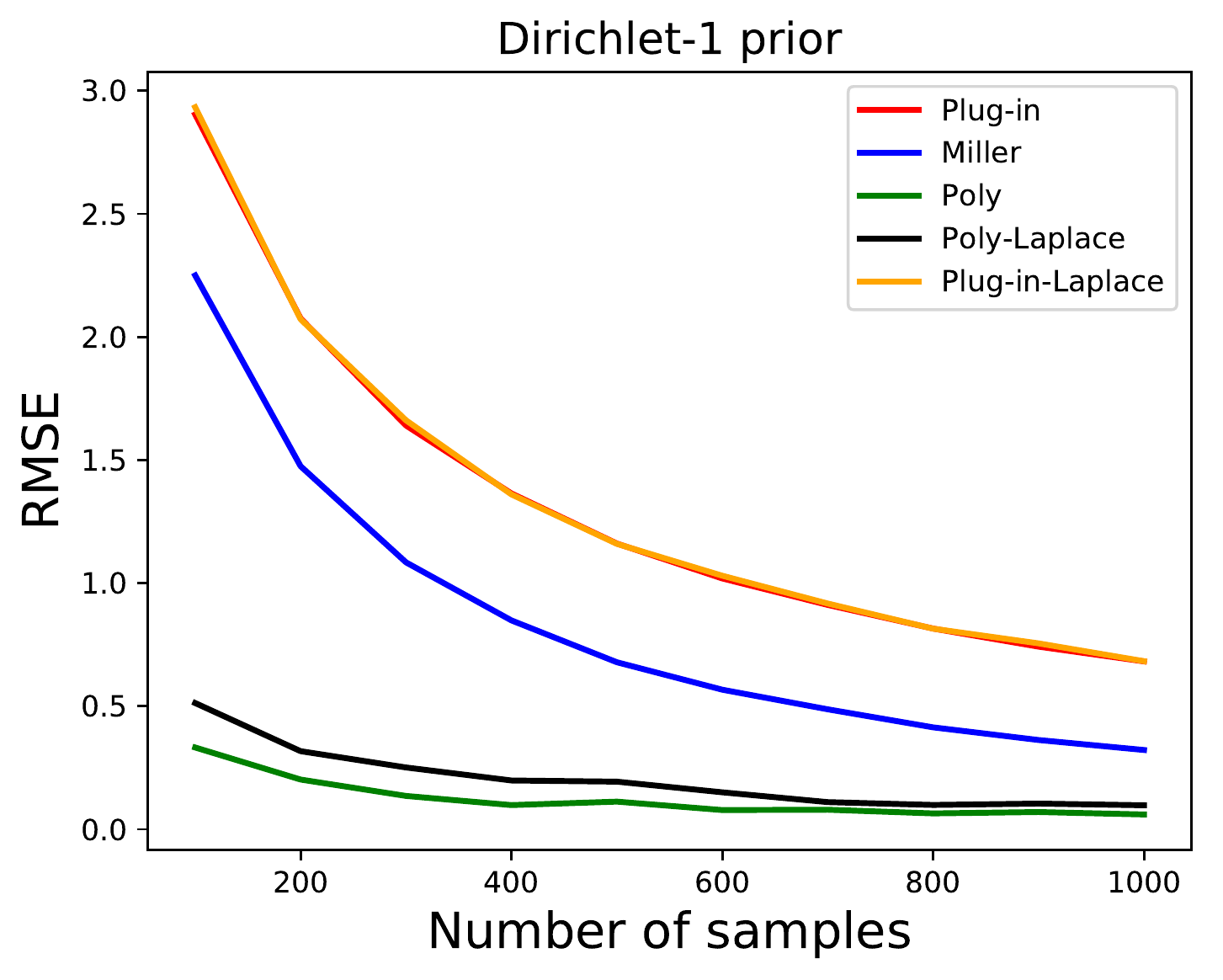}
\end{minipage}
}
\subfigure[]{
\begin{minipage}[b]{0.3\textwidth}
\includegraphics[width=1\textwidth]{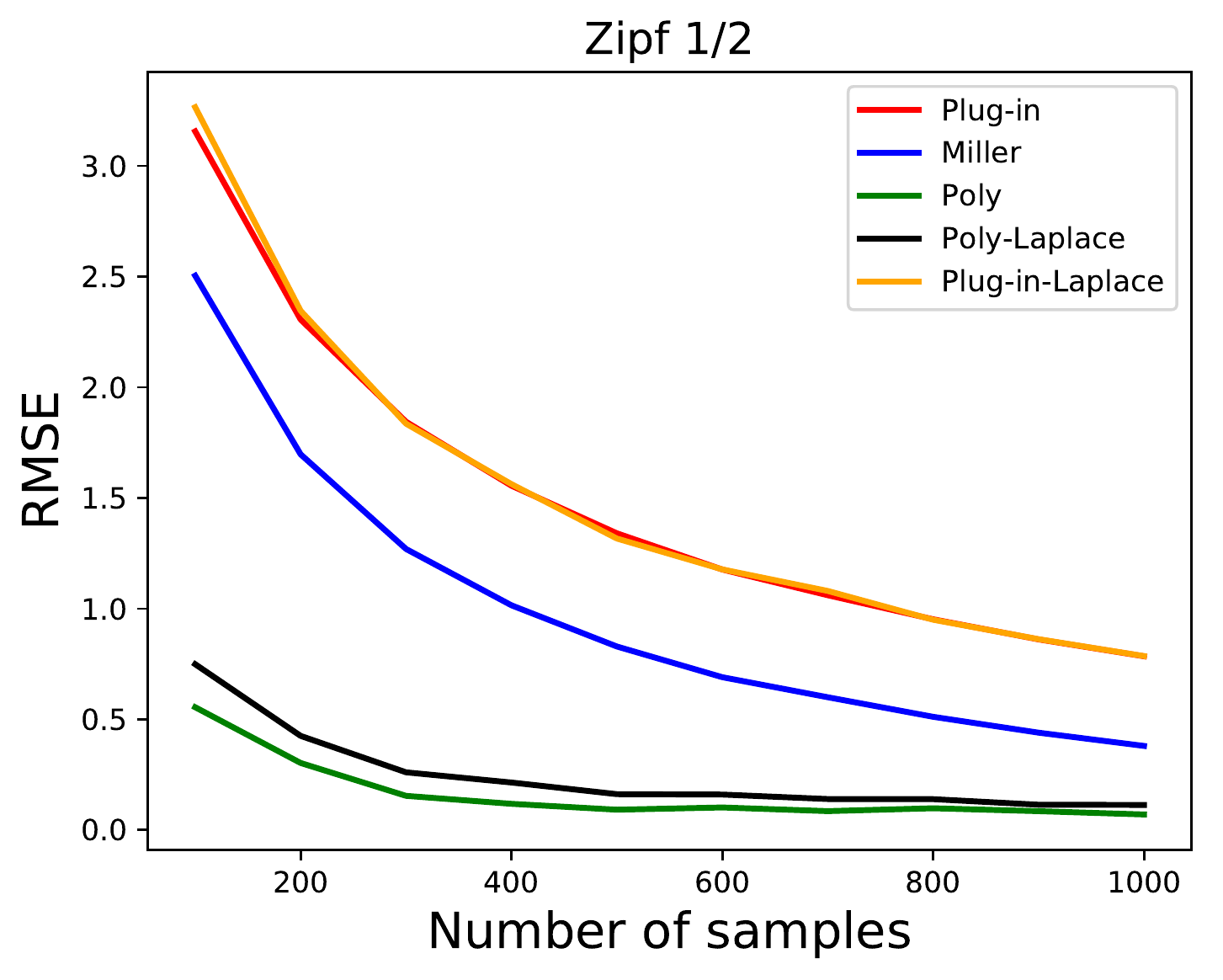}
\includegraphics[width=1\textwidth]{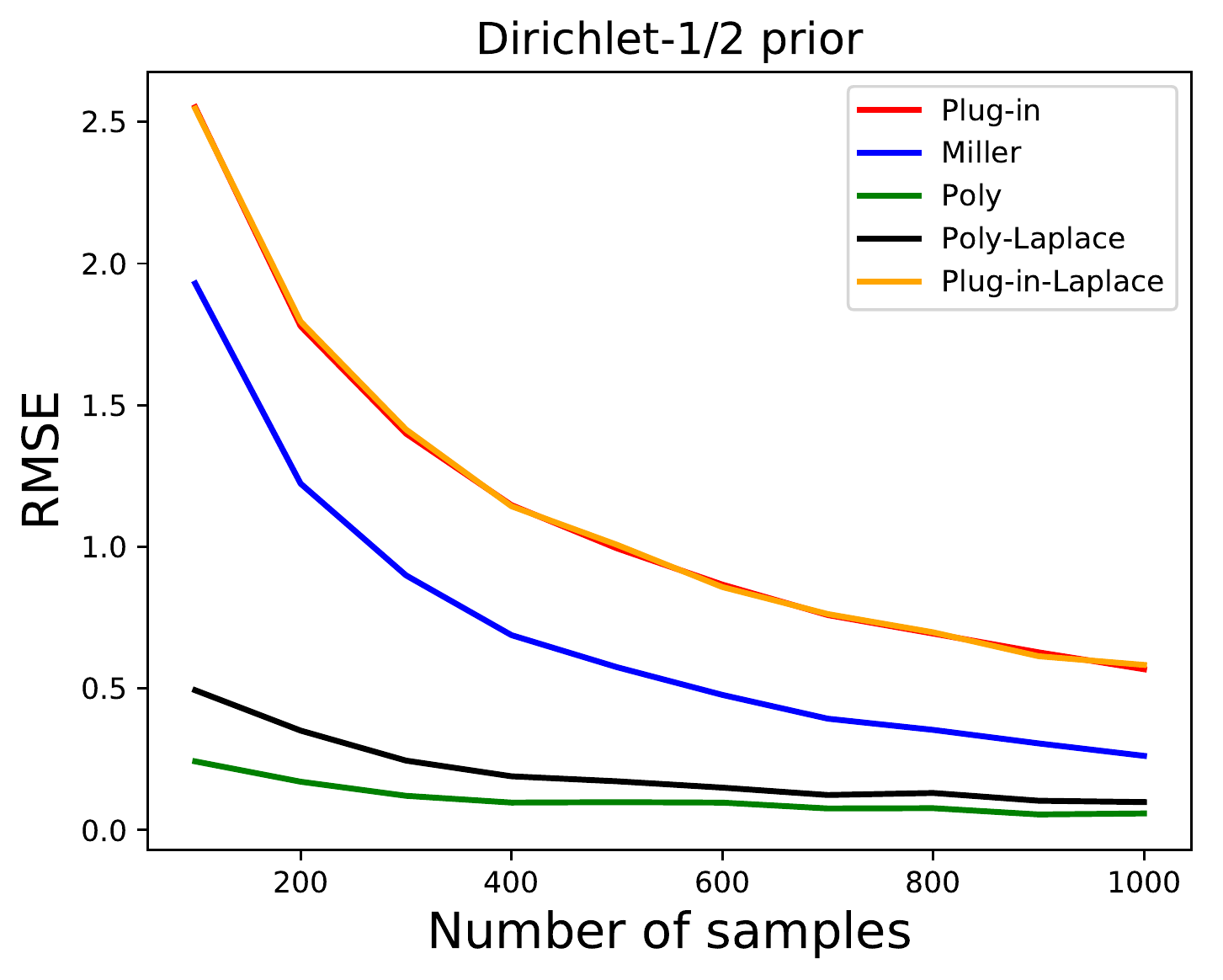}
\end{minipage}
}
\caption{Comparison of various estimators for the entropy, $k=1000$, $\eps =0.5$.} 
\label{fig:entropy-k1000-eps05}
\end{figure*}
\begin{figure*}
\centering
\subfigure[]{
\begin{minipage}[b]{0.3\textwidth}
\includegraphics[width=1\textwidth]{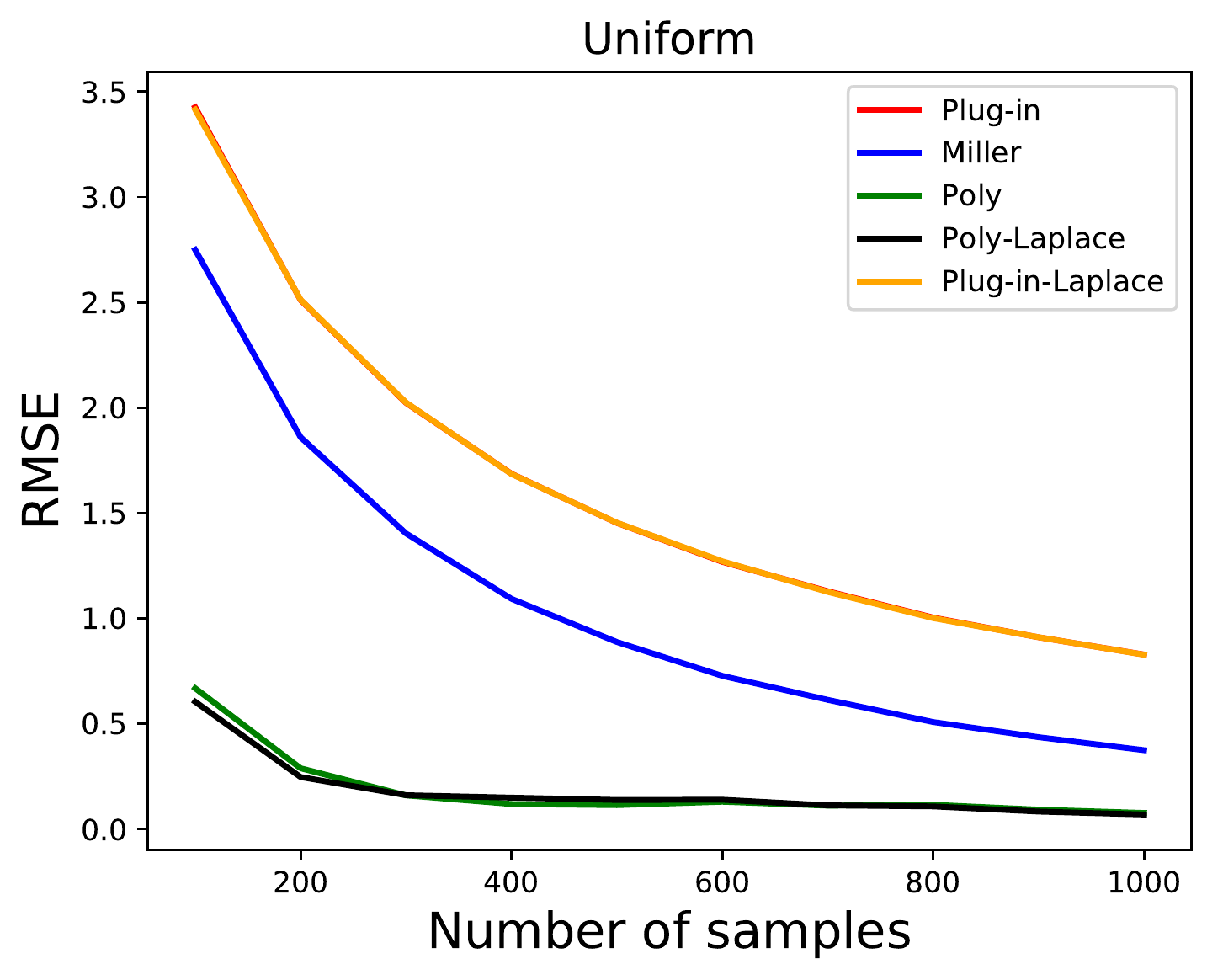}
\includegraphics[width=1\textwidth]{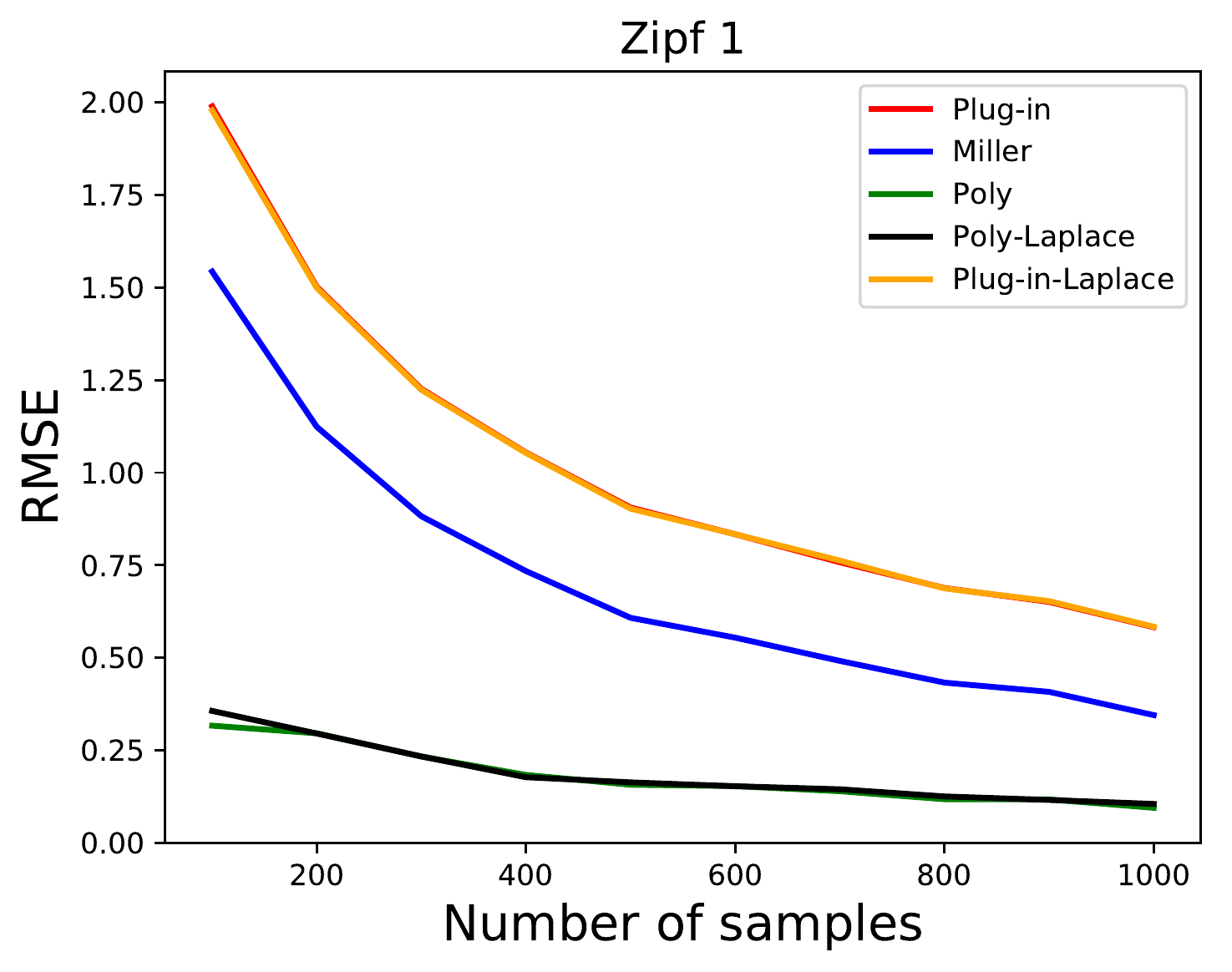}
\end{minipage}
}
\subfigure[]{
\begin{minipage}[b]{0.3\textwidth}
\includegraphics[width=1\textwidth]{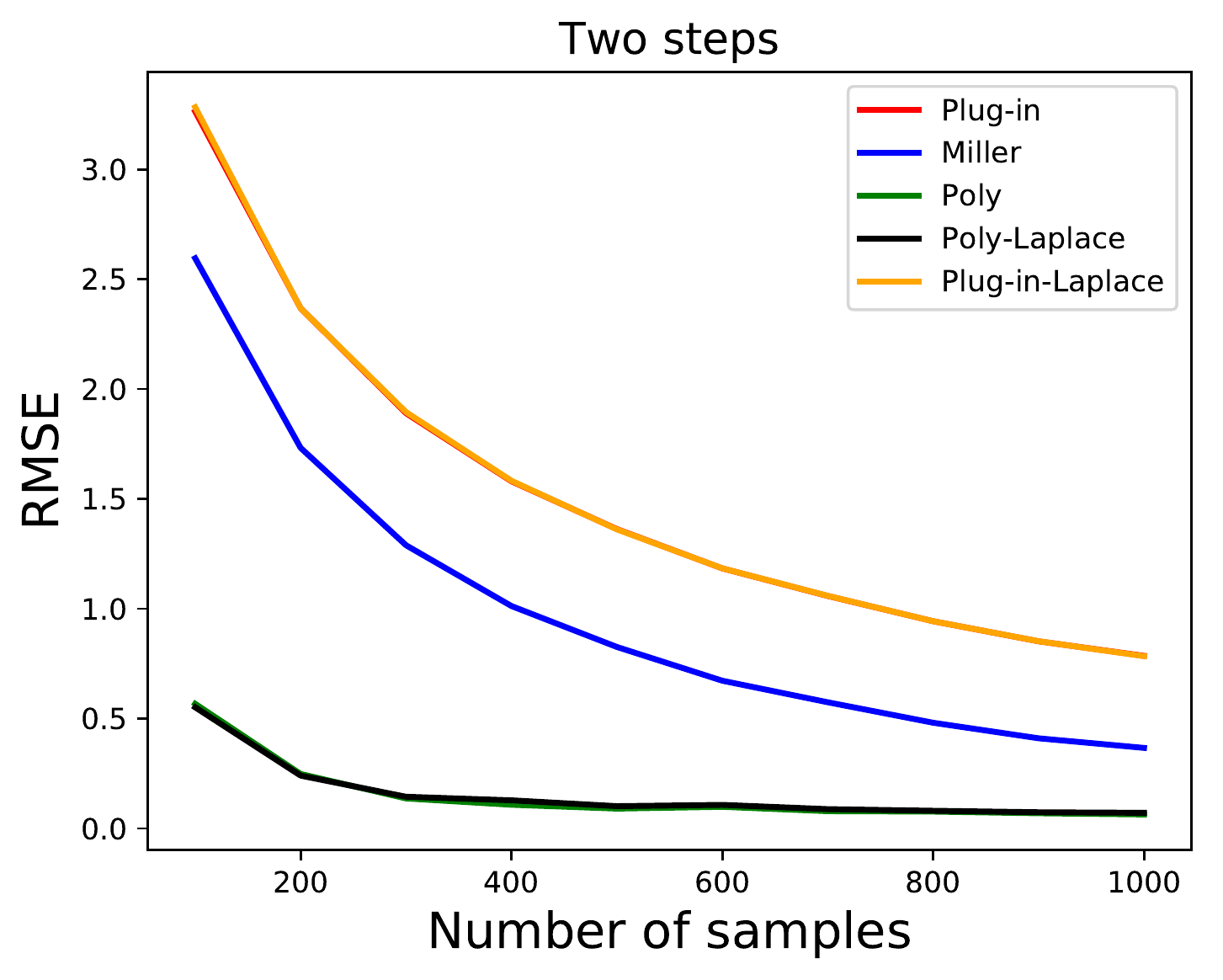}
\includegraphics[width=1\textwidth]{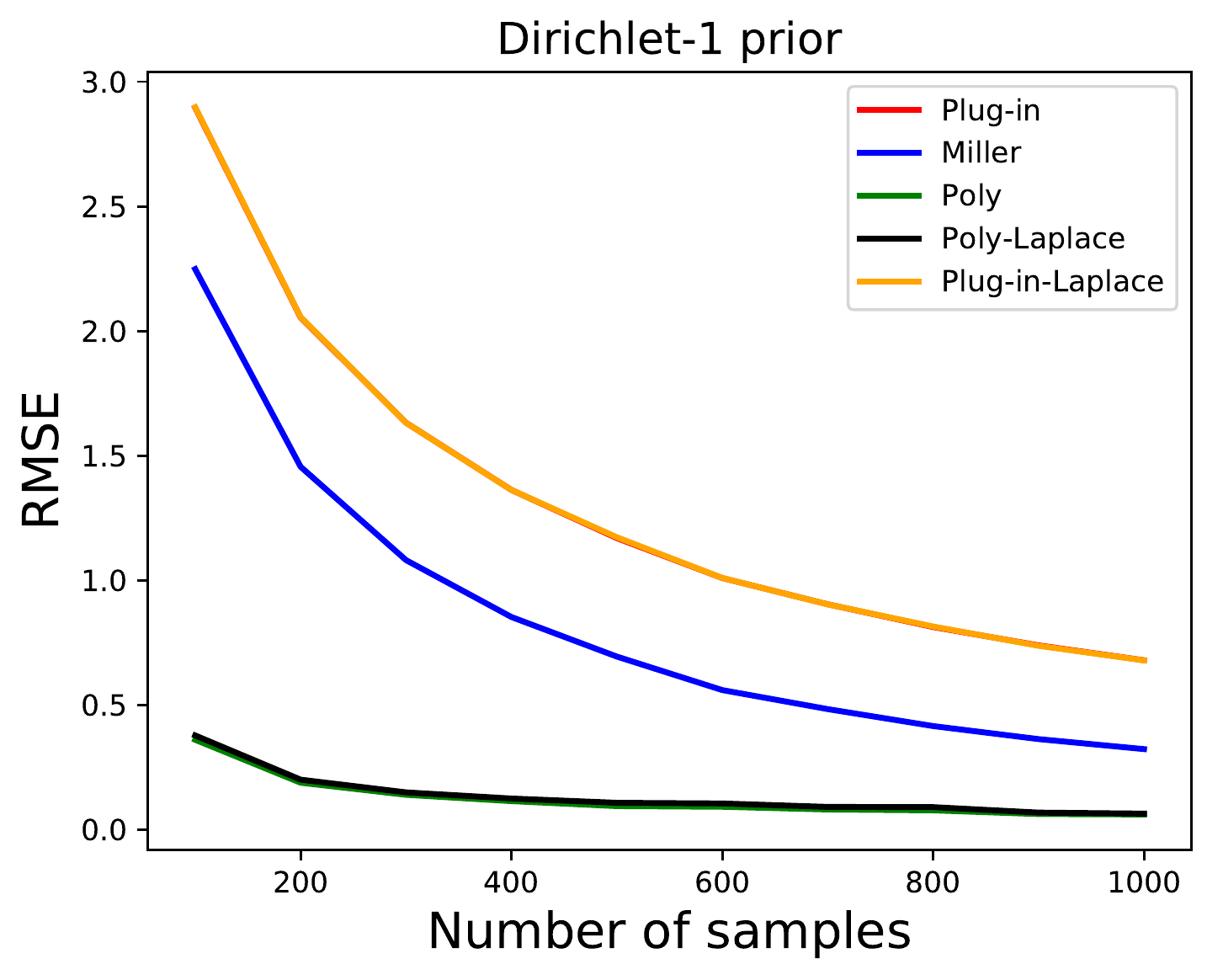}
\end{minipage}
}
\subfigure[]{
\begin{minipage}[b]{0.3\textwidth}
\includegraphics[width=1\textwidth]{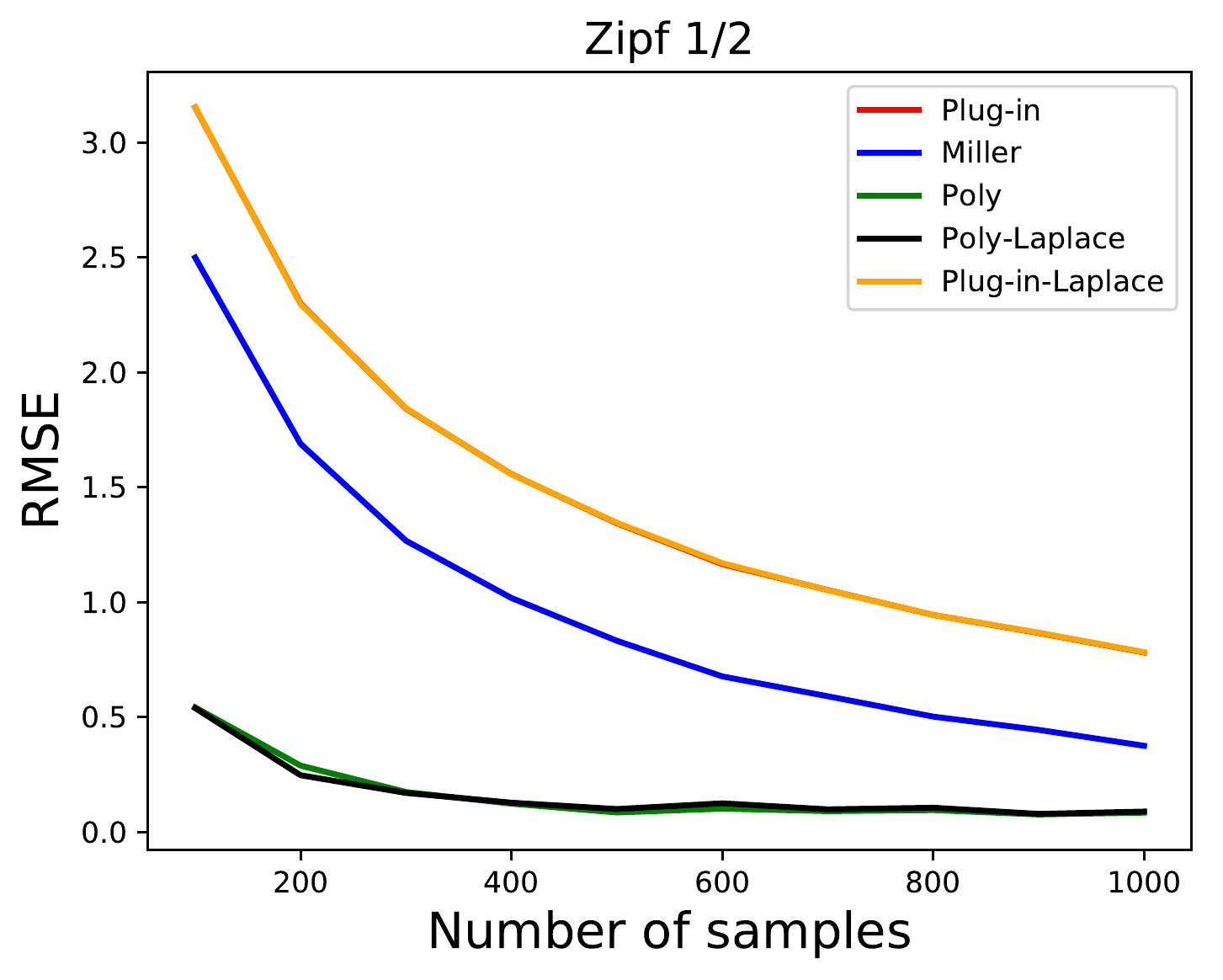}
\includegraphics[width=1\textwidth]{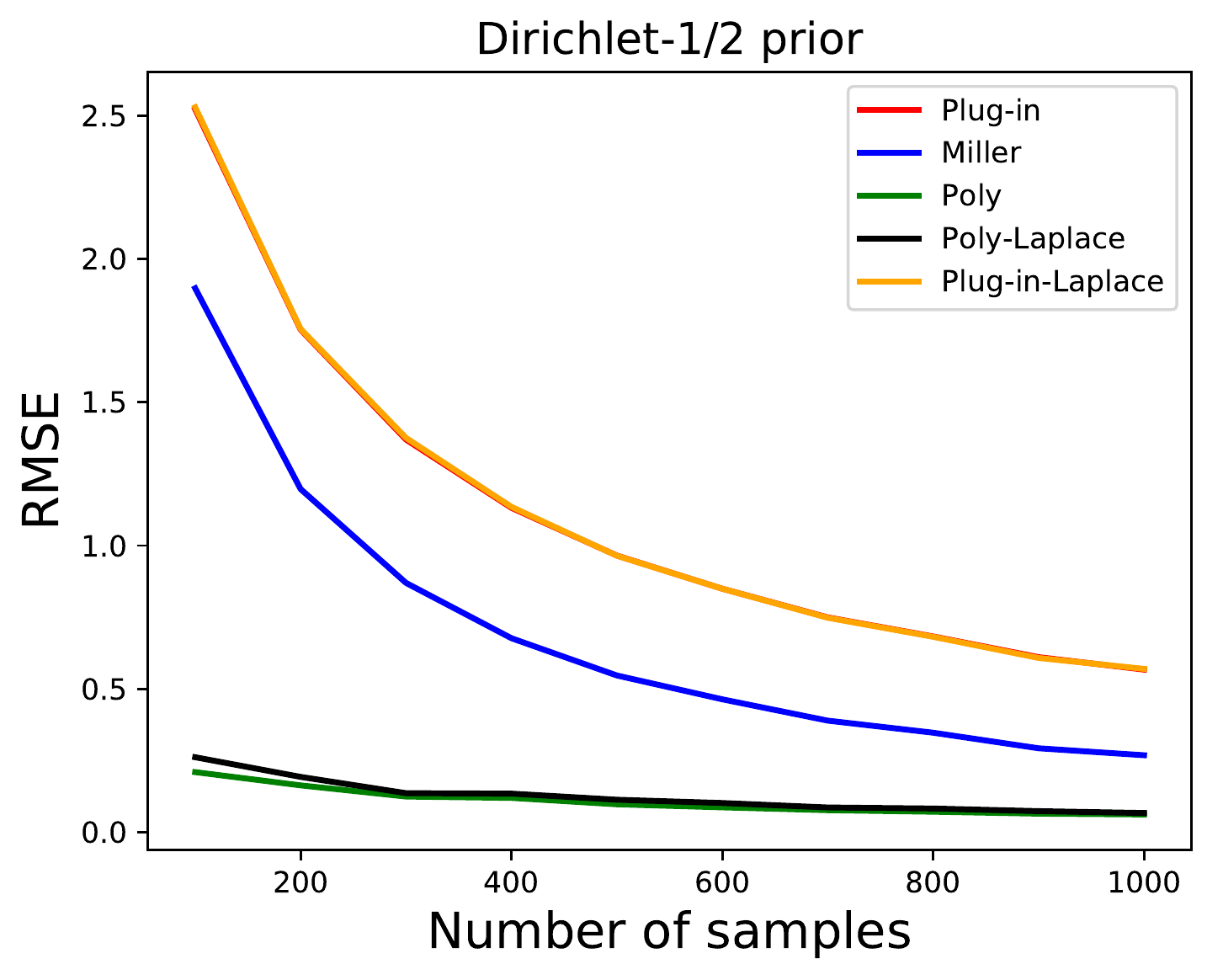}
\end{minipage}
}
\caption{Comparison of various estimators for the entropy, $k=1000$, $\eps =2$.} 
\label{fig:entropy-k1000-eps2}
\end{figure*}

\subsection{Support Coverage}
\label{sec:supp-exp-coverage}
We present three additional plots of our synthetic experimental results for support coverage estimation.
In particular, Figures~\ref{fig:coverage-k1000}, \ref{fig:coverage-k5000}, and~\ref{fig:coverage-k100000} show support coverage for $k$ = 1000, 5000, 100000.
\begin{figure*}
\centering
\subfigure[]{
\begin{minipage}[b]{0.3\textwidth}
\includegraphics[width=1\textwidth]{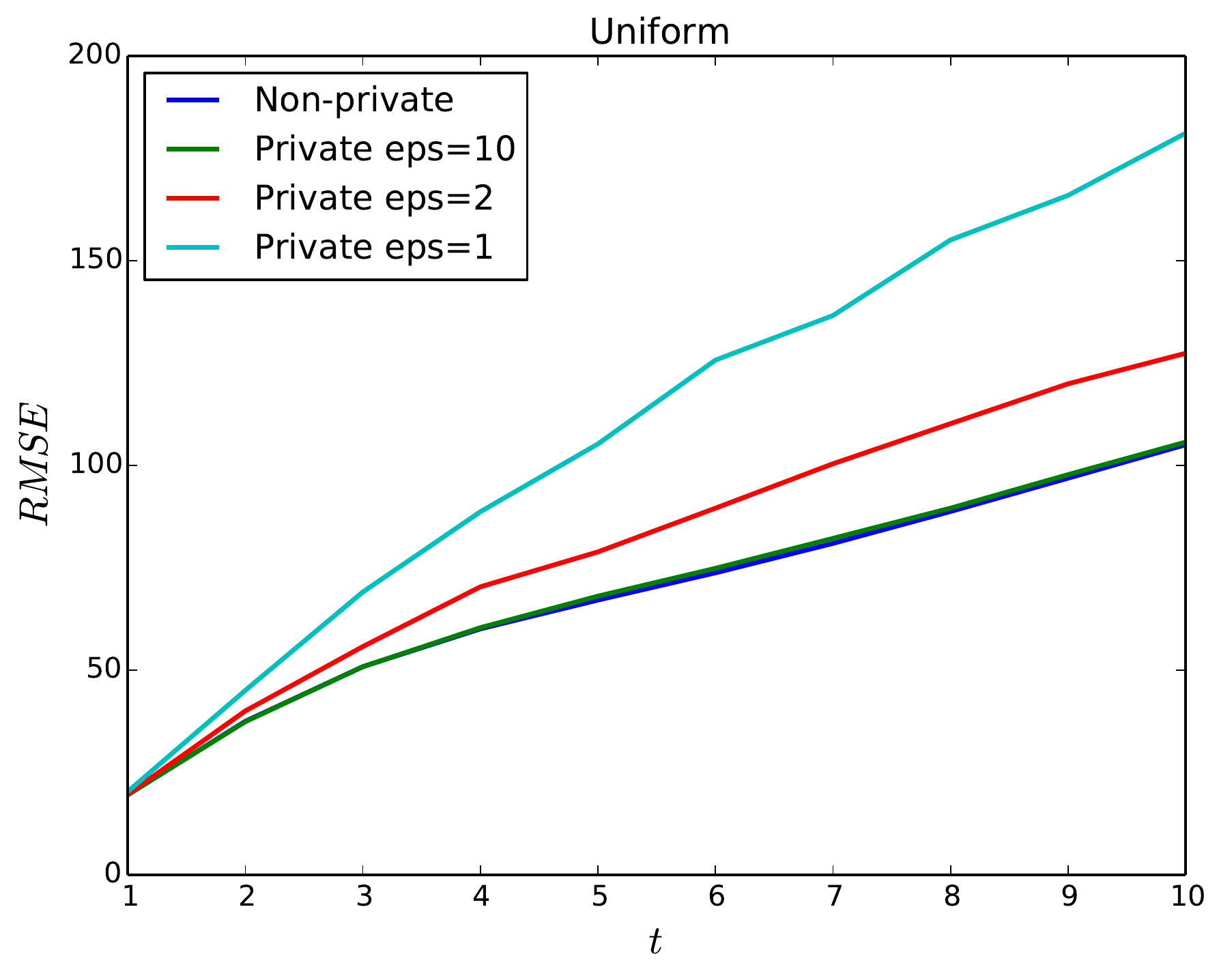}
\includegraphics[width=1\textwidth]{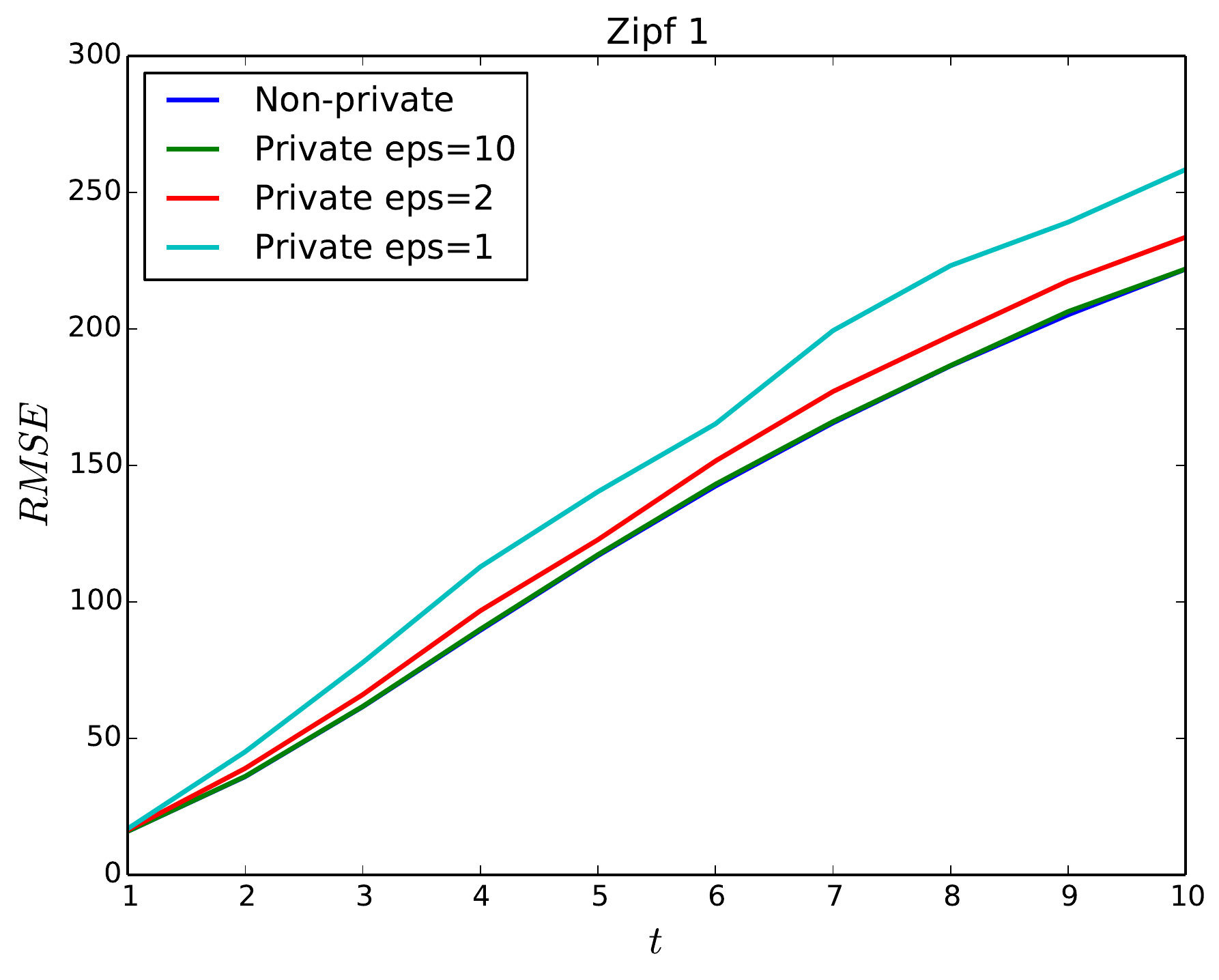}
\end{minipage}
}
\subfigure[]{
\begin{minipage}[b]{0.3\textwidth}
\includegraphics[width=1\textwidth]{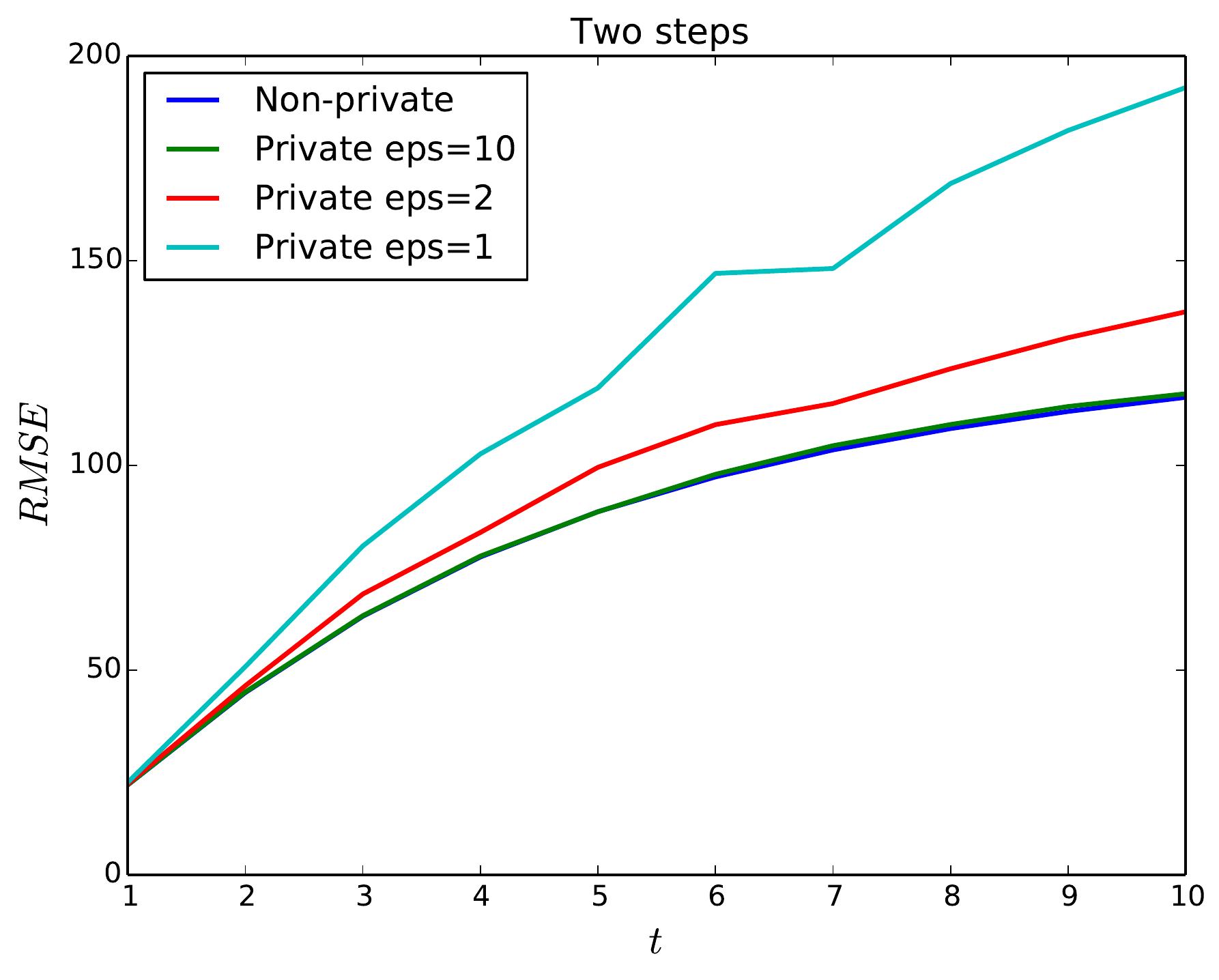}
\includegraphics[width=1\textwidth]{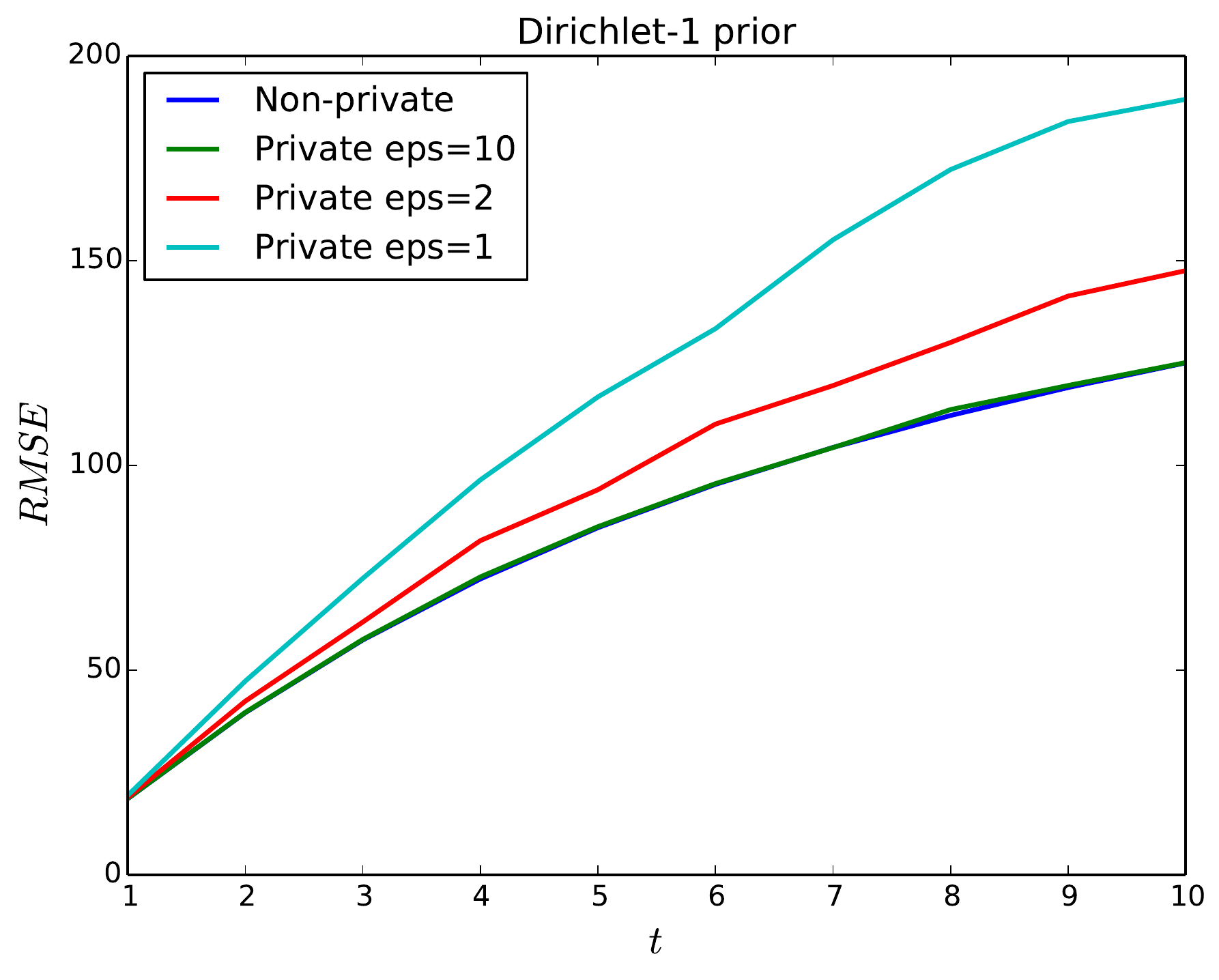}
\end{minipage}
}
\subfigure[]{
\begin{minipage}[b]{0.3\textwidth}
\includegraphics[width=1\textwidth]{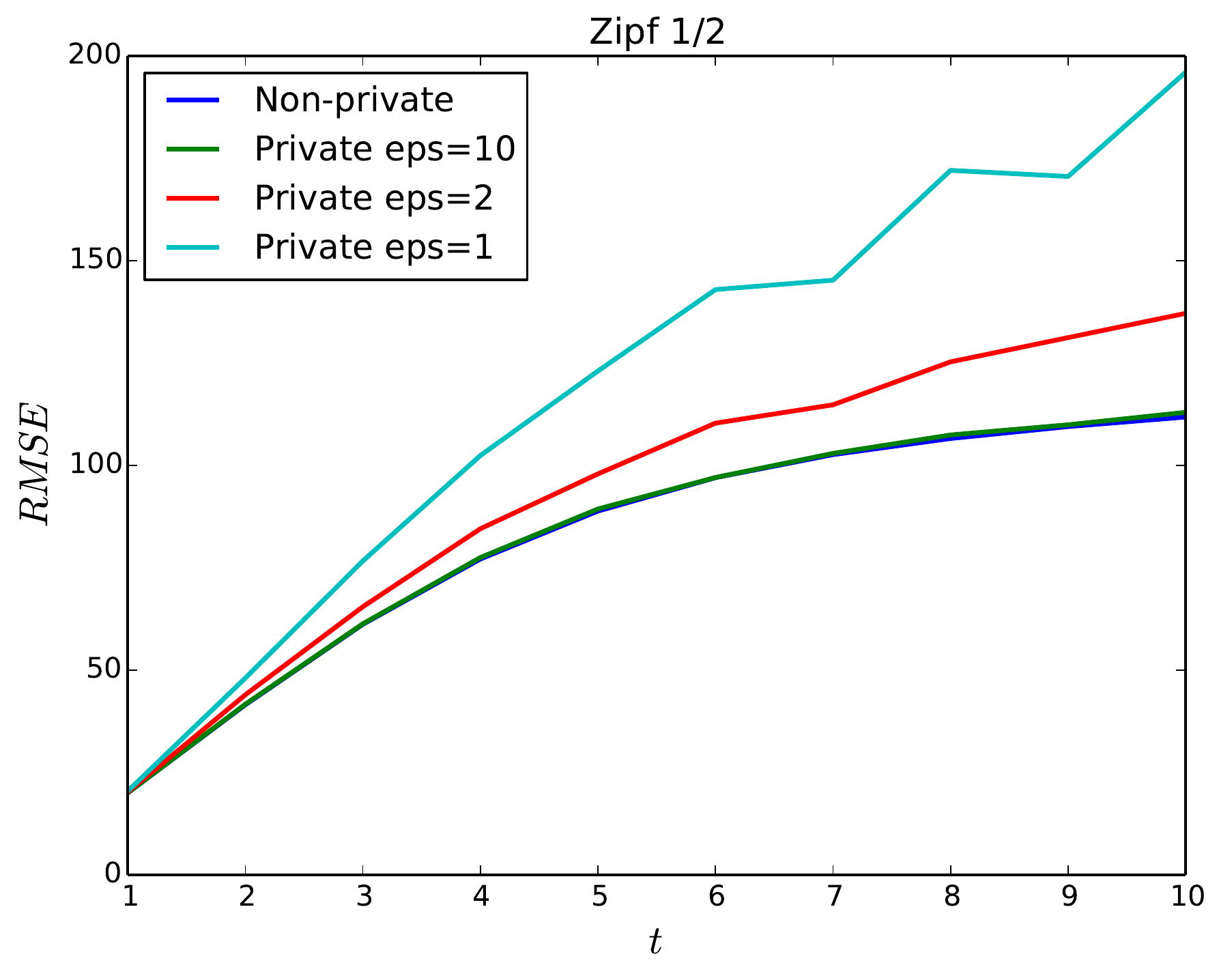}
\includegraphics[width=1\textwidth]{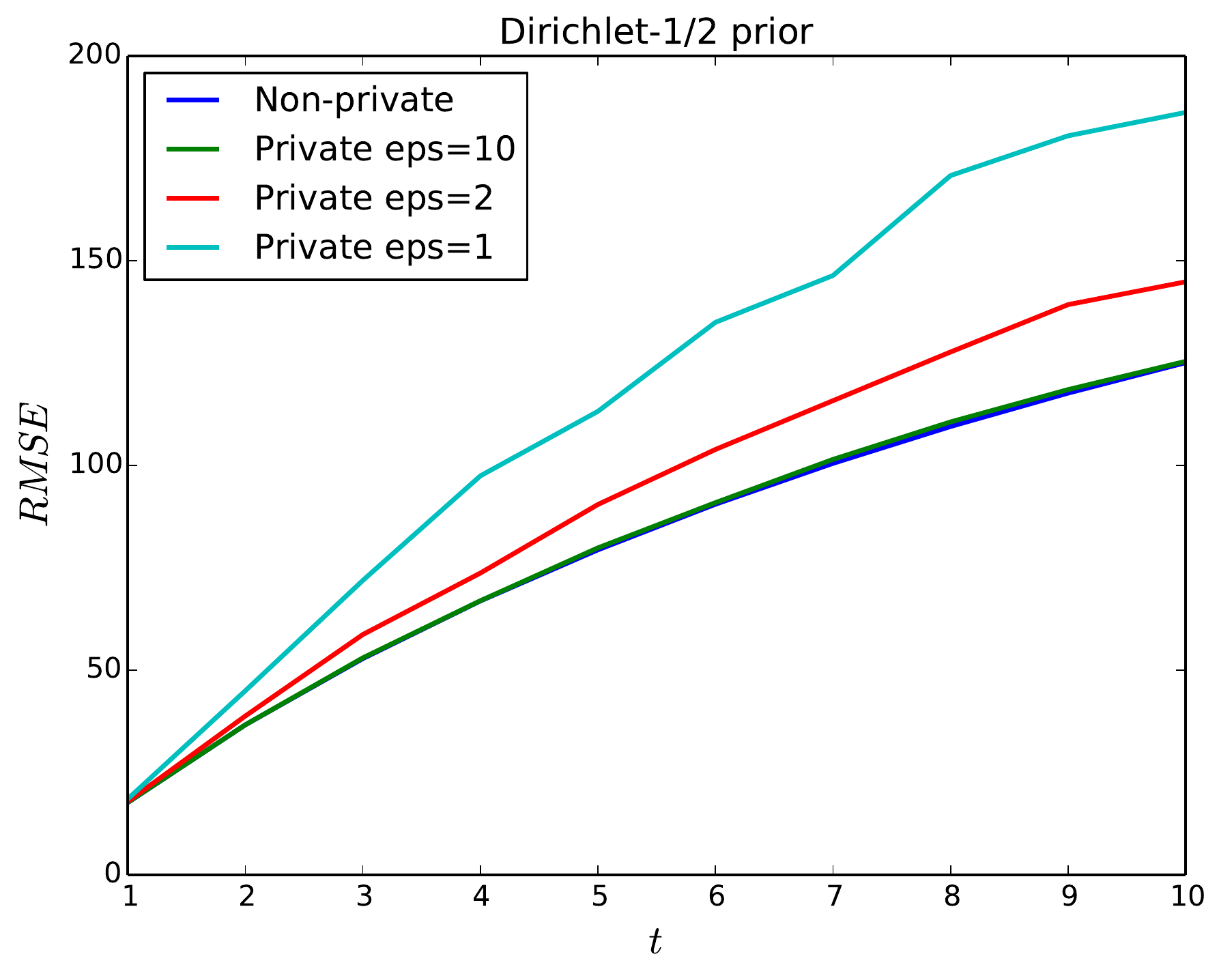}
\end{minipage}
}
\caption{Comparison between the private estimator with the non-private SGT when $k=1000$.} 
\label{fig:coverage-k1000}
\end{figure*}
\begin{figure*}
\centering
\subfigure[]{
\begin{minipage}[b]{0.3\textwidth}
\includegraphics[width=1\textwidth]{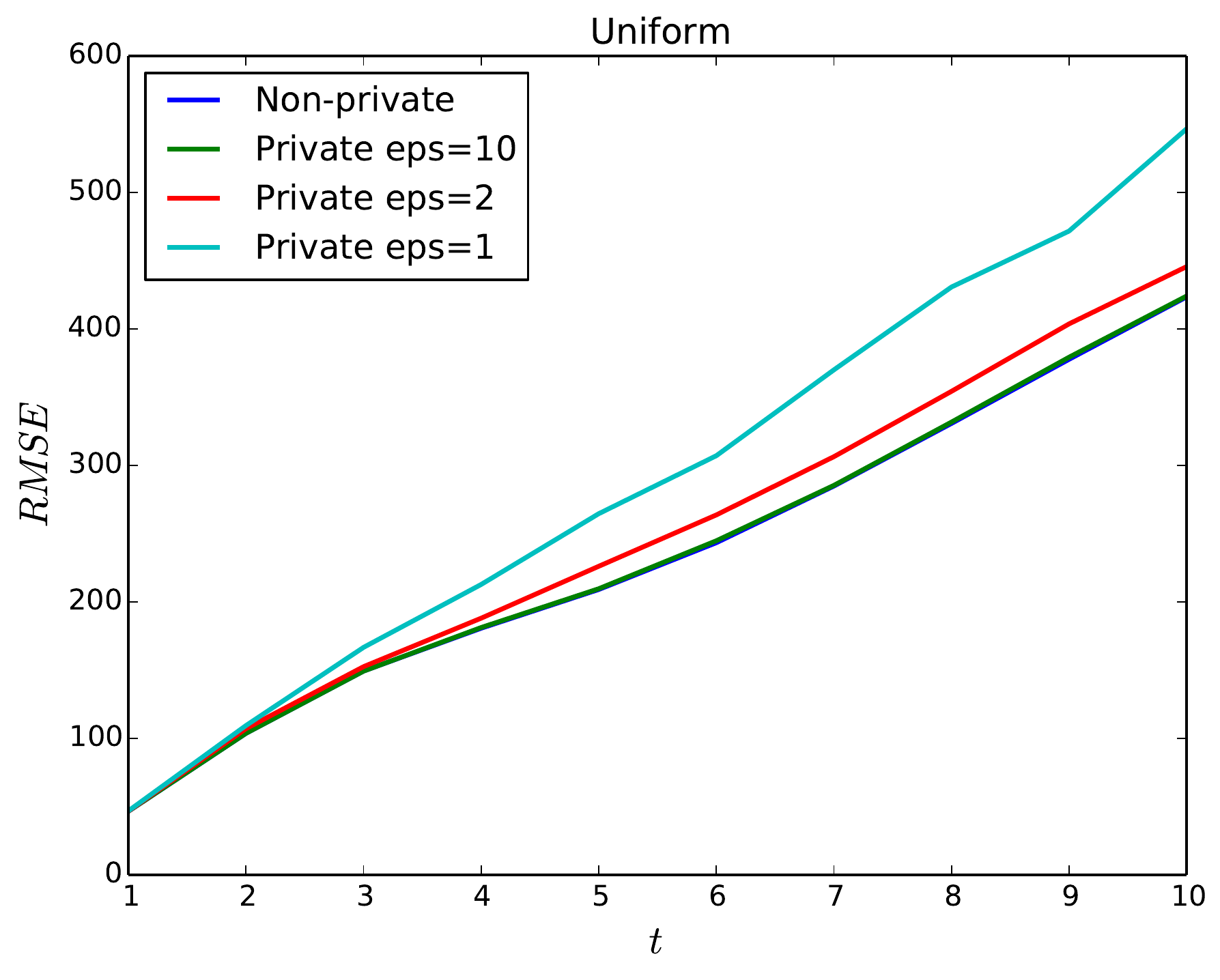}
\includegraphics[width=1\textwidth]{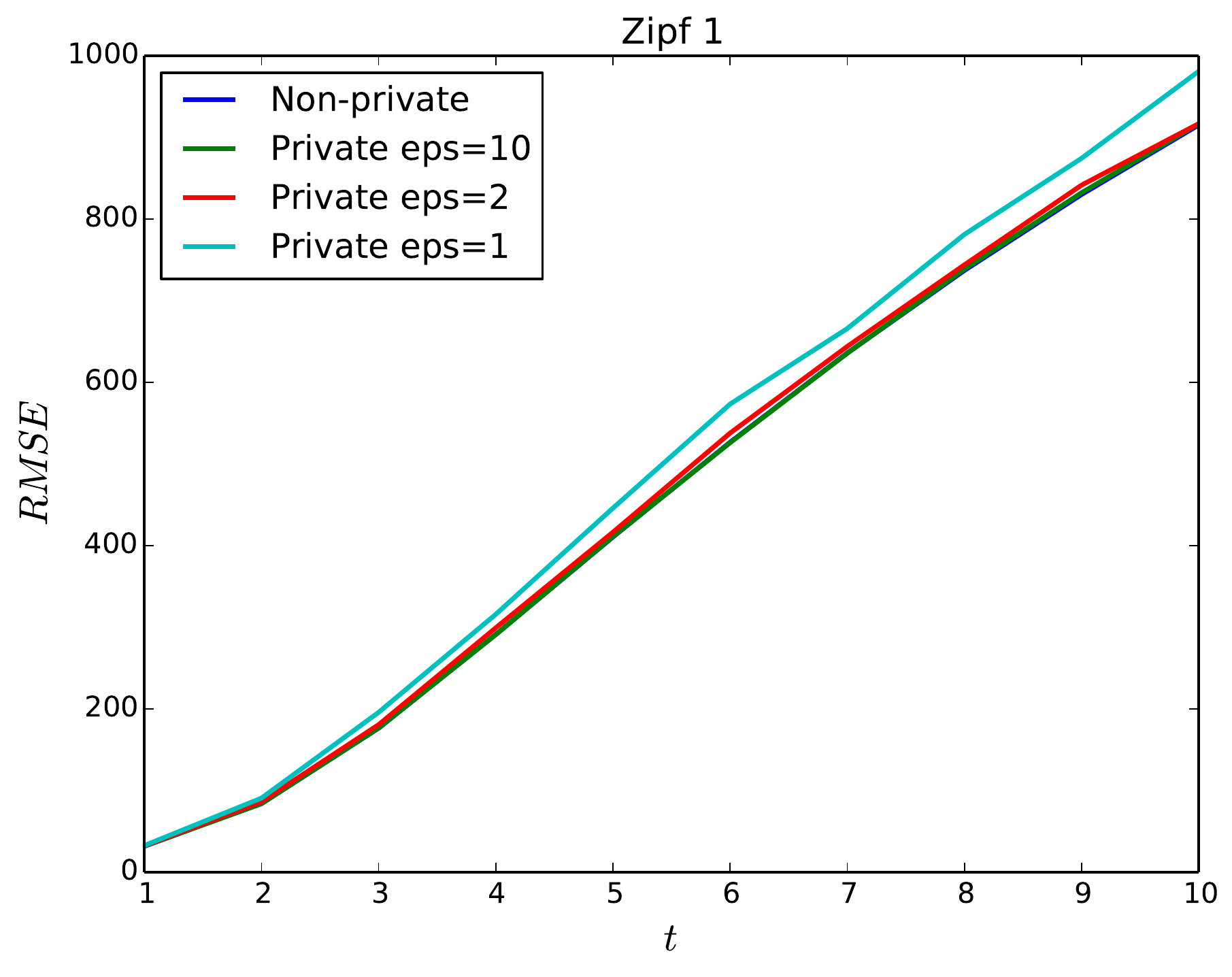}
\end{minipage}
}
\subfigure[]{
\begin{minipage}[b]{0.3\textwidth}
\includegraphics[width=1\textwidth]{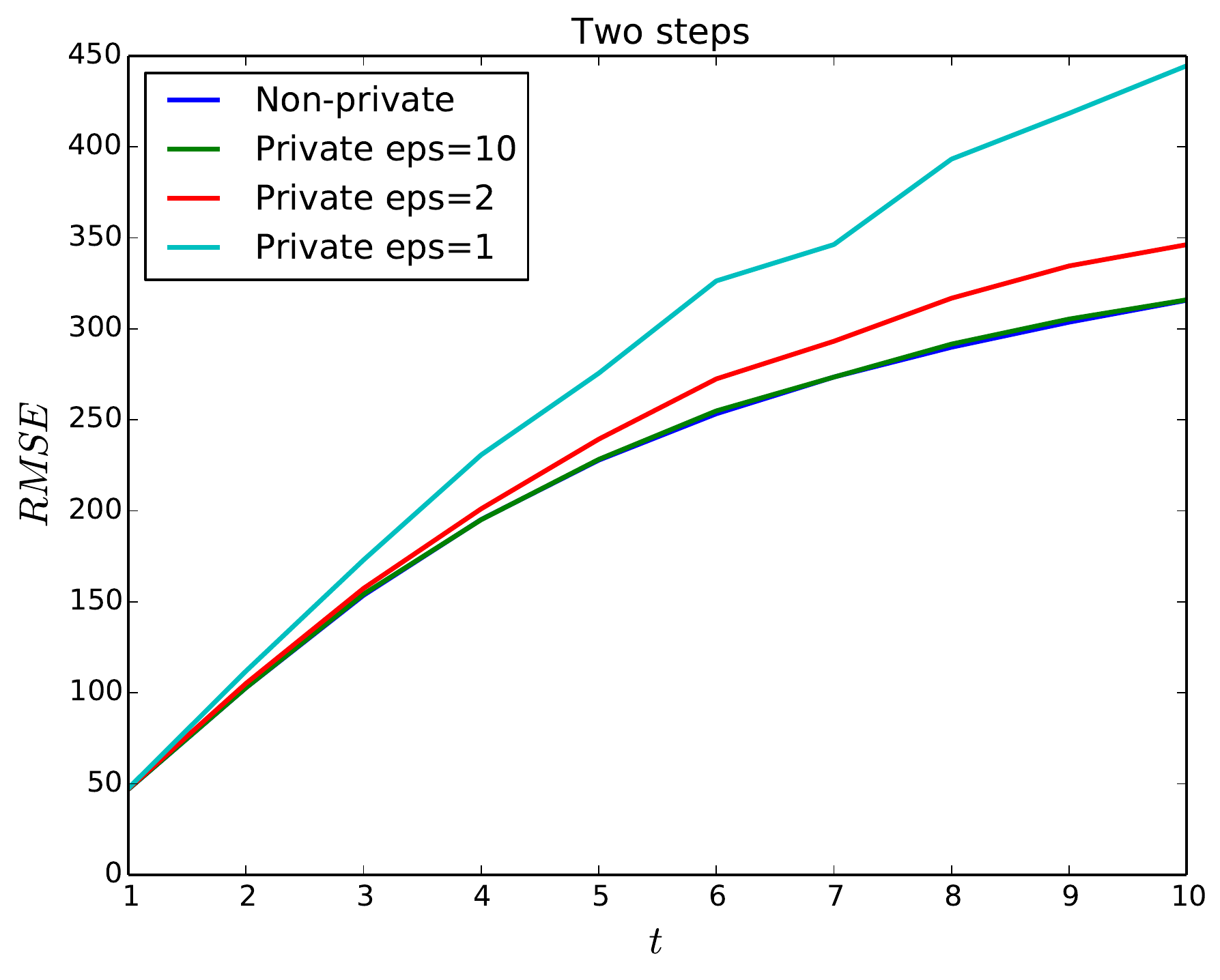}
\includegraphics[width=1\textwidth]{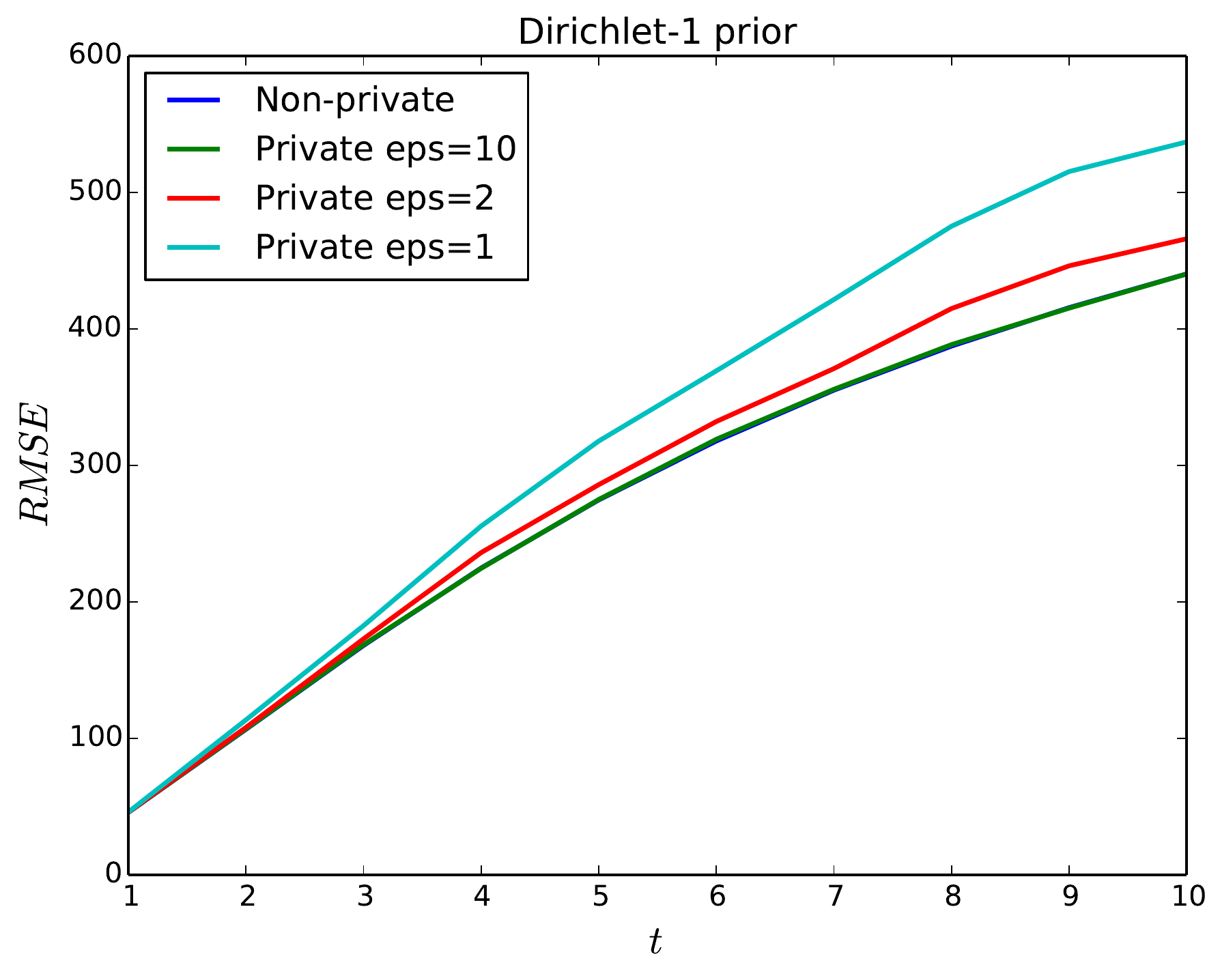}
\end{minipage}
}
\subfigure[]{
\begin{minipage}[b]{0.3\textwidth}
\includegraphics[width=1\textwidth]{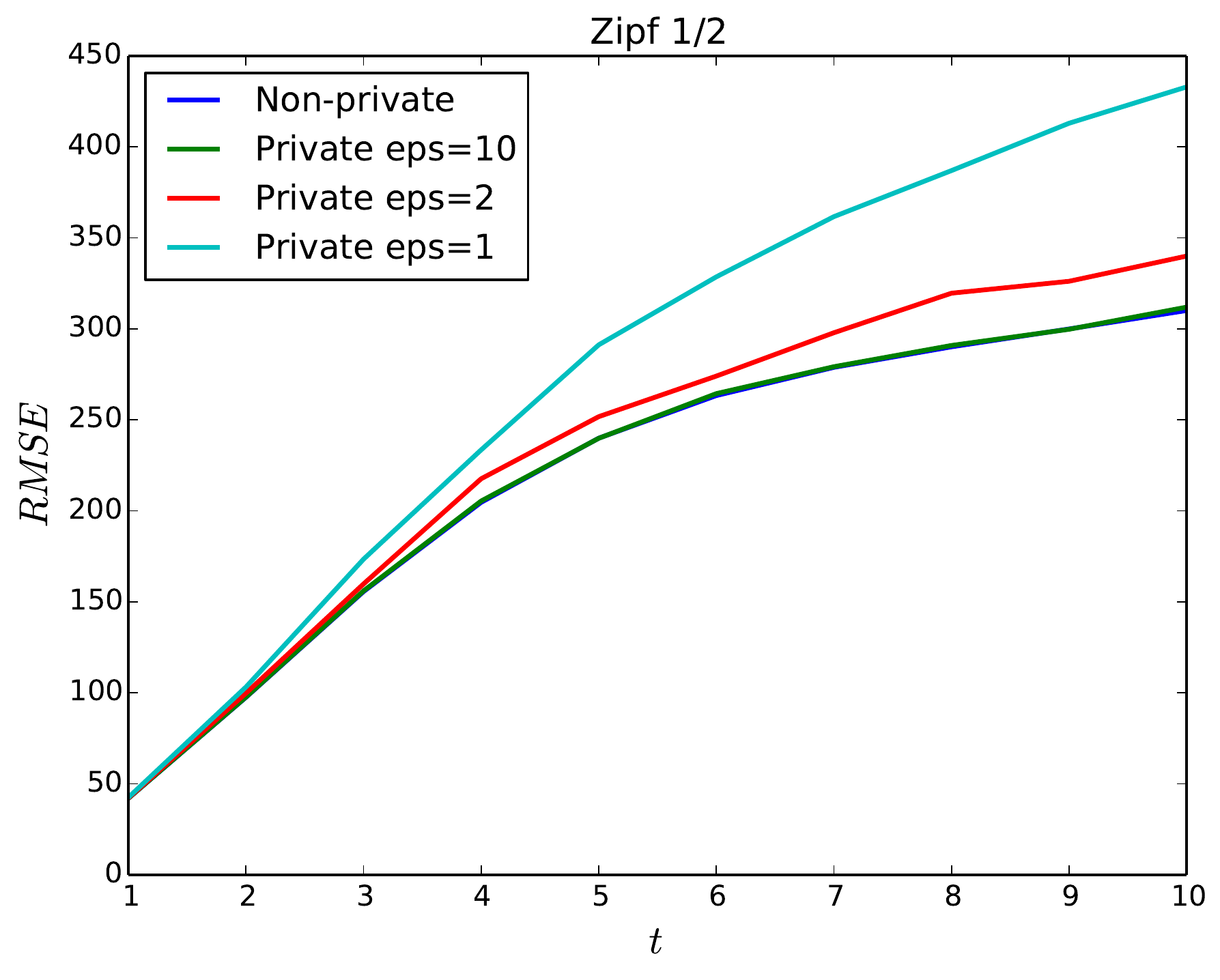}
\includegraphics[width=1\textwidth]{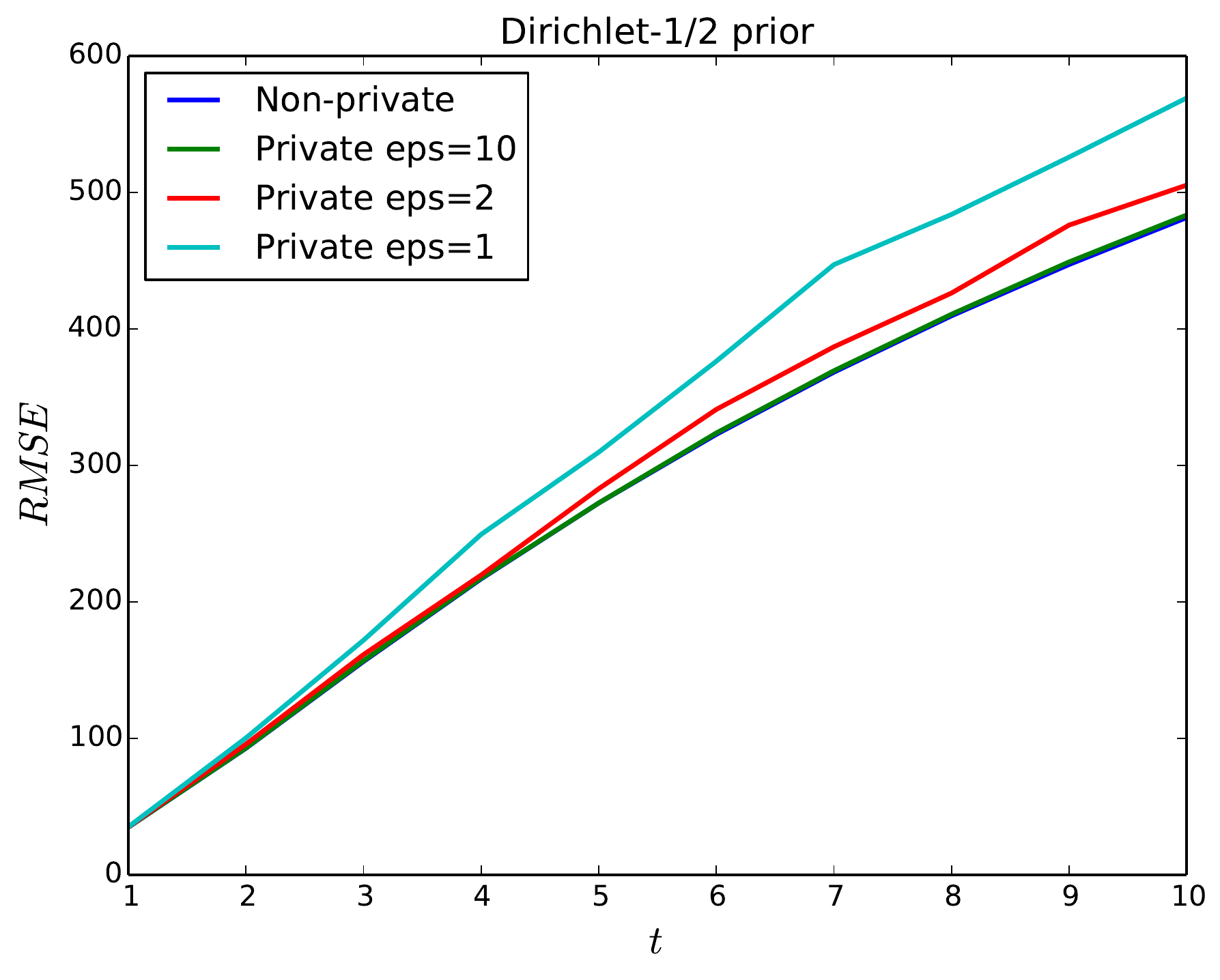}
\end{minipage}
}
\caption{Comparison between the private estimator with the non-private SGT when $k=5000$.} 
\label{fig:coverage-k5000}
\end{figure*}
\begin{figure*}
\centering
\subfigure[]{
\begin{minipage}[b]{0.3\textwidth}
\includegraphics[width=1\textwidth]{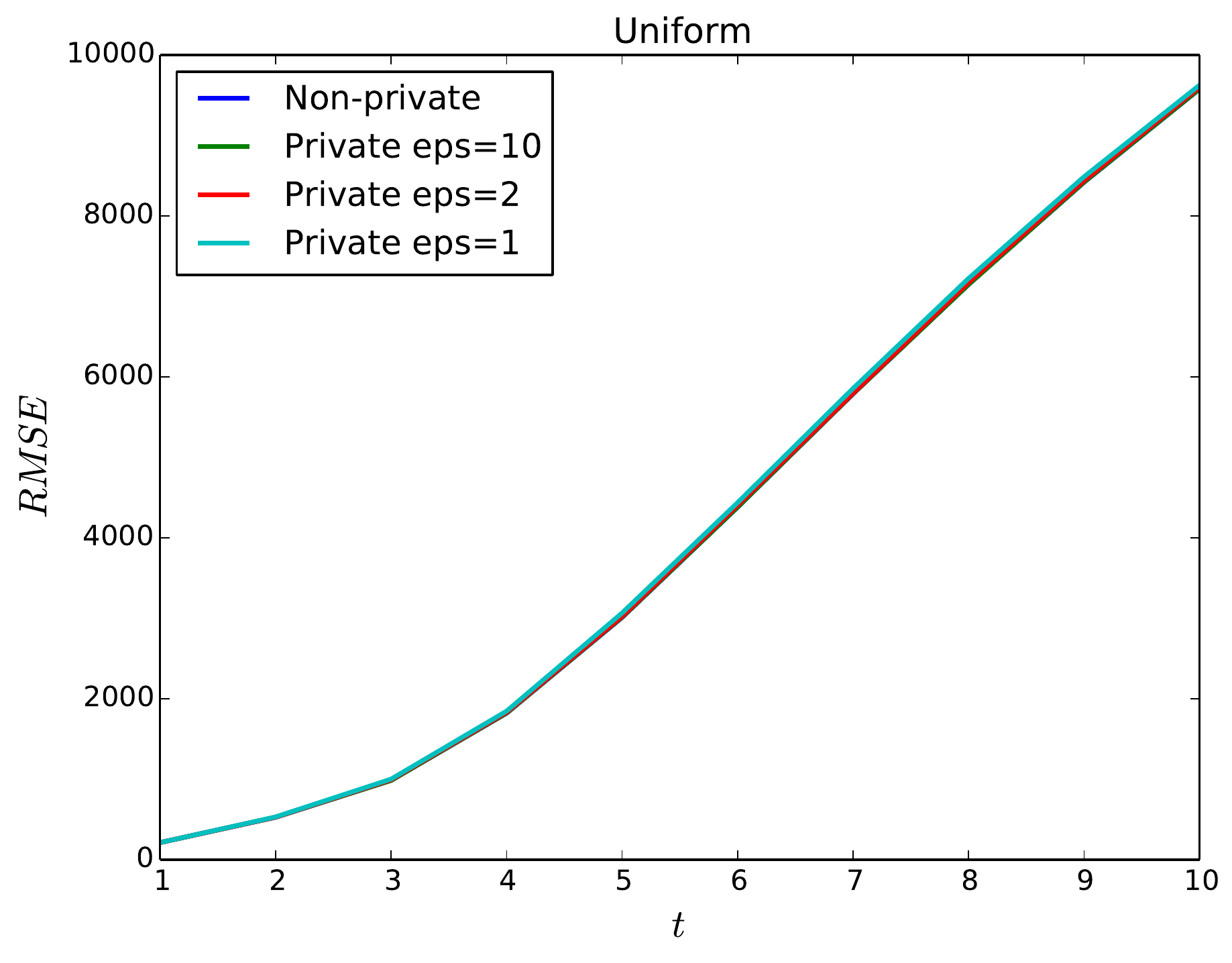}
\includegraphics[width=1\textwidth]{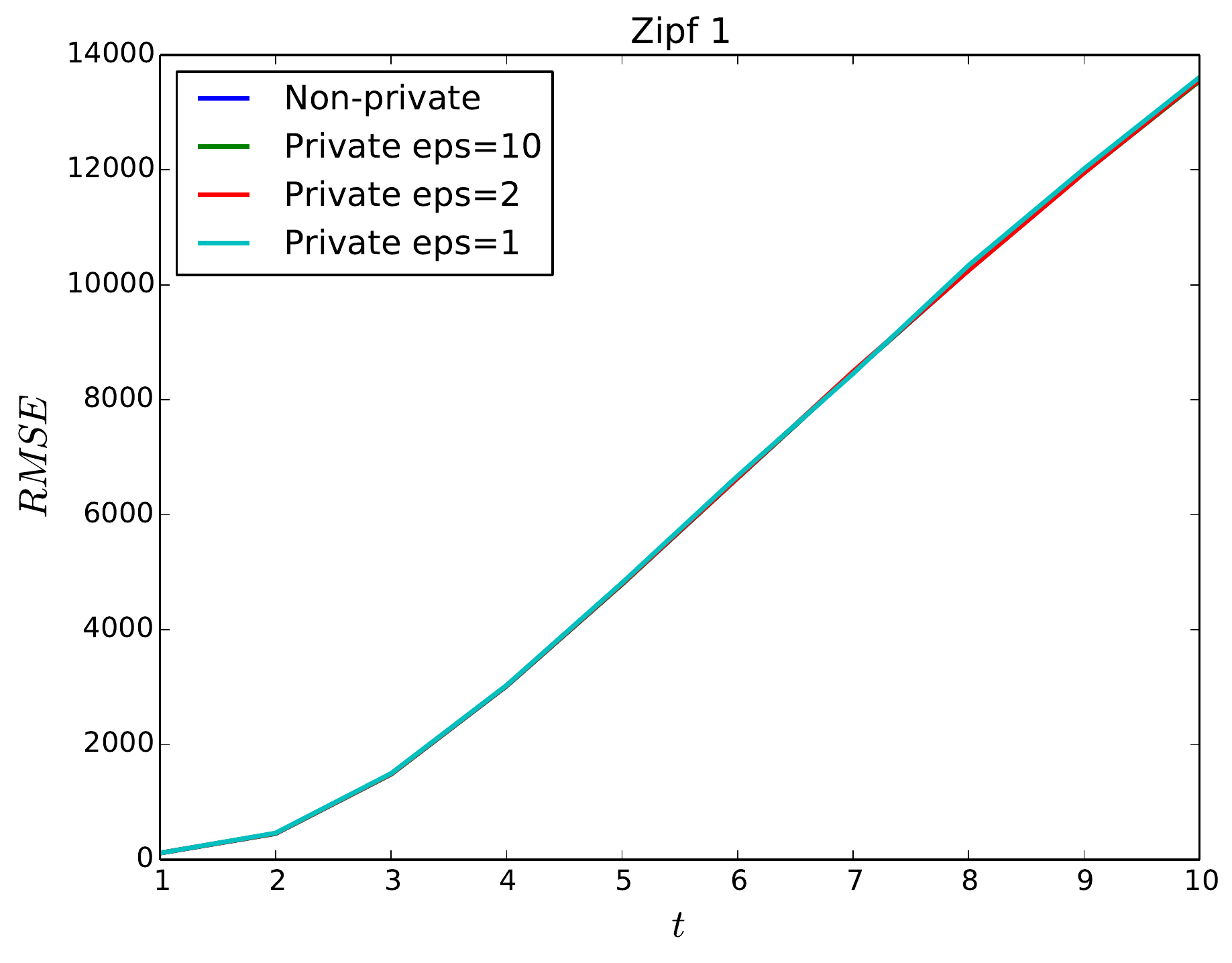}
\end{minipage}
}
\subfigure[]{
\begin{minipage}[b]{0.3\textwidth}
\includegraphics[width=1\textwidth]{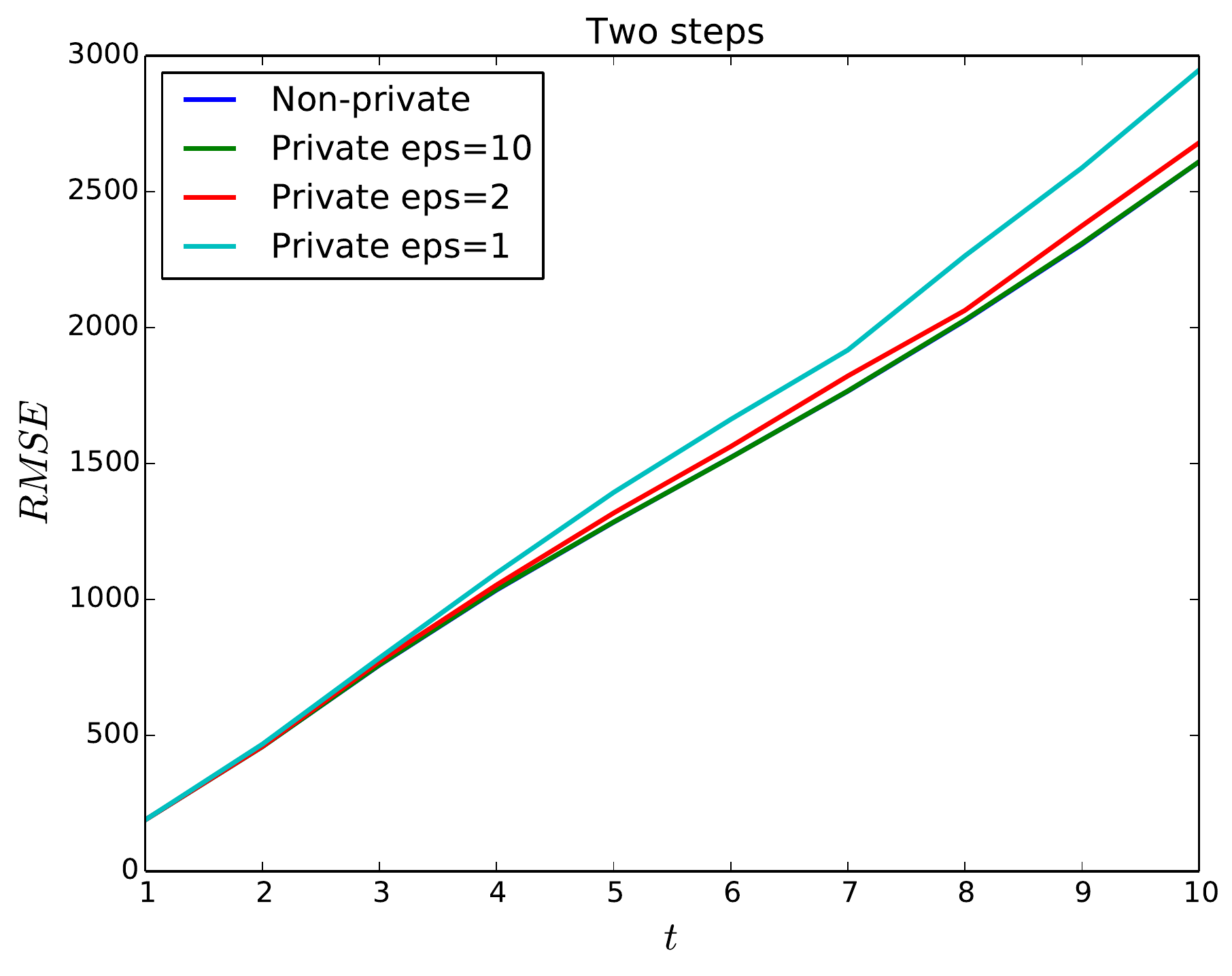}
\includegraphics[width=1\textwidth]{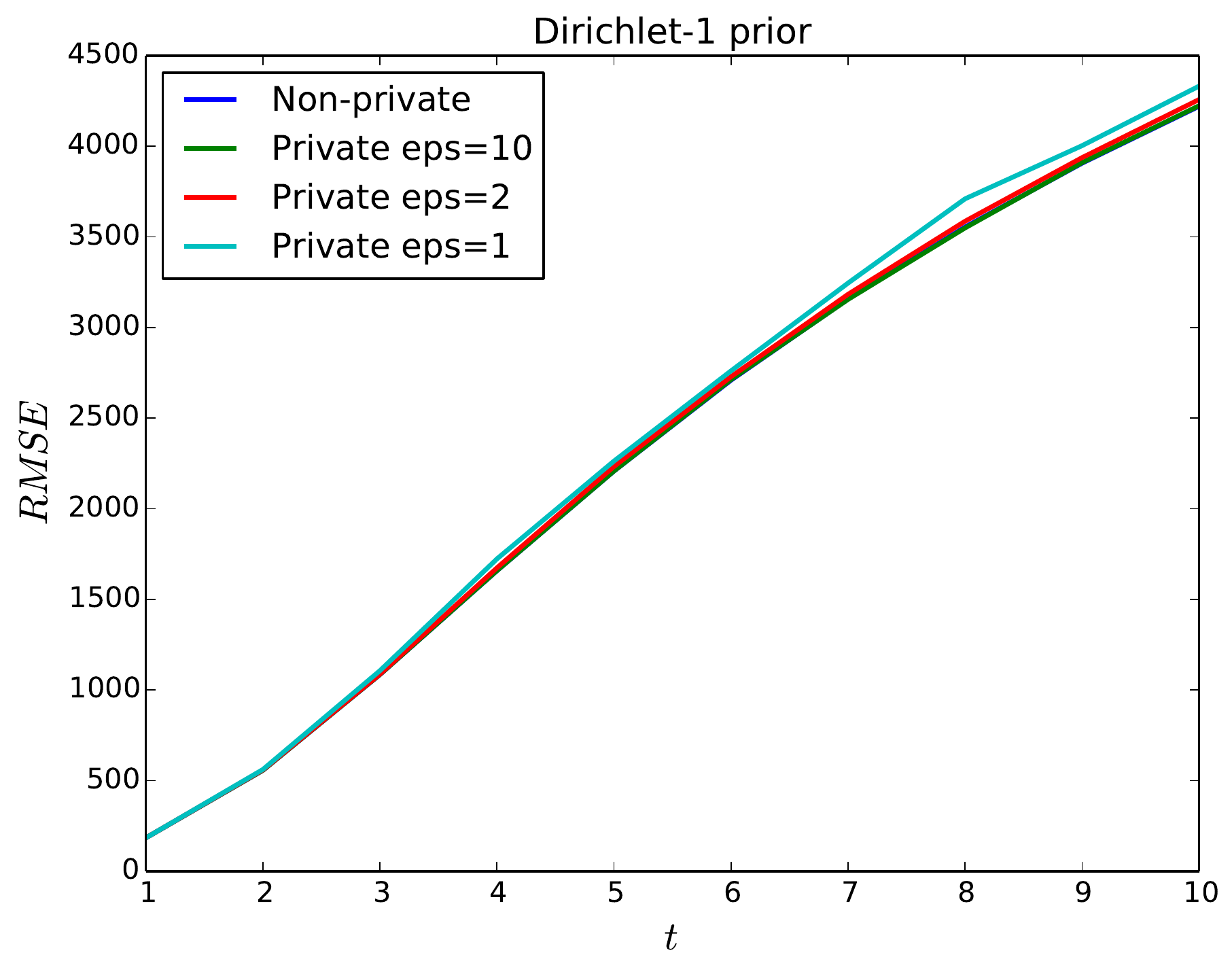}
\end{minipage}
}
\subfigure[]{
\begin{minipage}[b]{0.3\textwidth}
\includegraphics[width=1\textwidth]{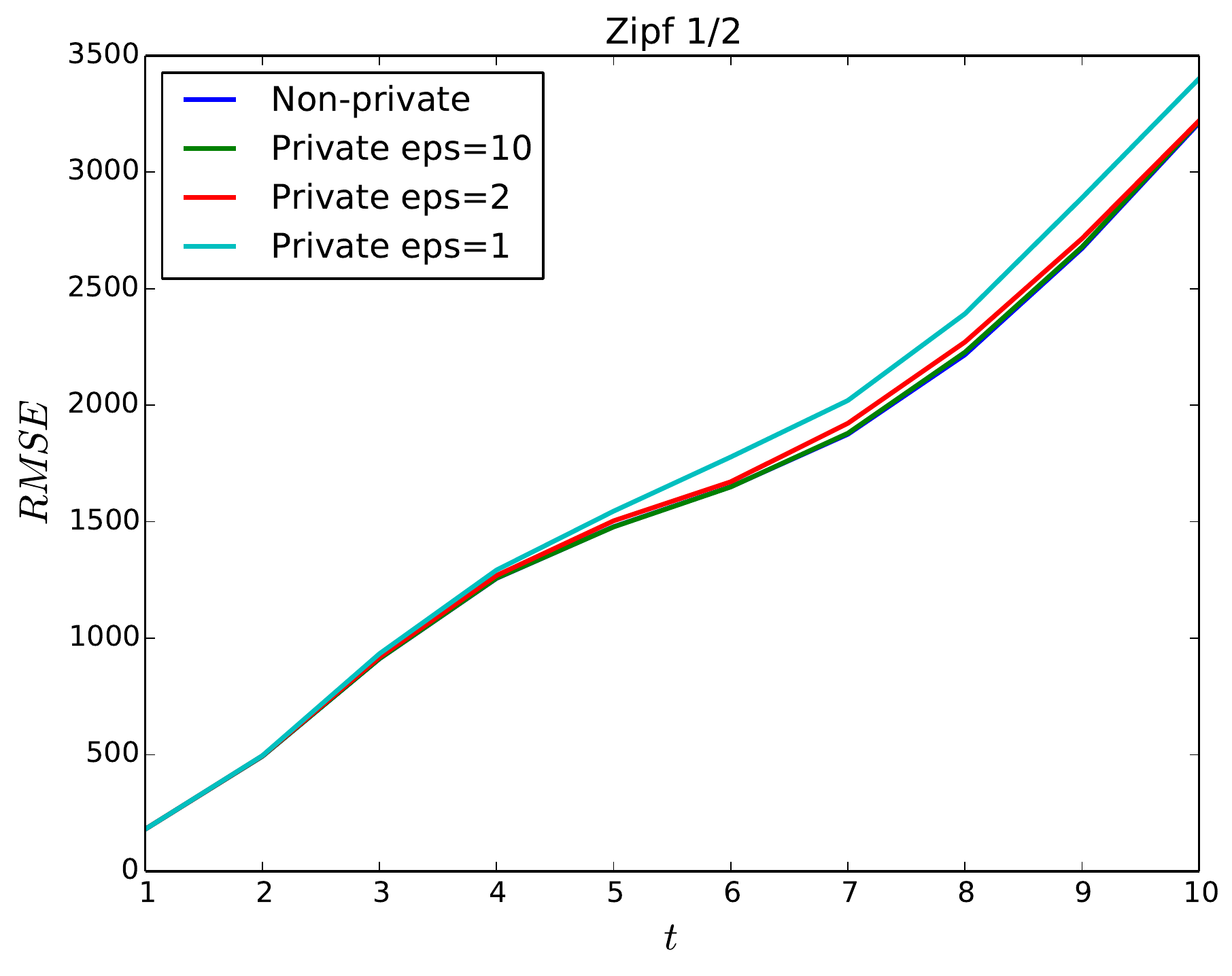}
\includegraphics[width=1\textwidth]{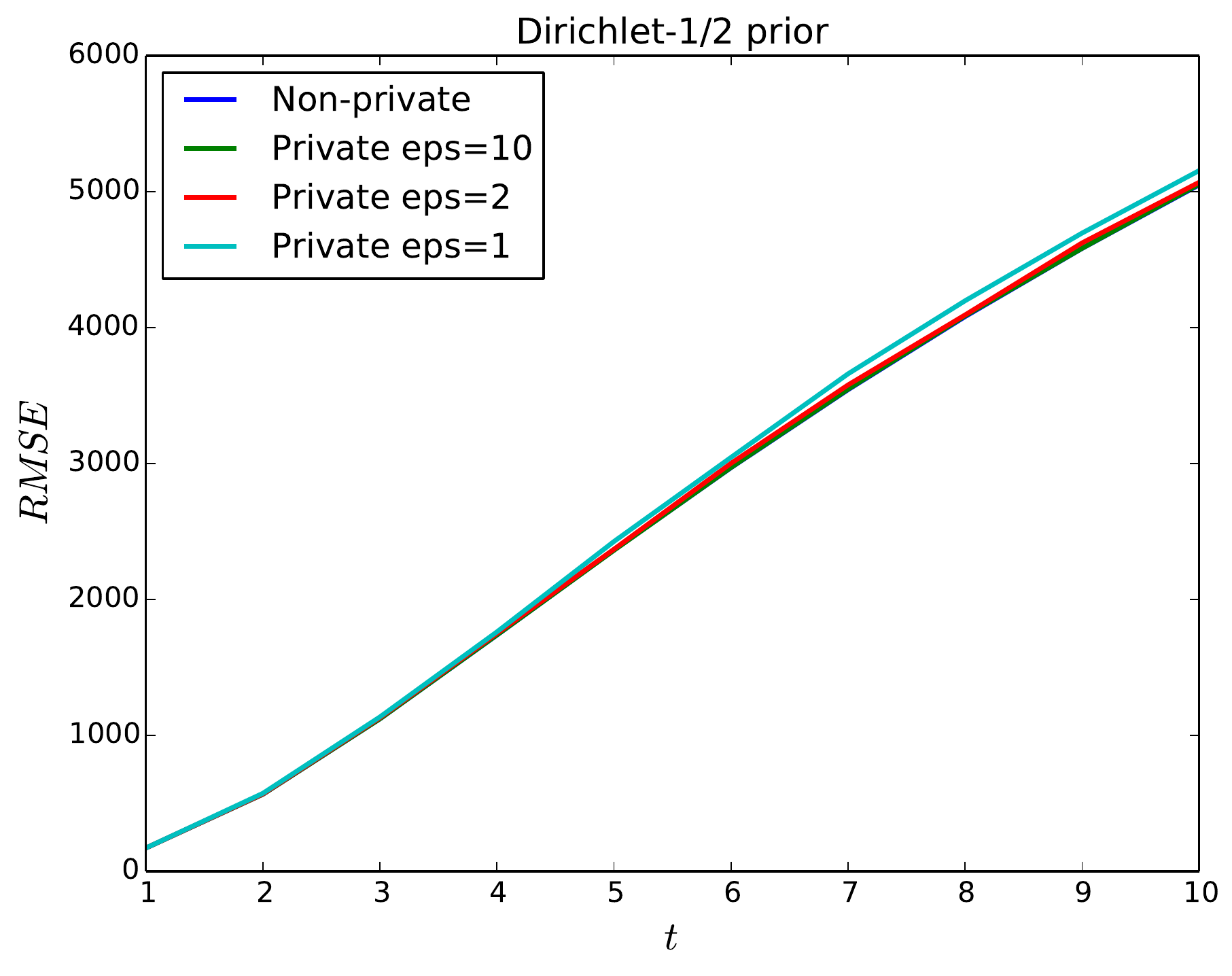}
\end{minipage}
}
\caption{Comparison between the private estimator with the non-private SGT when $k=100000$.} 
\label{fig:coverage-k100000}
\end{figure*}

\end{document}